\documentclass[pra,superscriptaddress,showpacs,onecolumn,twoside,11pt,notitlepage,a4paper,aps]{revtex4-1}

\usepackage{graphicx,epic,eepic,epsfig,amsmath,latexsym,amssymb,verbatim,color}
\usepackage{amsthm}
\newtheorem{definition}{Definition}[section]
\newtheorem{proposition}[definition]{Proposition}
\newtheorem{lemma}[definition]{Lemma}

\newtheorem{theorem}[definition]{Theorem}
\newtheorem{corollary}[definition]{Corollary}

\newtheorem{problem}[definition]{Problem}
\newtheorem{example}[definition]{Example}
\usepackage[activate={true,nocompatibility},final,tracking=true,kerning=true,spacing=true,factor=1100,stretch=10,shrink=10]{microtype}

\usepackage[hidelinks]{hyperref}

\begin{document}
\title{Combinatorial Entanglement}

\author{Joshua Lockhart}
\email{joshua.lockhart.14@ucl.ac.uk}
\affiliation{Department of Computer Science, University College London, WC1E 6BT, London, U.K.}

\author{Simone Severini}
\email{s.severini@ucl.ac.uk}
\affiliation{Department of Computer Science, University College London, WC1E 6BT, London, U.K.}
\affiliation{Institute of Natural Sciences, Shanghai Jiao Tong University, Shanghai 200240, China}

\begin{abstract}
We present new combinatorial objects, which we call \emph{grid-labelled graphs}, and show how these can be used to represent the quantum states arising in a scenario which we refer to as the \emph{faulty emitter scenario}: we have a machine designed to emit a particular quantum state on demand, but which can make an error and emit a different one. The device is able to produce a list of candidate states which can be used as a kind of debugging information for testing entanglement. By reformulating the Peres-Horodecki and matrix realignment criteria we are able to capture some characteristic features of entanglement: we construct new bound entangled states, and demonstrate the limitations of matrix realignment. We show how the notion of LOCC is related to a generalisation of the graph isomorphism problem. We give a simple proof that asymptotically almost surely, grid-labelled graphs associated to very sparse density matrices are entangled. We develop tools for enumerating grid-labelled graphs that satisfy the Peres-Horodecki criterion up to a fixed number of vertices, and propose various computational problems for these objects, whose complexity remains an open problem. The proposed mathematical framework also suggests new combinatorial and algebraic ways for describing the structure of graphs.
\end{abstract}

\maketitle

\tableofcontents

\setcounter{page}{1}

\section{Introduction}
\subsection{Preliminaries}
The central goal in the field of quantum information processing is to gain a deeper understanding of quantum mechanical systems as objects for encoding, transferring, securing, and manipulating information. In quantum mechanics, the standard mathematical tool used for combining physical systems is the well-studied notion of \emph{tensor product}. Many scenarios in quantum information processing require working with tensor product structures. As a direct consequence of the axiomatisation of quantum mechanics, the tensor product has acquired a fundamental role in describing the behaviour of various physical phenomena: \emph{quantum entanglement} is one of those. Even if entanglement is somehow a mathematical artifact, it quickly emerged as one of the defining characteristics of quantum mechanics itself. In quantum information processing, it is now clear that entanglement has ubiquitous applications, being a necessary ingredient in protocols for cryptography, a number of distributed tasks, and also computation \cite{Horodecki2009}. Indeed, we have evidence to the fact that algorithmic speedup promised by quantum computers seems to rely on the presence of a sufficiently large amount of entanglement \cite{JozsaLinden,Vidal}. However, recently, other quantities have been proposed as being responsible for quantum algorithmic speedup -- see contextuality \cite{magic}.

On a more firm ground, we can say that entanglement is responsible for the advantage exhibited by quantum strategies in a variety of multi-party games and interactive proof systems. Moreover, there is a whole set of information theoretic tasks in which entanglement provides a neat distinction between classical and quantum behaviour -- see, for example, zero-error quantum information theory \cite{DuanSeveriniWinter, BeigiShor, Leung}, superactivation of channels \cite{Cubitt,CubittSuper,Duan}, \emph{etc}. Finally, at the foundational level, quantum entanglement has been traditionally interpreted as an inherently quantum feature embodying non-classical correlations between the components of a physical system \cite{CHSH}. 

Despite the fact that entanglement manifests itself as a ``physical property'' of quantum systems, detecting its presence with the mathematical tools made available by the theory is a complex problem in general. The task of deciding if a quantum state is \emph{entangled} or, conversely, \emph{separable}, is in fact highly non-trivial: no necessary and sufficient criterion exists for detecting entanglement. As a matter of fact, the problem of testing separability of an arbitrary density matrix is NP-hard \cite{Gurvits}. Over the years a rich set of mathematical tools has been developed in relation to this problem. The toolbox includes hierarchies of positive semi-definite programs, results from operator theory and functional analysis, ideas from linear algebra and matrix theory, \emph{etc.} \cite{Ioannou, Horodecki2009}. While pure state entanglement is relatively well understood, the study of the mixed state case has revealed numerous subtleties. Additionally, there are substantial differences between the case involving only two subsystems, and the generic case. One of the difficulties lies in the emergence of properties specifically connected with the presence of multiple subsystems, and in the choice of an appropriate entanglement measure.

It turns out that in some cases, the performance of certain entanglement measures depends on the type of state under consideration. In other words, different entanglement measures lend themselves to the analysis of particular families of states. With this perspective in mind, it is useful to introduce new families of mixed states as a means for testing different measures and, more generally, to study the behaviour of multipartite entanglement.

\subsection{Results}
In this work, we focus on the problem of detecting entanglement present in mixed states, using a classical description of the state as input. We present new combinatorial objects, which we call \emph{grid-labelled graphs}, which are obtained by labelling simple graphs with the point coordinates of a given grid. We show how grid-labelled graphs can be used to represent the quantum states arising in a scenario which we refer to as the \emph{faulty emitter scenario}. Here we have a device designed to emit a particular quantum state on demand, but which occasionally in error emits a different one. While the device does not know exactly which state it just emitted, it is able to produce a short-list of candidate states. The goal is then to use these candidate states to deduce if entanglement is present -- see Section \ref{subsection:faultyemittersandcombinatorialentanglement}.

The main focus of the paper is the study of grid-labelled graphs, and how their structural properties allow us to reason about separability in the underlying quantum state. This framework is not fully new, since it is closely related to the study of combinatorial Laplacians as quantum states initiated by Braunstein \emph{et al.} \cite{Braunstein2006b}. On one hand we reconsider their ideas from the ground up, but our key innovation is the introduction of a natural vertex labelling induced by the tensor product. This labelling provides a more physical interpretation of the states. Moreover, the labelling allows us to employ a wealth of mathematical tools of a combinatorial flavour. In particular, we shall explore how two entanglement criteria, the Peres-Horodecki criterion and the matrix realignment criterion, manifest themselves in terms of graph structure and related combinatorics. 

We use this novel perspective to explore the interplay between these entanglement criteria. We are able to use it to generate an infinite family of entangled states that are not detected by the matrix realignment criterion, as well as two entangled states that are not detected by the Peres-Horodecki criterion. The latter states are said to be \emph{bound entangled} \cite{Horodecki2009}. Let us recall that there is no general algorithm to decide whether a given state is bound entangled.

The states arising in the faulty emitter scenario make up a discrete subset of all possible density matrices. Despite the apparent simplicity of the states, it seems that there is no method for testing separability that works in all cases. This fact suggests that these states are an interesting playground in which to explore the interplay between entanglement and computational complexity. A proof of NP-hardness of testing separability of these states is likely to rely on graph theoretic machinery, and could possibly be much simpler than the well known proof of Gurvits \cite{Gurvits}. Furthermore, such a hardness result would be interesting in its own right, as it would be an explicit demonstration of a family of density matrices for which separability is hard to test.

In our attempts to understand the entanglement properties of grid-labelled graphs we need to introduce a number of new mathematical concepts. While these are described in detail in the subsequent sections of this paper, in the remainder of this introduction we will attempt to give a taste of some of them. As shown by Braunstein \emph{et al.} \cite{Braunstein2006}, if the density matrix under consideration is a combinatorial Laplacian matrix then the partial transpose corresponds to an operation on the edges of the corresponding graph. Testing whether a density matrix is entangled by means of the Peres-Horodecki criterion then amounts to observing the effect of this operation on the edges of the graph. The corresponding density matrix is entangled if the vertex degrees are modified by the partial transpose. Naturally, our canonical labelling of the graph means that this \emph{degree criterion} can still be utilised in our grid-labelled graph framework. In Section \ref{section:separabilityin3x3}, we are interested in characterising grid-labelled graphs that satisfy the criterion. In order to do so, we introduce \emph{edge contribution} matrices and tables, as tools to gain insight into the effect of the partial transpose operation. Such notions allow us to gain a ``spatial'' view of the contribution of each edge to the vertices on the grid before and after the partial transpose. In Appendix \ref{appendix:contributiontables} we isolate objects and study them in their own right, formulating a single player board game. We leave the classification of the computational complexity of this entanglement inspired game as an open problem.

In the process of classifying grid-labelled graphs that satisfy the degree criterion, we come across two particular instances that have bound entangled quantum states. An advantage of using our grid-labelled graph framework is that it prompts us to see patterns in the structure of the quantum states we are considering. The two bound entangled states in question have grid-labelled graphs with a particular edge pattern. It is easy to see how this pattern can be scaled up to generate similar quantum states in higher dimensions. We leave it as an open question to determine whether or not these higher dimensional states with the particular edge pattern are bound entangled. Our initial calculations suggest that such a cross-hatch structure is indeed a signature of bound entanglement \cite{otfried}. This is encouraging evidence in the favour of using ``structural'' arguments to characterise quantum states like we do.

In quantum mechanics, if a multi-partite quantum state can be transformed into another by means of a sequence of \emph{local operations}, that is, a sequence of individual manipulations performed on each part of the system, then we say that the states are \emph{local operations and classical communication equivalent} (for short, LOCC equivalent). Naturally, such local manipulations can not be used to increase the amount of entanglement in the system. In our discussion, we are able to meaningfully express a subset of the LOCC operations as manipulations of the edges of the grid-labelled graph. Concretely, we say that two grid-labelled graphs are \emph{locally isomorphic} if they can be obtained from one another by permuting the vertices in the rows and columns of the associated grid. This idea expresses a special case of the graph isomorphism problem. In Section \ref{subsection:localisomorphisms}, not surprisingly, we prove that pairs of locally isomorphic graphs have the same entanglement properties. Of course, acting with a generic permutation on a grid-labelled graph could lead to a state that has completely different entanglement properties. As such, isomorphisms in full generality could be interpreted as including the entangling operations, while local isomorphisms correspond to LOCC. This is another example of a mathematical concept that we feel is interesting in its own right, and could potentially lead to new combinatorial questions and ideas. For example, in Section \ref{subsection:localisomorphisms} we show that testing local isomorphism of grid-labelled graphs is at least as hard as the graph isomorphism problem. 

Finally, in Section \ref{section:thematrixrealignmentcriterion}, we introduce two matrices that encode the structure of a graph, the \emph{degree structure} and \emph{adjacency structure} matrices. We show that, in some cases, the entanglement present in the quantum state of a grid-labelled graph can be detected by calculating the eigenvalues of these matrices. This suggests that properties of grid-labelled graphs are also captured by methods from algebraic graph theory, beyond a direct combinatorial analysis. Throughout the paper, we state a number of computational problems, whose complexity is left open. In particular, the problem $(k,a,b)$-\textsc{GraphSeparability} is concerned with testing separability of grid-labelled graphs with a particular number of vertices and edges. Additionally, \textsc{SubsetPPT} and \textsc{SubgraphDC} deal with constructing examples of quantum states that satisfy the Peres-Horodecki criterion and come from our analysis of the entanglement-inspired single player game discussed in Appendix \ref{appendix:contributiontables}. While intuitively there is a connection between \textsc{SubsetPPT} and the well known problem \textsc{SubsetSum}, there is additional structure to \textsc{SubsetPPT} which does not indicate an immediate reduction.

\subsection{Previous Works}
The combinatorial Laplacian matrix of a graph has been treated as a density matrix in \cite{Braunstein2006b}. In that work, it was shown that for such density matrices, the Peres-Horodecki entanglement criterion \cite{Peres, HorodeckiSep} can be expressed in terms of the structure of the corresponding graph via a criterion referred to as the \emph{degree criterion}. This idea is further explored in \cite{Braunstein2006}, where it was conjectured that this criterion is necessary and sufficient for testing separability of such graph Laplacian states. Wu \cite{Wu2006} proved that the degree criterion is necessary and sufficient for graph Laplacian states in $\mathbb{C}^2\otimes\mathbb{C}^q$ for any $q$. Shortly afterwards however, Hildebrand \emph{et al.} \cite{Hildebrand+2008} provided a counterexample to the conjecture, demonstrating a graph Laplacian state in $\mathbb{C}^3\otimes \mathbb{C}^3$ that satisfies the degree criterion, but is entangled. Variations and generalisations on the work of \cite{Braunstein2006b} has been carried out, mostly focusing on weighted graphs and multipartite correlations \cite{Hassan2007,Hassan2008,Hassan2007b,Hui2013,Adhikari2012,Li2015,Dutta2016,Xie2013}. In this work we focus on simple graphs and bipartite entanglement.

\subsection{Synopsis}
The structure of the paper is as follows. In Section \ref{section:preliminaries}, we give a brief overview of the mathematical tools involved in the standard formulation of quantum mechanics for the benefit of the reader not familiar with this topic. In Section \ref{subsection:faultyemittersandcombinatorialentanglement} we describe the faulty emitter scenario and define the main decision problem we study in the work, $(k,a,b)$-\textsc{EdgeSeparability}, before defining grid-labelled graphs in Section \ref{subsection:gridlabelledgraphs}. In Section \ref{section:thedegreecriterion} we show how the degree criterion of \cite{Braunstein2006b} can be interpreted in the context of our grid-labelled graph framework. In Section \ref{section:thedegreecriterion} we also introduce the concept of separable decompositions of a grid-labelled graph. In Section \ref{subsection:decompositions}, we demonstrate the existence of a family of quantum states for which the Peres-Horodecki criterion is necessary and sufficient, corresponding to what we refer to as \emph{stratified graphs}. In Section \ref{section:separabilityin3x3} we classify all $3\times 3$ grid-labelled graphs that satisfy the degree criterion, and prove that asymptotically almost surely, random grid-labelled graphs with a small number of edges are entangled. Finally, in Section \ref{section:thematrixrealignmentcriterion} we study the matrix realignment criterion of Chen and Wu \cite{MR} applied to grid-labelled graphs, and show that it has a combinatorial interpretation. We use it to construct an infinite family of entangled quantum states that the matrix realignment criterion does not detect.

\section{Preliminaries}
\label{section:preliminaries}
	\subsection{States and entanglement}
	\label{subsection:statesandentanglement}
In the standard mathematical axiomatization of quantum mechanics,
a \emph{ket} $|\psi \rangle $
is defined as a column unit vector in some Hilbert space. A \emph{bra} $%
\langle \psi |$ is the functional sending $|\varphi \rangle $ to the inner
product $\langle \psi |\varphi \rangle $. The state of a quantum system is by convention represented by a ket, referred to as the system's \emph{state vector}, which belongs to a Hilbert space referred to as the system's \emph{state space}.

Let $\mathcal{H}_A\cong\mathbb{C}^m$
 and $\mathcal{H}_B\cong\mathbb{C}^n$ be the state spaces of two quantum systems $A$ and $B$. The state space of their composite system is the bipartite tensor product space
$	\mathcal{H}_{AB}:=\mathcal{H}_A\otimes\mathcal{H}_B\cong \mathbb{C}^m\otimes\mathbb{C}^n.$
This notion can be used iteratively to obtain multipartite composite state spaces. In this work we will consider only bipartite systems. A state $|\psi_{AB}\rangle\in\mathcal{H}_{AB}$ is said to be \emph{separable} if there exists $|\psi_A\rangle\in\mathcal{H}_A$ and $|\psi_B\rangle\in\mathcal{H}_B$ such that $|\psi_{AB}\rangle=|\psi_A\rangle\otimes|\psi_B\rangle$.
If a state is not separable then it is called \emph{entangled}. An equivalent representation of the state of a quantum system utilizes the \emph{density operator} formalism. A set of states $\{|\psi_1\rangle,\dots,|\psi_k\rangle\}$ together with probabilities $p_1,\dots,p_k$ summing to $1$ is referred to as an \emph{ensemble}. The density operator corresponding to such an ensemble is defined by
\begin{align*}
\rho:=\sum_{i=1}^k p_i|\psi_i\rangle\langle\psi_i|.
\end{align*}
Their statistical interpretation requires density operators to be positive semi-definite and have unit trace. For a Hilbert space $\mathcal{H}$, any linear operator acting on $\mathcal{H}$, $\rho\in\mathcal{L}(\mathcal{H})$ enjoying these two properties is a valid density operator, and describes some ensemble of state vectors in $\mathcal{H}$.
A density operator $\rho $
is said to be \emph{pure} if $\rho ^{2}=\rho$, otherwise it is called \emph{mixed}.  We will often refer to density operators as simply \emph{states}. A density operator $\rho\in\mathcal{L}(\mathcal{H}_{AB})$ is separable if it can be written as
\begin{align}
\rho =\sum_{i}p_{i}\rho _{i}\otimes \sigma _{i},  \label{eq:separable}
\end{align}
where $\rho_i\in\mathcal{L}(\mathcal{H}_A)$ and $\sigma_i\in\mathcal{L}(\mathcal{H}_B)$ are density operators, and $p_i\ge 0$ such that $\sum_i p_i=1$. As before, a density operator that is not separable is called entangled.

Of course, in the study of quantum systems with finite dimensions, one can fix a basis and represent the system's state using a \emph{density matrix}. Since we will be working in a fixed basis, we refer to density matrices throughout, rather than density operators. In \cite{Gurvits}, Gurvits showed that it is NP-hard to decide if a 
 density matrix of a bipartite state is separable. Subsequent work on the computational complexity of testing separability of quantum states is covered in detail by Ioannou \cite{Ioannou}.

\subsection{Faulty emitters and combinatorial entanglement}
\label{subsection:faultyemittersandcombinatorialentanglement}
We describe here a setting which is useful for giving a physical interpretation of some of the results discussed in this paper. Consider a hypothetical device which is designed to emit a specific pure state. Perhaps this state is entangled, and the device is designed to be an entanglement source for use in a quantum key distribution protocol \emph{\`a la} Ekert \cite{Ekert}. Suppose that on occasion, the device makes an error, and emits a state that is different from the one it is designed to emit. Suppose further that a fail-safe mechanism within the device can tell when an erroneous emission is made. This mechanism is able to prepare a shortlist of pure states such that, while the device may not know what state it just emitted, it knows with certainty that the state is on the shortlist.

\begin{figure}[ht!]
\centering
\includegraphics[scale=0.35]{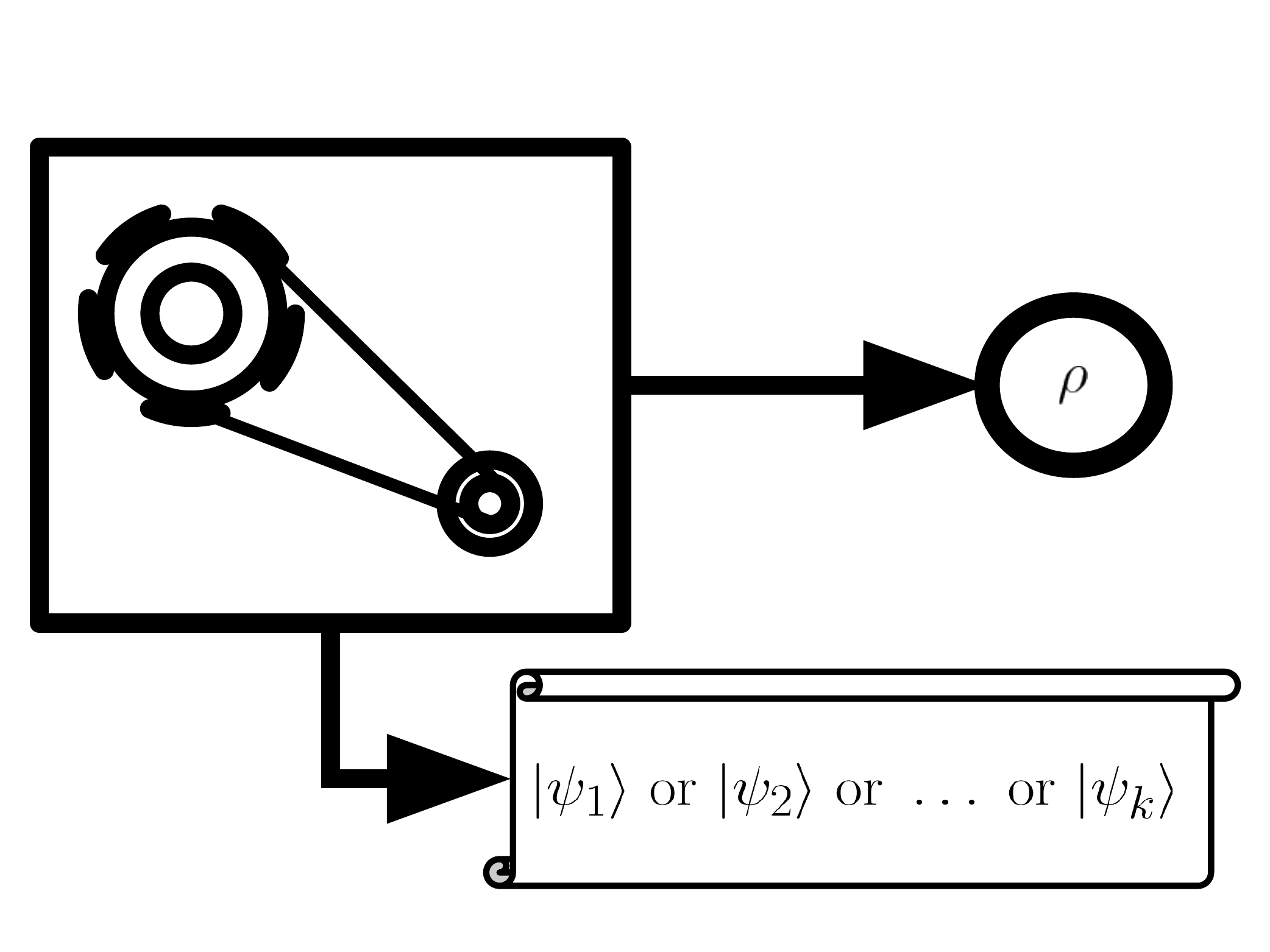}
\caption{A depiction of the faulty emitter scenario. When there is uncertainty about the state that has been emitted, this case is described by a density matrix $\rho$. In this work we will attempt to understand how the debugging information given to us by the emitter, denoted by the writing on the scroll, can help us understand the entanglement in the state $\rho$.}
\label{fig:machine}
\end{figure}

In this work, we explore how this ``debugging information'' can sometimes be used to test if the emitted state is entangled, without the recipient of the state interacting with it. In Figure \ref{fig:machine}, we depict the scenario pictorially. More formally, when the device makes an erroneous  emission of a pure state $|e\rangle$, it also emits a classical description of a list of $k$ states $E$ such that
for any state $|e'\rangle\in E$, the probability that $e=e'$,
\begin{align*}
P(e=e')=\frac{1}{k}.
\end{align*} 
The state of the experimenter's system after an erroneous emission with shortlist $E$ is therefore
\begin{align*}
\rho(E)=\frac{1}{k}\sum_{|e\rangle \in E}|e\rangle\langle e|.
\end{align*}
We have motivated the study of the following problem.

\begin{problem}
	$(k,a,b)$-\textsc{Separability}\\
		\textit{Input:} \emph{A finite set of $k$ pure states $E=\{|e\rangle\in\mathbb{C}^a\otimes\mathbb{C}^b\}$.}\\
		\textit{Question:} \emph{Is the state}
		\begin{align*}
		\rho(E)=\frac{1}{k}\sum_{|e\rangle \in E}|e\rangle\langle e|
		\end{align*}
		\emph{separable?}
\end{problem}
It is convenient at this early of analysis to discuss a refinement of this problem, the study of which will make up the majority of this work. We require the following definition.
\begin{definition}[Edge States]
	Let $\mathcal{H}_{A}\otimes\mathcal{H}_{B}\cong\mathbb{C}^m\otimes \mathbb{C}^n$ be a bipartite tensor product Hilbert space. Let $\{|i\rangle,|k\rangle\}\subset\mathcal{H}_A$ (resp. $\{|j\rangle,|l\rangle\}\subset\mathcal{H}_B$) be standard basis elements of $\mathcal{H}_A$ (resp. $\mathcal{H}_B$). Then the \emph{edge state} $|i,j;k,l\rangle\in\mathcal{H}_A\otimes\mathcal{H}_B$ is defined
	\begin{align*}
	|i,j;k,l\rangle:=\frac{1}{\sqrt{2}}\left(|i,j\rangle-|k,l\rangle\right).
	\end{align*}
	\label{def:edgeStates}
\end{definition}
The following problem is a refinement of $(k,a,b)-$\textsc{Separability}, where the states in the uniform mixture are always edge states.
\begin{problem}
		$(k,a,b)$-\textsc{EdgeSeparability}\\
		\textit{Input:} \emph{A finite set of $k$ edge states $E=\{|e\rangle\in\mathbb{C}^a\otimes\mathbb{C}^b\}$.}\\
		\textit{Question:} \emph{Is the state}
		\begin{align*}
		\rho(E)=\frac{1}{k}\sum_{|e\rangle \in E}|e\rangle\langle e|
		\end{align*}
		\emph{separable?}
\end{problem}
As we shall see, working with this  problem will require the study of a combinatorial object which we call a \emph{grid-labelled graph}.

\subsection{Grid-labelled graphs}
\label{subsection:gridlabelledgraphs}
We begin this section by defining the central notion of the present work:

\begin{definition}[Grid-labelled graph]
	\label{def:gridlabelledgraph}
	A \emph{grid-labelled graph} is a tuple $G_{l}^{a,b}=(G,l;a,b)$, where $G$ is
	a simple graph on $n=a\cdot b$ vertices $(a,b\ge 1)$, $l:V(G)\longrightarrow
	\lbrack a]\times \lbrack b]$ is a bijective labelling of the vertices of $G$, with $%
	[a]:=\{1,2,...,a\}$ and $[b]:=\{1,2,...,b\}$. When $a$ and $b$ are
	specified, we say that the grid-labelled graph $G_{l}^{a,b}$ is of \emph{type}
	$(a,b)$.
\end{definition}
When the context is clear, we refer to these objects as ``graphs'' and ``grid-labelled graphs'' interchangeably.
Given a labelling $l$, we write $l(G)$ instead of $l(V(G))$, in order to simplify the notation.
Any graph $G$ on $k\leq a\cdot b$ vertices can be associated to a
grid-labelled graph $G_{l}^{a,b}$, by choosing a labelling $l(G)$ and by
adding $a\cdot b-k$ isolated vertices. We say that $G$ is \emph{the graph}
of $G_{l}^{a,b}=(G,l;a,b)$.
\begin{figure}[ht!]
\includegraphics[scale=0.08]{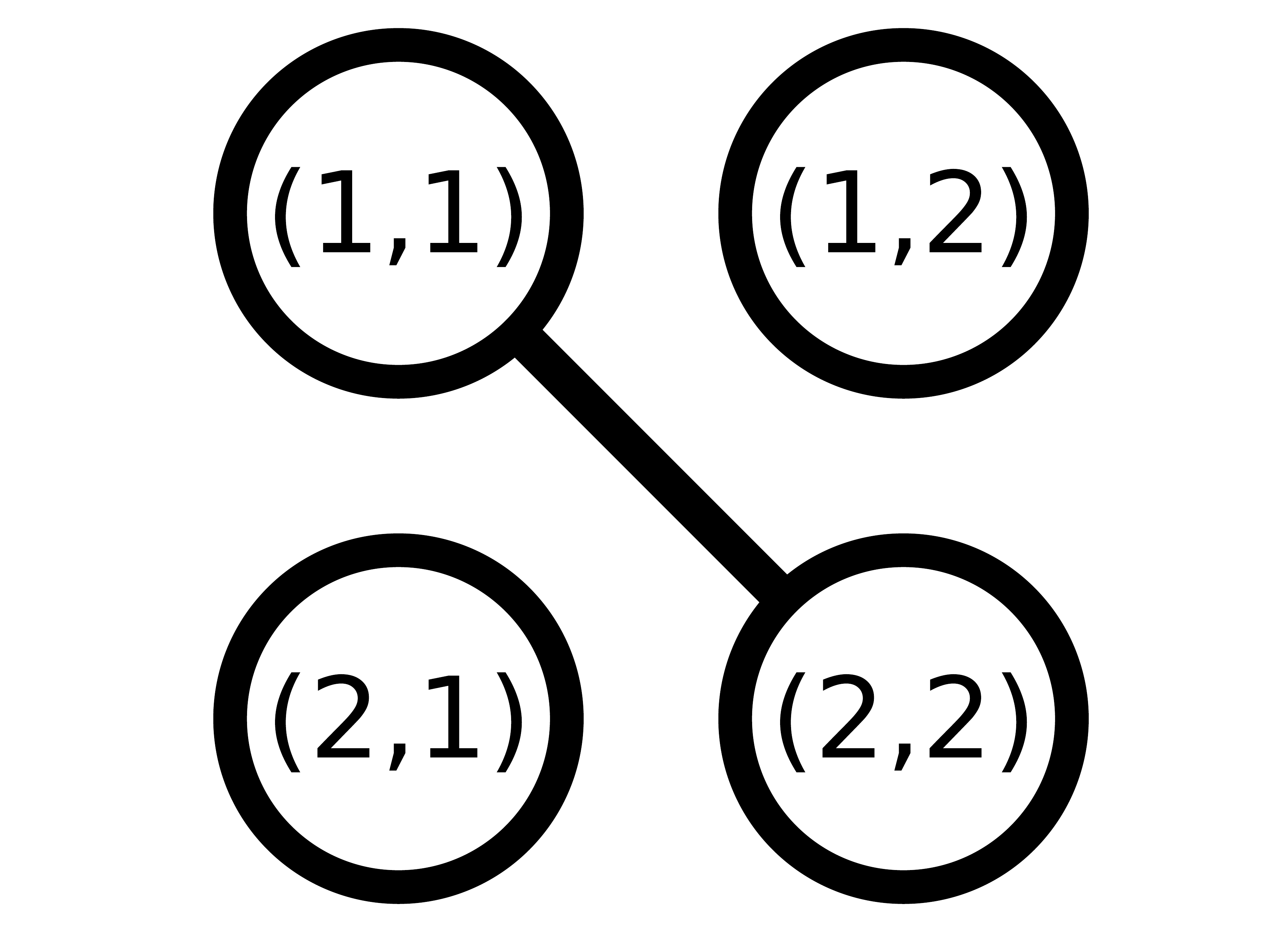}
\caption{The grid-labelled graph $G_l^{2,2}$ as described in Example \ref{example:k2}.}
\label{fig:examplek2h}
\end{figure}

\begin{example}
	\label{example:k2}\emph{Let us consider }$K_{2}$\emph{, where we denote by }$%
	K_{n}$\emph{\ the }complete graph\emph{\ on }$n$\emph{\ vertices. Let us use this to construct a grid-labelled graph }$G_{l}^{2,2}=(G,l;2,2)$\emph{. First of
		all, by Definition \ref{def:gridlabelledgraph}, we work with }$G=K_{2}\uplus K_{1}\uplus K_{1}$\emph{, in
		order to have }$|V(G)|=4=2\cdot 2$. \emph{The graph }$K_{1}$\emph{\ is said to be an \emph{isolated
	vertex}. The symbol `$\uplus$' denotes disjoint union of two graphs. Let us choose }$V(K_{2})=\{1,2\}$%
	\emph{\ and }$V(K_{1}^{(1,2)})=\{3\},\{4\}$\emph{. Then, we can define }$%
	l(1)=(1,1)$\emph{, }$l(2)=(2,2)$\emph{, }$l(3)=(1,2)$\emph{, and }$%
	l(4)=(2,1) $\emph{. The graph obtained is illustrated in Figure \ref{fig:examplek2h}.}
\end{example}

A $d$\emph{-dimensional grid} over finite sets $X_1\dots X_d\subset \mathbb{N}$ is the set 
\begin{equation*}
X_{1}\times \cdots \times X_{d}:=\{(x_{1},...,x_{d}):x_{i}\in X_{i}, \text{ for all }1\leq i\leq
d\}\subset \mathbb{N}^{d}.
\end{equation*}%
In our case, $X_{1}=[a]$ and $X_{2}=[b]$. Even if obvious, it is important
to clarify the following.

\begin{lemma}
	There is a unique way of drawing a grid-labelled graph on $n=a\cdot b$
	vertices on the grid $[a]\times \lbrack b]$, where each point corresponds to
	an ordered pair $(i,j)$ with $i\in \lbrack a]$ and $j\in \lbrack b]$.
\end{lemma}

For our purposes, it will be convenient to represent grid-labelled graphs by
their (normalised) \emph{combinatorial Laplacian} matrices. Given a labelled graph $%
G $ on $n$ vertices $\{1,2,...,n\}$ and $m$ edges, the combinatorial Laplacian of $%
G $ is the $n\times n$ matrix, indexed by the vertices of $G$, with $uv$-th
entry $[L(G)]_{u,v}=d_G(u)$ if $u=v$, $[L(G)]_{u,v}=0$ if $\{u,v\}\notin E(G)$%
, and $[L(G)]_{u,v}=-1$ if $\{u,v\}\in E(G)$. It follows that 
\begin{equation*}
L(G)=D(G)-A(G),
\end{equation*}%
where $D(G)$ is the \emph{degree matrix} of $G$ and $A(G)$ is its 
\emph{adjacency matrix}. Recall that $[D(G)]_{u,v}=\delta
_{u,v}d_{G}(i)$, where $\delta _{u,v}$ is the Kronecker symbol and $%
d_{G}(u)=|\{v:\{u,v\}\in E(G)\}|$ is the \emph{degree} of $u\in V(G)$; also, 
$[A(G)]_{u,v}=1$ if $\{u,v\}\in E(G)$, and $[A(G)]_{u,v}=0$ otherwise.
The adjacency, degree and combinatorial Laplacian matrices of a grid-labelled graph are defined in the obvious way, with matrix entries indexed by labelled vertices $(i,j)$. For a grid-labelled graph $%
G_l^{a,b}$ with $m$ edges, we can obtain a density matrix from its combinatorial Laplacian
matrix by normalisation, 
\begin{align*}
\rho(G_l^{a,b}):=L(G_l^{a,b})/2m.
\end{align*}

\begin{example}
	\emph{The density matrix associated with the grid-labelled graph $%
		G_{l}^{2,2}$ with the labelling $l$ defined as in Example \ref%
		{example:k2} is} 
	\begin{equation*}
	\rho (G_{l}^{2,2})=\frac{1}{2}%
	\begin{pmatrix}
	1 & 0 & 0 & -1 \\ 
	0 & 0 & 0 & 0 \\ 
	0 & 0 & 0 & 0 \\ 
	-1 & 0 & 0 & 1%
	\end{pmatrix}%
	.
	\end{equation*}
\end{example}

Let $\rho $ be a $a\cdot b\times a\cdot b$ density matrix. As we have mentioned in Section \ref{subsection:statesandentanglement}, $\rho $ is
said to be \emph{separable} if it can be written as
$
\rho =\sum_{i}p_{i}\rho _{i}\otimes \sigma _{i}, 
$
where $\rho _{i}$ and $\sigma _{i}$ are $a\times a$ and $b\times b$ density
matrices respectively, acting on the Hilbert spaces of the subsystems, and $p_{i}\geq 0$ such that $\sum_{i}p_{i}=1$. The set of separable density matrices is denoted by $\mathcal{S}$.

We can interpret $\rho $ as a linear operator acting on a tensor product
Hilbert space $\mathcal{H}_{AB}=\mathcal{H}_{A}\otimes \mathcal{H}_{B}$,
where $\mathcal{H}_{A}\cong \mathbb{C}_{A}^{a}$ and $\mathcal{H}_{B}\cong 
\mathbb{C}_{B}^{b}$, with $\dim (\mathcal{H}_{A})=a$ and $\dim (\mathcal{H}%
_{B})=b$. Notice that there is a correspondence between $[a]$ (resp. $[b]$) and the
elements of the standard basis of $\mathcal{H}_{A}$ (resp. $\mathcal{H}_{B}$%
). It is useful to label the entries of $\rho(G_l^{a,b})$ with the elements of an
orthonormal basis of $\mathcal{H}_{AB}$. For any grid-labelled graph $%
G_{l}^{a,b}$, we work with the standard basis elements $|i,j\rangle
:=|i\rangle \otimes |j\rangle \in \mathcal{H}_{A}\otimes \mathcal{H}_{B}$,
corresponding to the ordered pairs that label the vertices of $G_{l}^{a,b}$%
. Hence, once we have fixed the labelling $l$, we can simply write 
\begin{equation*}
V(G_{l}^{a,b})=\{l(v)=(i,j):v\in V(G \uplus_{a\cdot b-|V(G)|}K_{1})\}
\end{equation*}%
and, similarly,%
\begin{equation*}
E(G_{l}^{a,b})=\{\{l(u),l(v)\})=\{(i,j),(k,l)\}:\{u,v\}\in E(G)\}.
\end{equation*}%
With this notation, we have
\begin{align}
\rho(G_{l}^{a,b})&:=\frac{1}{2|E(G_{l}^{a,b})|}\sum_{\{(i,j),(k,l)\}\in
	E(G_{l}^{a,b})}(|i,j\rangle -|k,l\rangle )(\langle i,j|-\langle k,l|).\label{eq:basisform}
\end{align}
Clearly, $\rho(G_{l}^{a,b})$ is equal to $L(G)/2|E(G)|$, with the labelling $%
l.$
The states appearing in the sum in Equation (\ref{eq:basisform}) are associated with the edges of the graph $G_l^{a,b}$. Such states appear often in this work, indeed, these are the \emph{edge states} from Definition \ref{def:edgeStates}. Hence, we can define the density matrix $\rho(G_l^{a,b})$ in terms of edge states,
\begin{align*}
&\rho(G_l^{a,b}):=\frac{1}{|E(G_{l}^{a,b})|}\sum_{\{(i,j),(k,l)\}\in
	E(G_{l}^{a,b})}|i,j;k,l\rangle\langle i,j;k,l|.
\end{align*}%
Therefore, $(k,a,b)-$\textsc{EdgeSeparability} is exactly equivalent to the problem of testing separability of the normalised combinatorial Laplacian matrix of a grid-labelled graph of type $(a,b)$ with $k$ edges. This problem is defined formally as follows.

\begin{problem}
		$(k,a,b)$-\textsc{GraphSeparability}\\
		\textit{Input:} \emph{A grid-labelled graph $G_l^{a,b}$ with $k$ edges.}\\
		\textit{Question:} \emph{Is the state $\rho(G_l^{a,b})$
		separable?}
\end{problem}

Notice that the number of grid-labelled graphs on $n=a\cdot b$ vertices is $%
2^{n(n-1)/2}$. In fact, the number of labelled graphs on $n$ vertices
is $2^{n(n-1)/2}$ and there is a simple bijection between labelled graphs and
grid-labelled graphs. However, a grid-labelled graph can be seen as a labelled graph with an ordered pair labelling each vertex, a fact with a number of consequences.

\subsection{Local isomorphism}
\label{subsection:localisomorphisms}
A complex matrix $U$ is \emph{unitary} if $UU^{\dagger }=U^{\dagger }U=I$,
where $U^{\dagger }$ is the Hermitian conjugate of $U$ and $I$ is the identity matrix. The following result
is well-known (see, \emph{e.g.}, \cite{Horodecki2009}):

\begin{proposition}
	\label{proposition:localiso}Let $\rho $ be a density matrix acting on $%
	\mathcal{H}_{AB}$. Let $U_{A}$ and $U_{B}$ be unitary matrices acting on $%
	\mathcal{H}_{A}$ and $\mathcal{H}_{B}$, respectively. Let $\rho ^{\prime
	}=(U_{A}\otimes U_{B})\rho (U_{A}^{\dagger }\otimes U_{B}^{\dagger })$. Then, $%
	\rho ^{\prime }\in \mathcal{S}$ if and only if $\rho \in \mathcal{S}$.
\end{proposition}

A \emph{permutation} $\pi $ of length $n$ is a bijection $\pi
:[n]\rightarrow \lbrack n]$, where $[n]=\{1,2,...,n\}$. The \emph{%
	permutation matrix} of $\pi $, denoted by $P_{\pi }$, has $ij$-th entry $%
[P_{\pi }]_{i,j}=1$ if $\pi (i)=j$ and $[P_{\pi }]_{i,j}=0$ otherwise. We
write $P$ when we do not need to specify $\pi $. We denote by $M^{T}$ the
transpose of a real matrix $M$.

\begin{definition}[Locally isomorphic grid-labelled graphs]
	Let $G_{l}^{a,b}$ and $H_{m}^{a,b}$ be two grid-labelled graphs. These graphs are said to be 
	\emph{locally isomorphic} if there are permutations $\pi$ and $%
	\sigma$ on $[a]$ and $[b]$, respectively, such that \begin{align*}
	(P_{\pi }\otimes
	P_{\sigma })A(G_{l}^{a,b})(P_{\pi }^{T}\otimes P_{\sigma }^{T})=A(H_{m}^{a,b}).
	\end{align*} 
	When $G_{l}^{a,b}$ and $H_{m}^{a,b}$ are locally isomorphic, we write $%
	G_{l}^{a,b}\cong _{a,b}H_{m}^{a,b}$. The ordered pair of permutations $(\pi
	,\sigma )$ inducing the permutation matrices $P_{\pi }$ and $P_{\sigma }$ is
	said to be a \emph{local isomorphism}.
\end{definition}
It is useful to observe that locally isomorphic graphs can be obtained from
one another by permuting the rows and the columns of the Cartesian grid $%
[a]\times \lbrack b]$. Indeed, the permutation matrices $P_\pi$ and $P_\sigma$ act on the rows and columns of the grid respectively.
\begin{figure}[tbp]
	\centering
	\includegraphics[scale=0.2]{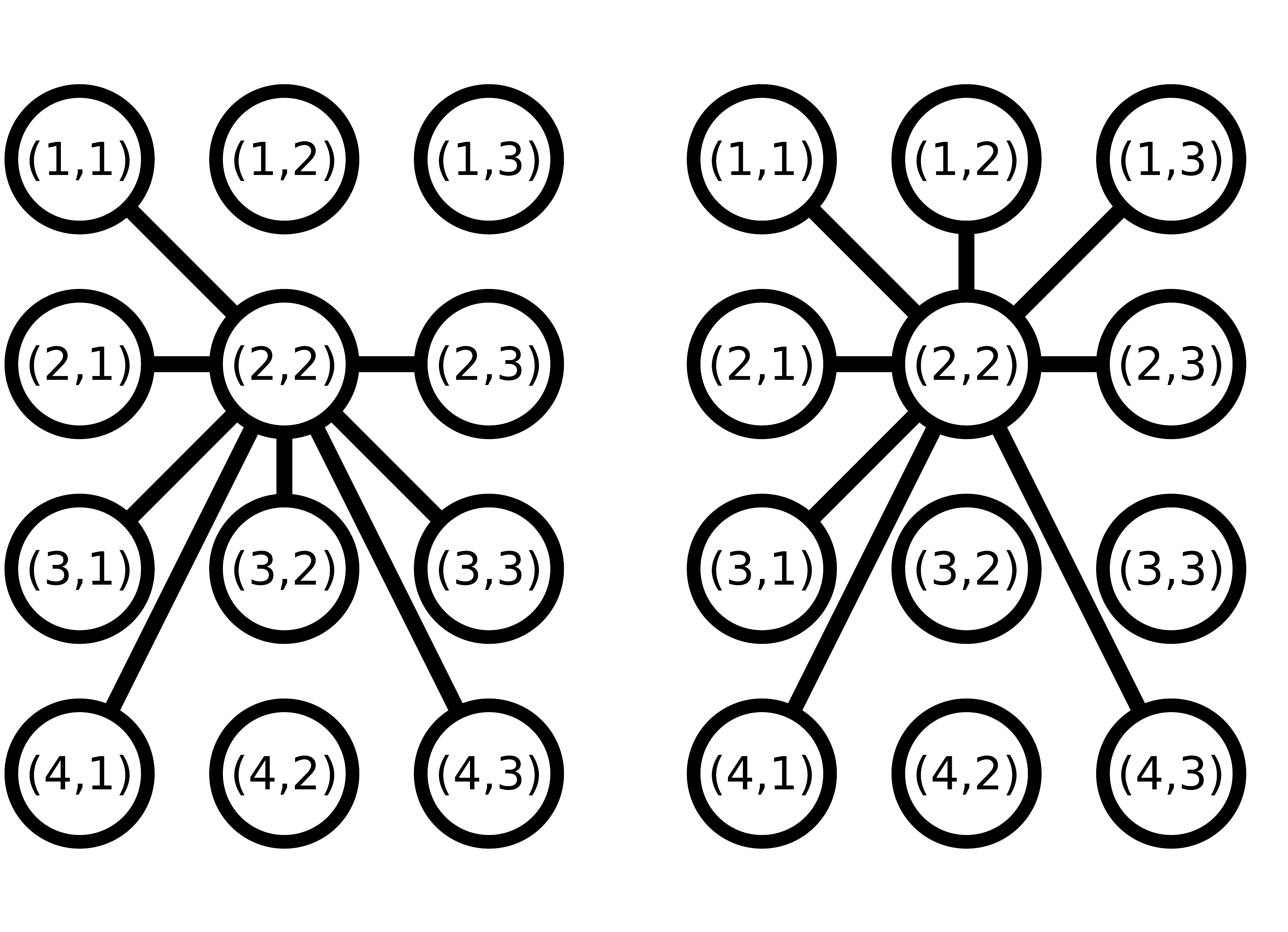}
	\caption{Two locally isomorphic grid-labelled graphs of type $(4,3)$}
	\label{fig:locallyisomorphic}
\end{figure}

\begin{example}
\emph{The grid-labelled graphs in Figure \ref{fig:locallyisomorphic} are locally isomorphic.
	Note how they can be obtained from one another by swapping the first and third rows of vertices.}
\end{example}

In standard graph-theoretic terminology, two graphs $G$ and $H$ are said to be \emph{isomorphic} if there is a
permutation matrix $P$ such that $PA(G)P^{T}=A(H)$. When $G$ and $H$ are
isomorphic, we write $G\cong H$.

\begin{proposition}
	Two isomorphic grid-labelled graphs are not necessarily locally
	isomorphic.
\end{proposition}
\begin{proof}
	Not every permutation matrix of
	size $a\cdot b$ can be written as $P\otimes Q$, for permutation matrices $P$
	and $Q$.
\end{proof}

For the sake of brevity, instead of writing $\rho(G_{l}^{a,b})\in \mathcal{S}$ we write $%
G_{l}^{a,b}\in \mathcal{S}$. Proposition \ref{proposition:localiso} can be
easily adapted to our context:
\begin{proposition}
	\label{proposition:localisomorphism}Let $G_{l}^{a,b}$ and $H_{m}^{a,b}$ be
	grid-labelled graphs. Let $G_{l}^{a,b}\cong _{a,b}H_{m}^{a,b}$. Then, $%
	G_{l}^{a,b}\in \mathcal{S}$ if and only if $H_{m}^{a,b}\in \mathcal{S}$.
\end{proposition}
\begin{proof}
	The statement follows from Proposition \ref{proposition:localiso} and from
	the fact that a permutation matrix $P$ is unitary (in fact,
	real-orthogonal), since $PP^{T}=P^{T}P=I$.
\end{proof}
Let us show that testing if two grid-labelled graphs are locally isomorphic is at least as hard as testing if two graphs are isomorphic. The graph isomorphism problem is as follows.
\begin{problem}
\textsc{GraphIsomorphism}\\
\textit{Input: } \emph{Two graphs, $G$, and $H$.}\\
\textit{Question: } \emph{Is $G$ isomorphic to $H$?}
\end{problem}
On the other hand, the local isomorphism problem is defined like so.
\begin{problem}
\textsc{LocalIsomorphism}\\
\textit{Input: } \emph{Two grid-labelled graphs, $G_l^{a,b}$, and $H_m^{a,b}$}.\\
\textit{Question: } \emph{Is $G_l^{a,b}$ locally isomorphic to $H_m^{a,b}$?}\\
\end{problem}
The following theorem states that deciding \textsc{LocalIsomorphism} is at least as hard as \textsc{GraphIsomorphism}.
\begin{theorem}
There exists a polynomial time Karp reduction from \textsc{GraphIsomorphism} to \textsc{LocalIsomorphism}.
\end{theorem}

\begin{proof}
In order to prove this, we must show that for every pair of simple undirected graphs $G$ and $H$ there exists a corresponding pair of grid-labelled graphs $G_l^{a,b}$ and $H_m^{a,b}$ such that $G_l^{a,b}\cong_{a,b} H_m^{a,b}$ if and only if $G\cong H$. If $G$ and $H$ do not have the same number of vertices, they can not be isomorphic, so w.l.o.g. we can assume they do.

Every graph $G=(V,E)$ with $n=|V|$ vertices labelled according to some bijection $l:V\rightarrow [n]$ has a unique grid-labelled graph $G_q^{1,n}$, where the labelling $q$ is defined such that $q(1,j)=l(j)$. It is obvious that these corresponding grid-labelled graphs can be constructed in time polynomial in the number of vertices in the original graph. The set of possible permutations of the vertex labels of $G$ and $G_q^{1,n}$ are identical. Then, by definition of local isomorphism, two graphs $G$ and $H$ are isomorphic if and only if their corresponding grid-labelled graphs defined in this way are locally isomorphic. 
\end{proof}
Hence, \textsc{LocalIsomorphism} is a generalisation of \textsc{GraphIsomorphism}. We will attempt no further classification of the problem in this paper, however, it is obvious that it is in the class NP.

\subsection{Local operations and classical communication}
\label{subsection:localoperationsandclassicalcommunication}
\emph{Local operations and classical communication (LOCC)} capture the set
of transformations of density matrices $\rho \in \mathcal{H}_{A}\otimes 
\mathcal{H}_{B}$ that do not increase the amount of entanglement \cite{Horodecki2009}. Local
operations can be interpreted as the set of transformations achievable by
acting locally on $\mathcal{H}_{A}$ and $\mathcal{H}_{B}$, respectively.
Consider the parts of $\rho $ corresponding to $\mathcal{H}_{A}$ and $%
\mathcal{H}_{B}$ being given to two parties (Alice and Bob) respectively. LOCC captures
everything that can be done to $\rho$ by Alice and Bob when acting
independently on their parts (local operations), and coordinating with one
another by sending bits back and forth (classical communication). As part of
a formal description of LOCC, Definition \ref{def:localunitary} is important:

\begin{definition}[Local unitary]
	\label{def:localunitary}A unitary matrix $U$ acting on a bipartite Hilbert
	space $\mathcal{H}_{AB}$ is \emph{local} if it can be expressed as $%
	U=U_{A}\otimes U_{B}$, for unitary matrices $U_{A}$ acting on $\mathcal{H}%
	_{A}$ and $U_{B}$ acting on $\mathcal{H}_{B}$.
\end{definition}

Below, we will work with LOCC and grid-labelled graphs. The next definition
is useful for our purpose:

\begin{definition}[Horizontal, vertical, diagonal edge]
	Let $G_{l}^{a,b}$ be a grid-labelled graph. Let $\{(i,j),(k,l)\}\in
	E(G_{l}^{a,b})$. An edge $\{(i,j),(k,l)\}\in E(G_l^{a,b})$ is said to be
	
	\begin{itemize}
		\item \emph{Horizontal} if $i=k$
		and $j\neq l$;
		
		\item \emph{Vertical} if $j=l$ and $%
		i\neq k$;
		
		\item \emph{Diagonal} if $i\neq k$
		and $j\neq l$.
	\end{itemize}
\end{definition}

For certain grid-labelled graphs, it is possible to construct their
corresponding density matrices purely via LOCC. Indeed, Theorem \ref%
{theorem:loccgraphs} below gives an infinite family of grid-labelled graphs
that can be constructed in such a way. The theorem requires a technical
lemma: 

\begin{lemma}
	\label{lemma:itsunitary}The unitary matrix $U_{(i,j),(k,l)}$ acting on a bipartite Hilbert space $\mathcal{H}_{AB}\cong \mathbb{C}^a\otimes \mathbb{C}^b$ for $a,b\ge 2$, for $(i,j)\neq(k,l)$ defined such that
	\begin{equation*}
	U_{(i,j),(k,l)}|1\rangle _{A}|1\rangle _{B}\mapsto \frac{1}{\sqrt{2}%
	}(|i\rangle _{A}|j\rangle _{B}-|k\rangle _{A}|l\rangle _{B})
	\end{equation*}%
	and which acts as identity elsewhere, is a local unitary if
	
	\begin{itemize}
		\item $i=k$ and $j\neq l$, or
		
		\item if $i\neq k$ and $j=l$.
	\end{itemize}
\end{lemma}

\begin{proof}
	Starting with the separable pure state $|1\rangle _{A}|1\rangle _{B}$, we
	construct the grid-labelled graph $(K_{2})_{l}^{a,b}$, for $a,b\ge 2$, with 
	with a single (horizontal) edge $\{(1,1),(1,2)\}$. To do that, we use the
	matrix 
	\begin{equation*}
	V=I_{A}\otimes (H_{B}\oplus I),
	\end{equation*}%
	where $I_{A}$ is the identity matrix of size $a\times a$ acting on $\mathcal{H}_{A}$, 
	\begin{equation*}
	H_{B}=\frac{1}{\sqrt{2}}\left( 
	\begin{array}{rr}
	1 & 1 \\ 
	-1 & 1%
	\end{array}%
	\right) ,
	\end{equation*}%
	and $I$ is an identity matrix of size $(b-2)\times (b-2)$. Here,
	\textquotedblleft $\oplus $\textquotedblright\ denotes matrix direct sum.
	The matrix $V$ is real-orthogonal, because $H_B$ is a Hadamard matrix. The grid-labelled graph $(K_{2})_{l}^{a,b}$, with
	density matrix 
	\begin{equation*}
	\rho((K_{2})_{l}^{a,b})=\frac{1}{2}(|1\rangle _{A}|1\rangle _{B}-|1\rangle
	_{A}|2\rangle _{B})(\langle 1|_{A}\langle 1|_{B}-\langle 1|_{A}\langle
	2|_{B}),
	\end{equation*}%
	is locally isomorphic to any other $(K_{2})_{l^{\prime }}^{a,b}$ with a
	different labelling $l^{\prime }$ such that $i=k$ and $j\neq l$. The same
	holds for the case $i\neq k,j=l$ by similar reasoning. The matrix $%
	U_{(i,j),(k,l)}$ is obtained by the application of $I_{A}\otimes
	(H_{B}\oplus I)$ followed by local isomorphisms. Thus, $U_{(i,j),(k,l)}$ is
	unitary.
\end{proof}

With the use of this lemma, we can now prove the theorem.
\begin{theorem}
	\label{theorem:loccgraphs}Let $G_{l}^{a,b}$ be a grid-labelled graph with $m$
	edges. Suppose that $G_{l}^{a,b}$ does not have diagonal edges. Then $%
	\rho(G_{l}^{a,b})$ can be obtained from the separable pure density matrix $%
	\rho _{0}=|1\rangle _{A}|1\rangle _{B}\langle 1|_{A}\langle 1|_{B}$ by application of local unitaries and $\log m$ bits of classical communication.
\end{theorem}

\begin{proof}
	We prove the statement by demonstrating a protocol by which two parties can turn the
	density matrix $\rho _{0}$ into $\rho(G_{l}^{a,b})$, for grid-labelled graphs $G_l^{a,b}$
	with no diagonal edges. The protocol uses only classical communication and
	local unitaries acting on $\rho _{0}$.
	
	With every edge $\{(i,j),(k,l)\}$, we can associate the unitary matrix $%
	U_{(i,j),(k,l)}$. If all edges are horizontal or vertical, then by Lemma \ref%
	{lemma:itsunitary} all $m$ of these matrices are implementable by applying local unitaries on $\rho _{0}$.
	Both Alice and Bob know what each of these matrices are, and which local
	unitaries they need to perform on their respective Hilbert space. Each
	unitary $U_{(i,j),(k,l)}$ is labelled by an integer $1,2,...,m$, such that the label assignments are known to
	both Alice and Bob. The protocol is as follows:
	\begin{itemize}
		\item[\emph{Step 1.}] Alice selects a uniformly random integer $x\in \lbrack m]$.
		
		\item[\emph{Step 2.}] Alice selects the unitary $U_{(i,j),(k,l)}$ labelled by $x$ and performs her half on $|1\rangle _{A}$. She sends $x$ to Bob.
		
		\item[\emph{Step 3.}] Bob performs his half of $U_{(i,j),(k,l)}$,
		with label $x$ on $|1\rangle _{B}.$
		
		\item[\emph{Step 4.}] Alice and Bob both erase the classical memory that stored
		their copy of $x$.
	\end{itemize}
	The protocol uses only local unitaries. At the end of the protocol, the
	final state is equal to 
	\begin{align*}
	\rho & =\frac{1}{m}\sum_{\{(i,j),(k,l)\}\in
		E(G_{l}^{a,b})}U_{(i,j),(k,l)}\rho _{0}{U_{(i,j),(k,l)}}^{\dagger } \\
	& =\frac{1}{2m}\sum_{\{(i,j),(k,l)\}\in E(G_{l}^{a,b})}(|i\rangle
	_{A}|j\rangle _{B}-|k\rangle _{A}|l\rangle _{B})(\langle i|_{A}\langle
	j|_{B}-\langle k|_{A}\langle l|_{B}) \\
	& =\rho(G_{l}^{a,b}),
	\end{align*}%
	via the edge state form in Equation (\ref{eq:basisform}). Alice needs to communicate $\log m$ bits to communicate the value of $x\in[m]$ to Bob.
\end{proof}
Note that local isomorphisms do not by any means capture all of LOCC. We can conclude this section with the following corollary, which shows how the adjacency structure of a grid-labelled graph influences physical properties of the corresponding quantum state. In a way, this will be the spirit of the present paper.
\begin{corollary}
	\label{corollary:horizontalandverticaledgesonly} If a grid-labelled graph $%
	G_{l}^{a,b}$ has only horizontal and vertical edges then $G_{l}^{a,b}\in 
	\mathcal{S}$.
\end{corollary}

\begin{proof}
	This follows from Theorem \ref{theorem:loccgraphs} and from Proposition \ref%
	{proposition:localisomorphism}. If LOCC is sufficient to obtain $\rho(G_l^{a,b})$ from the separable state $\rho_0$, provided the edges of $G_{l}^{a,b}$ are
	either horizontal or vertical, then for any such grid-labelled graph, $G_{l}^{a,b}\in 
	\mathcal{S}$.
\end{proof}

In the next section, we discuss entanglement criteria for grid-labelled graphs.
\section{The Degree Criterion}
\label{section:thedegreecriterion}
\subsection{The Peres-Horodecki and degree criteria}
A well known necessary condition for separability of a density matrix is based on the \emph{partial transpose}.

\begin{definition}[Partial transpose \cite{Horodecki2009}]
Let $\rho_{AB}\in\mathcal{L}(\mathcal{H}_{AB})$
be a bipartite density operator, and let $|i,j\rangle\subset\mathcal{H}_{AB}$ denote a fixed product basis of $\mathcal{H}_{AB}$. Then the \emph{partial transpose} with respect to subsystem $A$ of $\rho_{AB}$ is the density operator $\rho_{AB}^{\Gamma_A}$ defined such that
\begin{align*}
	\langle i,j|\rho_{AB}^{\Gamma_A}|k,l\rangle:=\langle k,j|\rho_{AB}|i,l\rangle.
\end{align*}
Respectively, the partial transpose with respect to subsystem $B$ of $\rho_{AB}$ is defined
\begin{align*}
\langle i,j|\rho_{AB}^{\Gamma_B}|k,l\rangle:=\langle i,l|\rho_{AB}|k,j\rangle.
\end{align*}
\end{definition}

Armed with this definition, we can state the Peres-Horodecki criterion.
\begin{theorem}[Peres-Horodecki Criterion \cite{Peres,HorodeckiSep}]
\label{theorem:pereshorodecki}
	Let $\rho_{AB}$ be a bipartite density operator. If $\rho_{AB}$ is separable, then $\rho_{AB}^{\Gamma_A}$ (resp. $\rho_{AB}^{\Gamma_B}$) is positive.
\end{theorem}

Indeed, Horodecki \emph{et al.} \cite{HorodeckiSep} subsequently prove that this criterion is necessary and sufficient for separability for density operators acting on bipartite Hilbert spaces of dimension $d\le 6$. 

\begin{theorem}[Horodecki \emph{et al.} \cite{HorodeckiSep}]
	Let $\rho_{AB}$ be a density operator acting on $\mathbb{C}^2\otimes \mathbb{C}^2$ or $\mathbb{C}^2\otimes \mathbb{C}^3$. Then $\rho_{AB}$ is separable if and only if $\rho_{AB}^{\Gamma_A}$ (resp. $\rho_{AB}^{\Gamma_B}$) is positive.
\end{theorem}

In our context, the Peres-Horodecki criterion for testing separability of
density matrices translates into a combinatorial statement called the \emph{%
	degree criterion}, as it was proved in \cite{Braunstein2006, Wu2006}. Let us
begin by defining the notion of partial transpose for grid-labelled graphs:

\begin{definition}[Partial transpose of a graph]
	Given a grid-labelled graph $G_{l}^{a,b}$, its \emph{partial transpose} is
	the grid-labelled graph $\Gamma (G_{l}^{a,b})=(G,l^{\Gamma };a,b)$, where for
	every edge $\{u,v\}\in E(G)$, for $l(u)=(i,j)$ and $l(v)=(k,l)$, the
	labelling $l^{\Gamma }$ is such that $l^{\Gamma }(u)=(k,j)$ and $l^{\Gamma
	}(v)=(i,l)$.
\end{definition}

\begin{figure}[h!]
	\centering
	\includegraphics[scale=0.3]{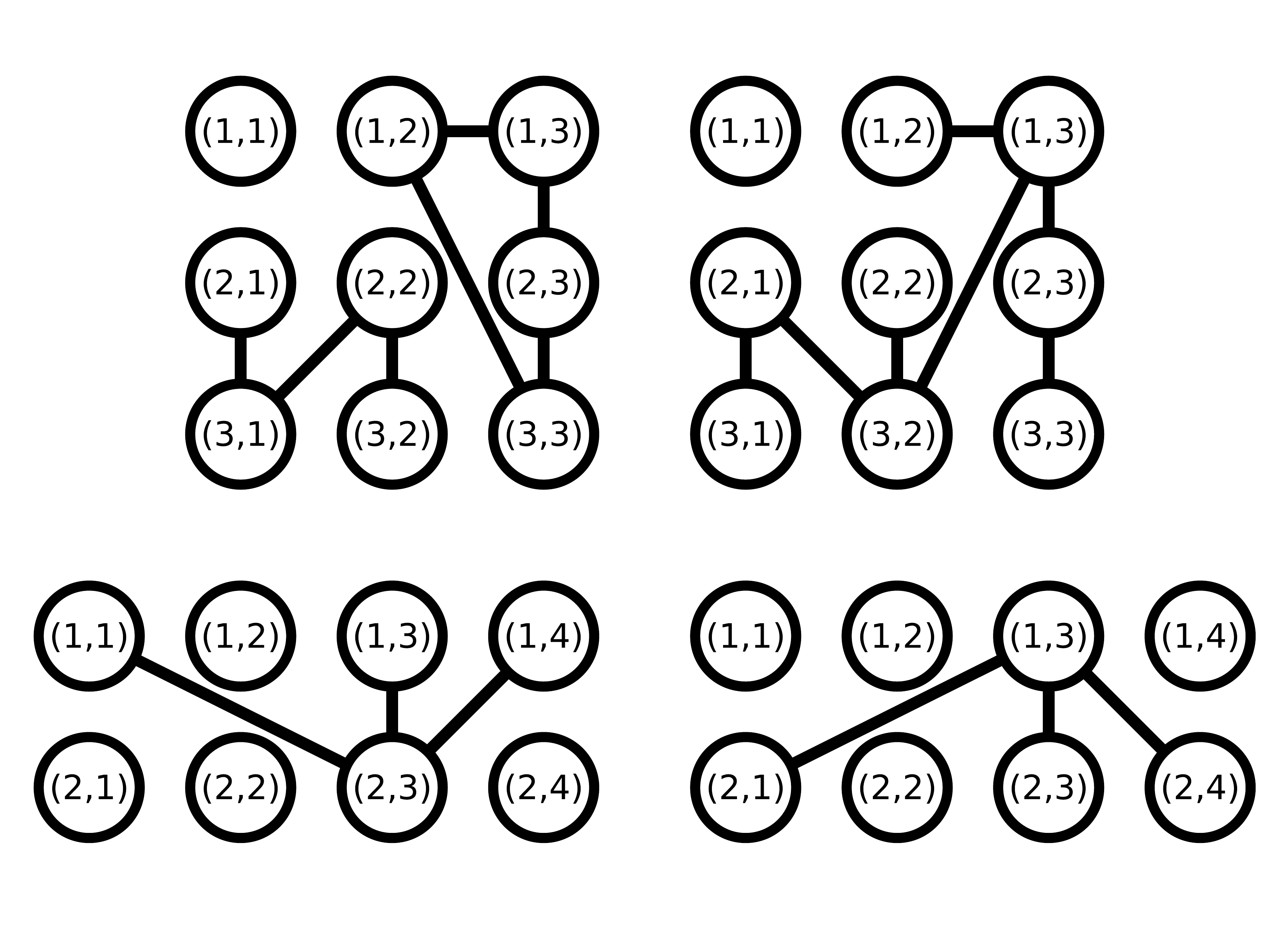}
	\caption{Two grid-labelled graphs side by side with their partial transpositions.}
	\label{fig:partialtransposeexamples}
\end{figure}

\begin{example}
	\emph{In Figure \ref{fig:partialtransposeexamples}, we show two grid-labelled
		graphs alongside their partial transpositions. Note that the partial transpose
		operation does not affect horizontal and vertical edges, but transposes each
		diagonal edge.}
\end{example}
Let us now state the degree criterion. The proof of the following theorem follows the proof of Theorem 2 in \cite{Braunstein2006}.
\begin{theorem}[Degree criterion]
	\label{theorem:degreecriterion}If $G_{l}^{a,b}\in \mathcal{S}$ then 
	\begin{equation*}
	D (G_{l}^{a,b})=D (\Gamma (G_{l}^{a,b})).
	\end{equation*}
\end{theorem}

\begin{proof}
	Let $G_{l}^{a,b}$ be a grid-labelled graph. Let $%
	A(G_{l}^{a,b})$ be the adjacency matrix of $G_{l}^{a,b}$, and let $%
	L(G_{l}^{a,b})$ be the Laplacian matrix of $G_{l}^{a,b}$. For a $a\cdot b\times a\cdot b$ matrix $M$, let the matrix $\Gamma (M):=M^{\Gamma_A}$ be its partial
	transpose. Without loss of generality, we consider only partial transpositions with respect to the $A$ subsystem, and do not consider normalisation factors. Clearly, 
	\begin{align*}
	\Gamma (L(G_{l}^{a,b}))& =D (G_{l}^{a,b})-A(\Gamma (G_{l}^{a,b})) \\
	& =D(G_{l}^{a,b})-D(\Gamma (G_{l}^{a,b}))+L(\Gamma
	(G_{l}^{a,b})).
	\end{align*}%
	Thus, 
	\begin{equation*}
	\Gamma (L(G_{l}^{a,b}))=L(\Gamma (G_{l}^{a,b}))+\Delta,
	\end{equation*}%
	where%
	\begin{equation*}
	\Delta=D (G_{l}^{a,b})-D (\Gamma (G_{l}^{a,b})).
	\end{equation*}%
	It is clear
	that $\Delta$ is an $a\cdot b\times a\cdot b$ real diagonal matrix with trace tr$(\Delta)=0$, since
	\begin{equation*}
	\text{tr}(D)=\text{tr}(D (G_{l}^{a,b}))-\text{tr}(D (\Gamma
	(G_{l}^{a,b})))=0.
	\end{equation*}%
	If $\Delta$ is not equal to the zero matrix, then it follows that there exists one or more negative entries
	on its diagonal. Let $\Delta_{i,i}=b_{i}$ be one such non-zero entry. Let $|{\psi
		_{0}}\rangle =\sum_{j=1}^{a\cdot b}|j\rangle $ be the all-ones vector, and $|{%
		\phi }\rangle := k|i\rangle$ for $k\in \mathbb{R}$.
	
	Then, for $|\chi \rangle =|\psi _{0}\rangle +|\phi \rangle $, 
	\begin{align*}
	\langle \chi |(L(G_{l}^{a,b})+\Delta)|\chi \rangle & =\langle \chi
	|L(G_{l}^{a,b})|\chi \rangle +\langle \chi |\Delta|\chi \rangle \\
	& =\langle \psi _{0}|L(G_{l}^{a,b})|\psi _{0}\rangle +\langle \psi
	_{0}|L(G_{l}^{a,b})|\phi \rangle +\langle \phi |L(G_{l}^{a,b})|\psi
	_{0}\rangle +\langle \phi |L(G_{l}^{a,b})|\phi \rangle \\
	& +\langle \psi _{0}|\Delta|\psi _{0}\rangle +\langle \psi _{0}|\Delta|\phi \rangle
	+\langle \phi |\Delta|\psi _{0}\rangle +\langle \phi |\Delta|\phi \rangle .
	\end{align*}%
	From the fact that $|\psi _{0}\rangle $ is a right eigenvector of $L(G_{l}^{a,b})$
	with eigenvalue $0$, and that $\langle \psi _{0}|\Delta|\psi _{0}\rangle =\text{tr%
	}(\Delta)=0,$ it follows that 
	\begin{align*}
	\langle \chi |(L(G_{l}^{a,b})+\Delta)|\chi \rangle & =\langle \psi
	_{0}|L(G_{l}^{a,b})|\phi \rangle +\langle \phi |L(G_{l}^{a,b})|\phi \rangle
	\\
	& +\langle \psi _{0}|\Delta|\phi \rangle +\langle \phi |\Delta|\psi _{0}\rangle
	+\langle \phi |\Delta|\phi \rangle .
	\end{align*}%
	Finally, we have 
	\begin{align*}
	\langle \psi _{0}|L(G_{l}^{a,b})|\phi \rangle &=\langle \phi
	|L(G_{l}^{a,b})^{T}|\psi _{0}\rangle =\langle \phi |L(G_{l}^{a,b})|\psi
	_{0}\rangle =0;\\
	\langle \phi |L(G_{l}^{a,b})|\phi \rangle &=k^{2}[L(G_{l}^{a,b})]_{i,i}=k^2 d(i),
	\end{align*}
	where $d(i)$ is the degree of vertex $i$,
	\begin{align*}
	\langle \phi |\Delta|\phi \rangle &=b_{i}k^{2},
	\end{align*}
	and
	\begin{align*}
	\langle \psi _{0}|\Delta|\phi \rangle &=\langle \phi |\Delta|\psi _{0}\rangle =b_{i}k.
	\end{align*}%
	Hence,
	\begin{align}
	\label{eq:negative}
	\langle \chi |(L(G_{l}^{a,b})+\Delta)|\chi \rangle =k^{2}(d(i)+b_{i})+2b_{i}k.
	\end{align}%
	Since $b_i<0$ by definition, a positive $k$ can be chosen small enough to make Equation \ref{eq:negative} negative. We
	therefore have that 
	\begin{equation*}
	L(G_l^{a,b})+\Delta\not\geq 0.
	\end{equation*}%
	If $G_{l}^{a,b}\in \mathcal{S}$ then by Theorem \ref{theorem:pereshorodecki}, $\Gamma (L(G_{l}^{a,b}))\geq 0$.
	Since $\Gamma (L(G_{l}^{a,b}))=L(\Gamma (G_{l}^{a,b}))+\Delta$, then this can
	only be true if $\Delta=0$, which means that $D (G_{l}^{a,b})=D (\Gamma
	(G_{l}^{a,b}))$.
\end{proof}

The degree criterion turns out to be necessary and sufficient in the
following case:

\begin{theorem}
	\label{theorem:degrees}For a grid-labelled graph $G_{l}^{2,b}$ with $b\geq 2$, 
	\begin{equation*}
	\Delta (G_{l}^{a,b})=\Delta (\Gamma (G_{l}^{a,b}))
	\end{equation*}%
	if and only if $G_{l}^{a,b}\in \mathcal{S}$.
\end{theorem}

To prove this we require some technical lemmas. Recall that a \emph{decomposition} of a
graph $G$ is a set of subgraphs $\{H_{1},H_{2},...,H_{k}\}$ that partition the
edges of $G$: $\bigcup_{i=1}^{k}H_{i}=G$ and for all $i\neq j$, $E(H_{i})\cap
E(H_{j})=\emptyset $. Notice that isolated vertices do not contribute to a
decomposition and so each $H_{i}$ can always be seen as a spanning subgraph
(this is a subgraph that contains all the vertices). 
It will be useful to redefine some of these graph theoretic
concepts for grid-labelled graphs.
\begin{definition}
	A grid-labelled graph $H_{l}^{a,b}$ is a \emph{subgraph} of a grid-labelled
	graph $G_{l}^{a,b}$, denoted $H_{l}^{a,b}\subseteq G_{l}^{a,b}$, if every
	edge of $H_{l}^{a,b}$ is an edge of $G_{l}^{a,b}$.
\end{definition}

\begin{definition}
	Two grid-labelled graphs $G_{l}^{a,b}$ and $H_{l}^{a,b}$ are \emph{edge
		disjoint} if they have no edges in common, \begin{align*}E(G_{l}^{a,b})\cap E(H_{l}^{a,b})=\emptyset.\end{align*}
\end{definition}

\begin{definition}
	The union of two edge disjoint grid-labelled graphs, $G_{l}^{a,b}$ and $%
	H_{l}^{a,b}$, is the grid-labelled
	graph $G_{l}^{a,b}\cup H_{l}^{a,b}$ with edge set $E(G_{l}^{a,b})\cup E(H_{l}^{a,b})$.
\end{definition}

\begin{definition}
	A set of edge disjoint subgraphs of $G_{l}^{a,b}$, $X=\{(H_{1})_{l}^{a,b},%
	\dots ,(H_{n})_{l}^{a,b}\}$, is called a \emph{decomposition} of $%
	G_{l}^{a,b} $ if 
	$
	\bigcup_{i=1}^{n}(H_{i})_{l}^{a,b}=G_{l}^{a,b}.
	$
\end{definition}

Lemma \ref{lemma:edgedecomposition} is easy to prove, but valuable:

\begin{lemma}[HVD decomposition]
	\label{lemma:edgedecomposition}Given a grid-labelled graph $G_{l}^{a,b}$ with 
	$m$ edges, its density matrix can be written in the form 
	\begin{equation*}
	\rho (G_{l}^{a,b})=\frac{\left\vert E(H_{l}^{a,b})\right\vert }{m}\rho
	(H_{l}^{a,b})+\frac{\left\vert E(V_{l}^{a,b})\right\vert }{m}\rho
	(V_{l}^{a,b})+\frac{\left\vert E(D_{l}^{a,b})\right\vert }{m}\rho
	(D_{l}^{a,b}),
	\end{equation*}%
	where the graphs $H_{l}^{a,b},V_{l}^{a,b},D_{l}^{a,b}$ contain all horizontal, vertical, and diagonal edges of $G_l^{a,b}$ respectively. These graphs form a unique
	decomposition of $G_{l}^{a,b}$, which we will call the \emph{HVD decomposition}.
\end{lemma}

\begin{proof}
	For any grid-labelled graph $G_{l}^{a,b}$ with $m$ edges, 
	\begin{align*}
	\rho (G_{l}^{a,b})& =\frac{1}{2m}L(G_{l}^{a,b}) \\
	& =\frac{1}{2m}(L(H_{l}^{a,b})+L(V_{l}^{a,b})+L(D_{l}^{a,b})) \\
	& =\frac{1}{2m}(2\left\vert E(H_{l}^{a,b})\right\vert \rho
	(H_{l}^{a,b})+2\left\vert E(V_{l}^{a,b})\right\vert \rho
	(V_{l}^{a,b})+2\left\vert E(D_{l}^{a,b})\right\vert \rho (D_{l}^{a,b}))\\
	&=\frac{|E(H_l^{a,b}|}{m}\rho(H_l^{a,b})+\frac{|E(V_l^{a,b}|}{m}\rho(V_l^{a,b})+\frac{|E(D_l^{a,b}|}{m}\rho(D_l^{a,b}).
	\end{align*}%
\end{proof}

\begin{lemma}
	\label{lemma:degreecriterioniff} Let $G_l^{a,b}$ be a grid-labelled graph.
	Let $H_l^{a,b}, V_l^{a,b}$ and $D_l^{a,b}$ be the components of the
	HVD-decomposition of $G_l^{a,b}$. Then, $G_l^{a,b}$ satisfies the degree
	criterion if and only if $D_l^{a,b}$ satisfies the degree criterion.
\end{lemma}

\begin{proof}
	By definition $H_l^{a,b}$ and $V_l^{a,b}$ contain only horizontal and
	vertical edges. These edges remain invariant under the partial transpose
	operation, and hence can be disregarded when considering the degree
	criterion.
\end{proof}

The last component needed to prove Theorem \ref{theorem:degrees} is the following theorem of Ando \cite{Ando}, which we state without proof.

\begin{theorem}[Ando \cite{Ando}:
	 Theorem 4.9]
	 \label{theorem:Ando}
Any positive semi-definite matrix of the form
\begin{align*}
\begin{pmatrix}
C&A\\
A^\dagger&C
\end{pmatrix},
\end{align*}
where $A$ and $C$ are $b\times b$ complex matrices, is separable in $\mathbb{C}^2\otimes \mathbb{C}^b$.
\end{theorem}

\begin{proof}[Proof of Theorem \protect\ref{theorem:degrees}]
	It suffices to prove that if a grid-labelled graph $G_{l}^{2,b}$ satisfies
	the degree criterion then $G_{l}^{2,b}\in \mathcal{S}$.
	We know from Lemma \ref{lemma:degreecriterioniff} that $G_l^{2,b}$ satisfies the degree criterion if and only if $D_{l}^{2,b}$ satisfies the degree criterion, where $D_{l}^{2,b}$ is from the HVD decomposition of $G_l^{2,b}$ and contains all diagonal edges of $G_l^{2,b}$. Hence, we need to show that if the degree criterion holds then $\rho (D_{l}^{2,b})\in\mathcal{S}$. 
	Up to normalisation, the structure of this density matrix is
	\begin{equation*}
	\rho (D_{l}^{2,b})=%
	\begin{pmatrix}
	\Delta ^{(1)} & A \\ 
	A^{T} & \Delta ^{(2)}%
	\end{pmatrix}%
	,
	\end{equation*}%
	where $\Delta ^{(i)}$, with $i=1,2$, are diagonal matrices encoding the
	degrees of the vertices in row $i$, 
	\begin{equation*}
	\Delta _{j,j}^{(i)}=d((i,j)),
	\end{equation*}%
	and the matrix $A$ encodes the diagonal edges of $G_{l}^{2,b}$. This can be seen from the fact that
	the adjacency matrix of any grid-labelled graph $G_{l}^{2,b}$ with diagonal
	edges \emph{only} is a $2\times 2$ symmetric block matrix with block diagonals
	equal to the zero matrix -- diagonal edges are incident to vertices in different rows by definition. The density matrix of the partial transpose of $D_{l}^{a,b}$
	is then 
	\begin{equation*}
	\rho (\Gamma (D_{l}^{a,b}))=%
	\begin{pmatrix}
	\Delta ^{(2)} & A^{T}\\ 
	A & \Delta ^{(1)}%
	\end{pmatrix}%
=	\begin{pmatrix}
	\Delta ^{(2)} & A \\ 
	A^T & \Delta ^{(1)}%
	\end{pmatrix}%
	.
	\end{equation*}%
	The degree criterion holds for a general grid-labelled graph $G_{l}^{a,b}$ if and only if $\text{diag}%
	(\rho (D_{l}^{a,b}))=\text{diag}(\rho (\Gamma (D_{l}^{a,b})))$, hence if the
	degree criterion holds for $D_l^{2,b}$ then $\Delta ^{(1)}=\Delta ^{(2)}$. Since density matrices are by definition positive semi-definite, by Theorem \ref{theorem:Ando} $\rho(D_l^{2,b})$ is separable.
\end{proof}

In proving this theorem we have uncovered the following interesting corollary, following from the requirement that the matrices $\Delta^{(1)}$ and $\Delta^{(2)}$ are equal for grid-labelled graphs of type $(2,b)$ that satisfy the degree criterion.

\begin{corollary}
	\label{corollary:rowdegrees}
Let $G_l^{2,b}$ be a grid-labelled graph. Then $G_l^{2,b}\in \mathcal{S}$ if and only if \begin{align*}d((1,j))=d((2,j))\end{align*} for $1\le j\le b$.
\end{corollary}
This means that grid-labelled graphs of type $(2,b)$ for $b\ge 2$ are separable if and only if the vertex degrees are equal for vertices in the same column.
\subsection{Extensions and second order local isomorphism}
We conclude this section by introducing a generalisation of local isomorphism, \emph{second order local isomorphism}, that can be used to generate infinite families of separable or entangled grid-labelled graphs. 
\begin{figure}[tbp]
	\centering
	\includegraphics[scale=0.2]{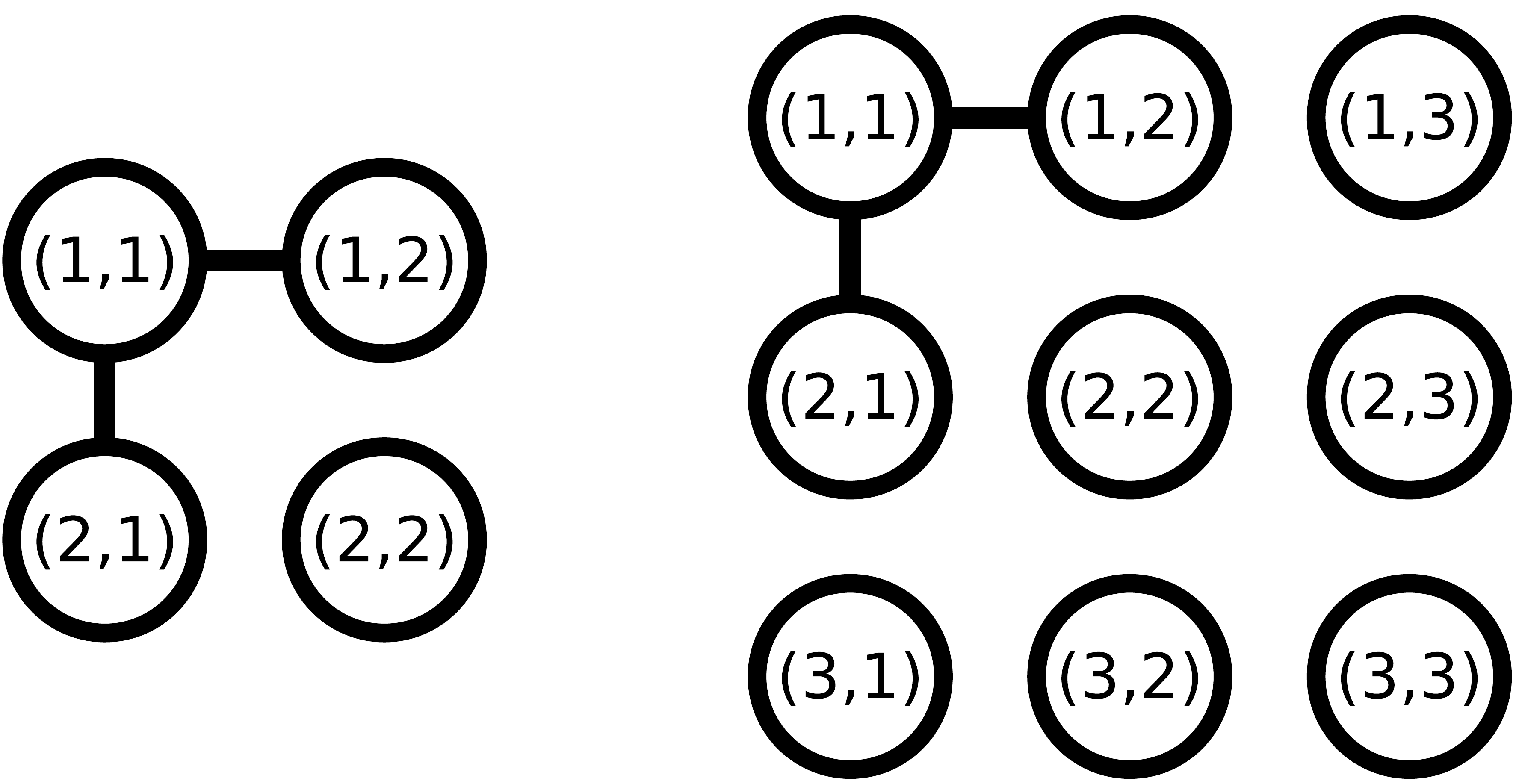}
	\caption{The grid-labelled graph on the right is an extension of that on the
		left.}
	\label{fig:extension}
\end{figure}

\begin{definition}[Grid-labelled graph extension]
	Let $G_{l}^{a,b}$ be a grid-labelled graph. A grid-labelled graph $H_{m}^{c,d}$
	is said to be an \emph{extension} of $G_{l}^{a,b}$ if $c\geq a$, $d\geq b$, $%
	E(G_{l}^{a,b})=E(H_{m}^{c,d})$, $V(H_{m}^{c,d})=V(G_{l}^{a,b}\uplus
	_{k}V(K_{1}))$, for $k=c\cdot d-a\cdot b$, and for all $v\in V(G)$, $%
	l(v)=m(v) $.
\end{definition}

\begin{example}
	\emph{In Figure \ref{fig:extension}, the graph on the right is an extension
		of the graph on the left, because it can be obtained by increasing the
		dimensions of the grid. }
\end{example}

\begin{definition}[Second order local isomorphism]
	Two grid-labelled graphs, $G_{l}^{a,b}$ and $H_{m}^{c,d}$, with $a\le c,b\le d$,
	are said to be \emph{second order locally isomorphic} if $G_{l}^{a,b}$ has an extension ${G}_{l^{\prime }}^{c,d}$, such that ${G}%
	_{l^{\prime }}^{c,d}\cong _{c,d}H_{m}^{c,d}$. This is denoted by $%
	G_{l}^{a,b}\cong _{c,d}H_{m}^{c,d}$.
\end{definition}

Note that second order local isomorphism is reflexive.
\begin{figure}[tbp]
	\centering
	\includegraphics[scale=0.2]{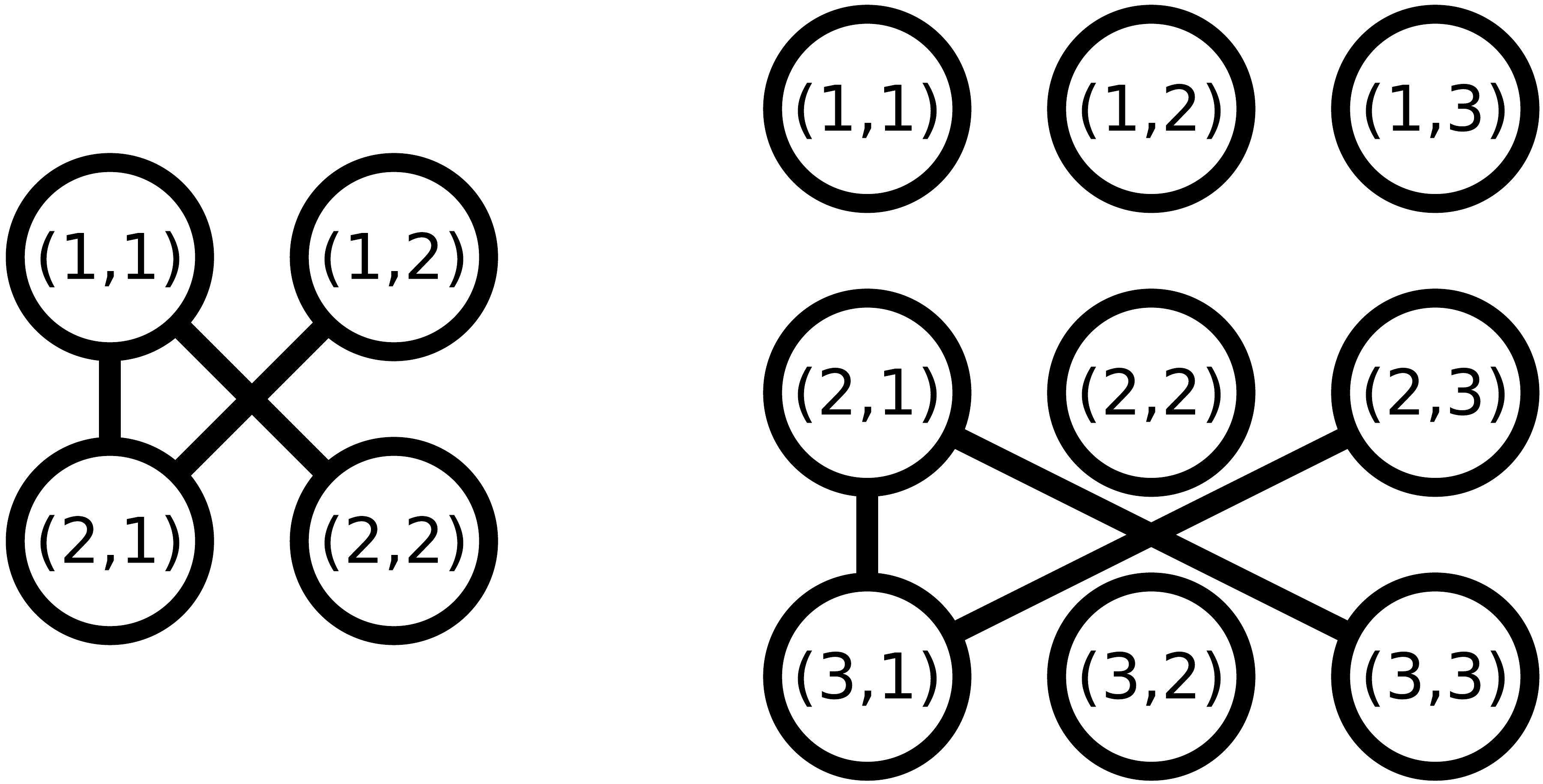}
	\caption{Two grid-labelled graphs that are second order locally isomorphic.}
	\label{fig:secondorderlocallyisomorphic}
\end{figure}

\begin{example}
	\emph{The two grid-labelled graphs in Figure \ref%
		{fig:secondorderlocallyisomorphic} are second order locally isomorphic. They
		can be obtained from one another by modifying the grid dimensions and
		permuting the vertex rows and columns.}
\end{example}

\begin{proposition}
	If two grid-labelled graphs are locally isomorphic then they are second order
	locally isomorphic. The converse is not necessarily true.
\end{proposition}

\begin{proof}
	If two grid-labelled graphs $G_{l}^{a,b}$ and $H_{m}^{c,d}$ are locally
	isomorphic, then $a=c$ and $b=d$. Trivially, the grid-labelled graphs are
	extensions of one another, and are hence second order locally isomorphic.
	
	The converse is not true precisely because $G_l^{a,b}\cong_{c,d}H_m^{c,d}$
	does not imply $a=c,b=d,$ which is required by the definition of local
	isomorphism.
\end{proof}

\begin{proposition}
	\label{proposition:extensionseparability} For a grid-labelled graph $%
	G_{l}^{a,b}$ with an extension  $H_{m}^{c,d}$, $G_{l}^{a,b}\in \mathcal{S}$ if
	and only if $H_{m}^{c,d}\in \mathcal{S}$.
\end{proposition}

\begin{proof}
	It is clear that adding isolated vertices will not affect separability of a
	grid-labelled graph.
\end{proof}

\begin{proposition}
	Let $G_{l}^{a,b}$ and $H_{m}^{c,d}$ be grid-labelled graphs. Let $%
	G_{l}^{a,b}\cong _{c,d}H_{m}^{c,d}$. Then, $G_{l}^{a,b}\in \mathcal{S}$
	if and only if $H_{m}^{c,d}\in \mathcal{S}$.
\end{proposition}

\begin{proof}
	If $a=c$ and $b=d$ then $G_{l}^{a,b}$ and $H_{m}^{c,d}$ are locally
	isomorphic, and the result follows from Proposition \ref%
	{proposition:localisomorphism}. Therefore it suffices to prove that
	extension will not affect separability, which we have proved already in
	Proposition \ref{proposition:extensionseparability}.
\end{proof}

In the next section we will consider decompositions of grid-labelled graphs in greater depth.
\subsection{Decompositions}
\label{subsection:decompositions}
A further sufficient condition for separability is based on decompositions:

\begin{theorem}[Separable Decompositions]
	\label{theorem:decomposition}If there exists a decomposition $X$ of $%
	G_{l}^{a,b}$ such that for all $H_{l}^{a,b}\in X$ we have $%
	H_{l}^{a,b}\in \mathcal{S}$, then $G_{l}^{a,b}\in \mathcal{S}$.
\end{theorem}

\begin{proof}
	We know that 
	\begin{align*}
	\rho(G_l^{a,b})&=\frac{1}{2|E(G_l^{a,b})|}L(G_l^{a,b})\\
	&=\frac{1}{2|E(G_l^{a,b})|}\sum_{H_l^{a,b}\in X}L(H_l^{a,b})\\
	&=\frac{1}{2|E(G_l^{a,b}))}\sum_{H_l^{a,b}\in X}2|E(H_l^{a,b})|\cdot\rho(H_l^{a,b})\\
	&=\sum_{H\in X}\frac{|E(H_l^{a,b})|}{|E(G_l^{a,b})|}\rho(H_l^{a,b}).
	\end{align*}
	If for all $H_l^{a,b}\in X$, $H_{l}^{a,b}\in \mathcal{S}$ then
	a convex combination of the form given by Equation (\ref{eq:separable}) can be formed by setting $p_{i}=|E(H)|/|E(G)|$. Therefore, $G_{l}^{a,b}\in \mathcal{S}$.
\end{proof}

\begin{definition}[Pair-symmetric grid-labelled graphs]
	A grid-labelled graph is said to be \emph{pair-symmetric} if each of its
	diagonal edges $\{(i,j),(k,l)\}$ have a \emph{counterpart} edge $\{(k,j),(i,l)\}$. An edge and its counterpart are referred to as a 
	\emph{counterpart pair}.
\end{definition}

The grid-labelled graphs in Figures \ref{fig:secondorderlocallyisomorphic}, \ref{fig:sep} and \ref{fig:sepeasy} are pair-symmetric.

\begin{proposition}
	Every pair-symmetric grid-labelled graph is separable.
\end{proposition}

\begin{proof}
	Let $G_{l}^{a,b}$ be a pair-symmetric grid-labelled graph with $k$
	counterpart pairs. Let $(H_{i})_{l}^{a,b}$ for $1\le i\le k$ be the
	subgraph of $G_{l}^{a,b}$ containing only the edges of the $i^{\text{th}}$
	counterpart pair. Let $(H_{k+1})_{l}^{a,b}$ be the subgraph of $G_{l}^{a,b}$
	containing the remaining edges of $G_{l}^{a,b}$, that is, those edges not part of the list of $k$ counterpart pairs.
\begin{figure}[ht!]
\centering
\includegraphics[scale=0.2]{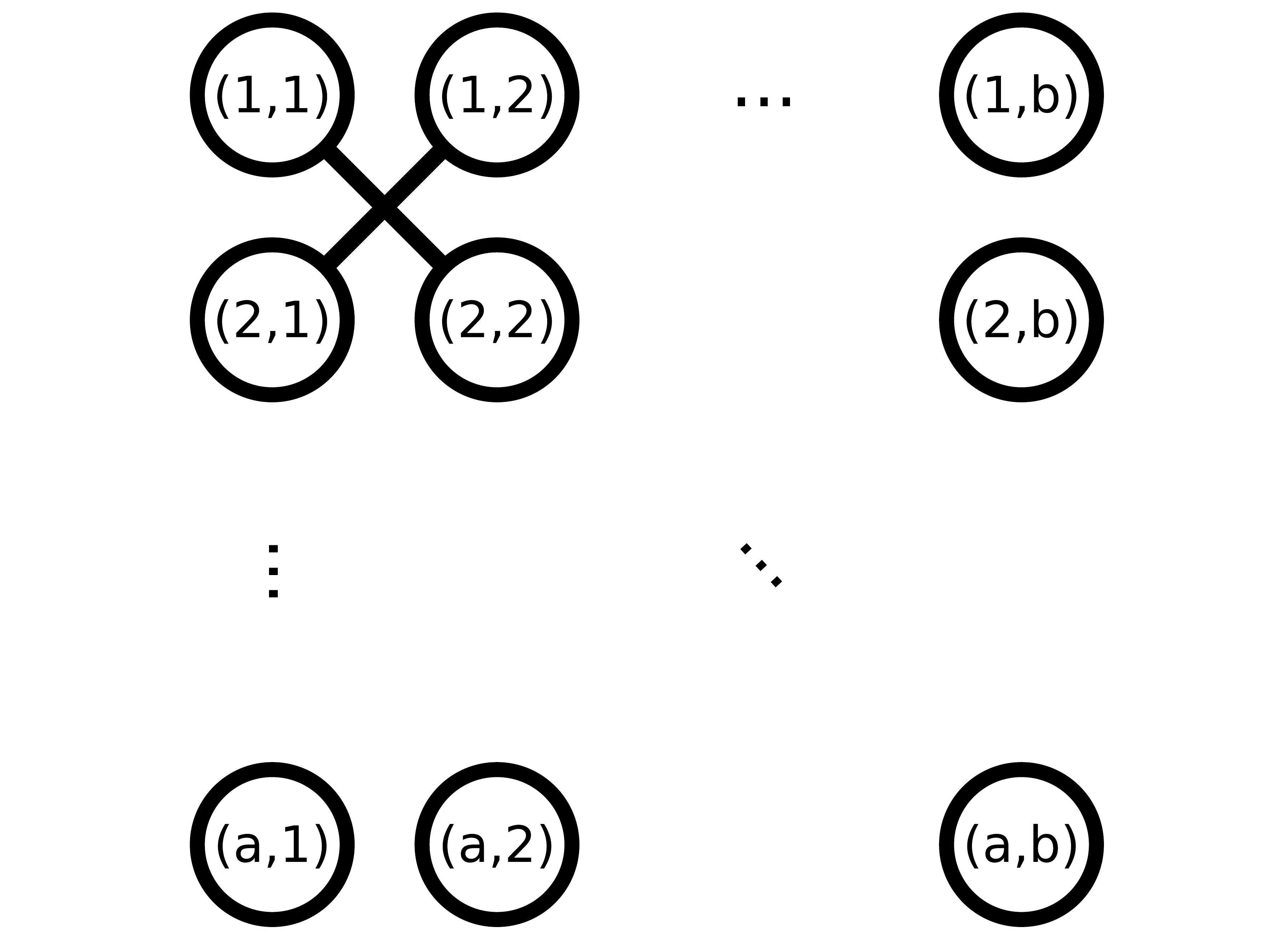}
\caption{A $a\times b$ graph with a cross.}
\label{fig:sep}
\end{figure}
\begin{figure}[ht!]
\centering
\includegraphics[scale=0.085]{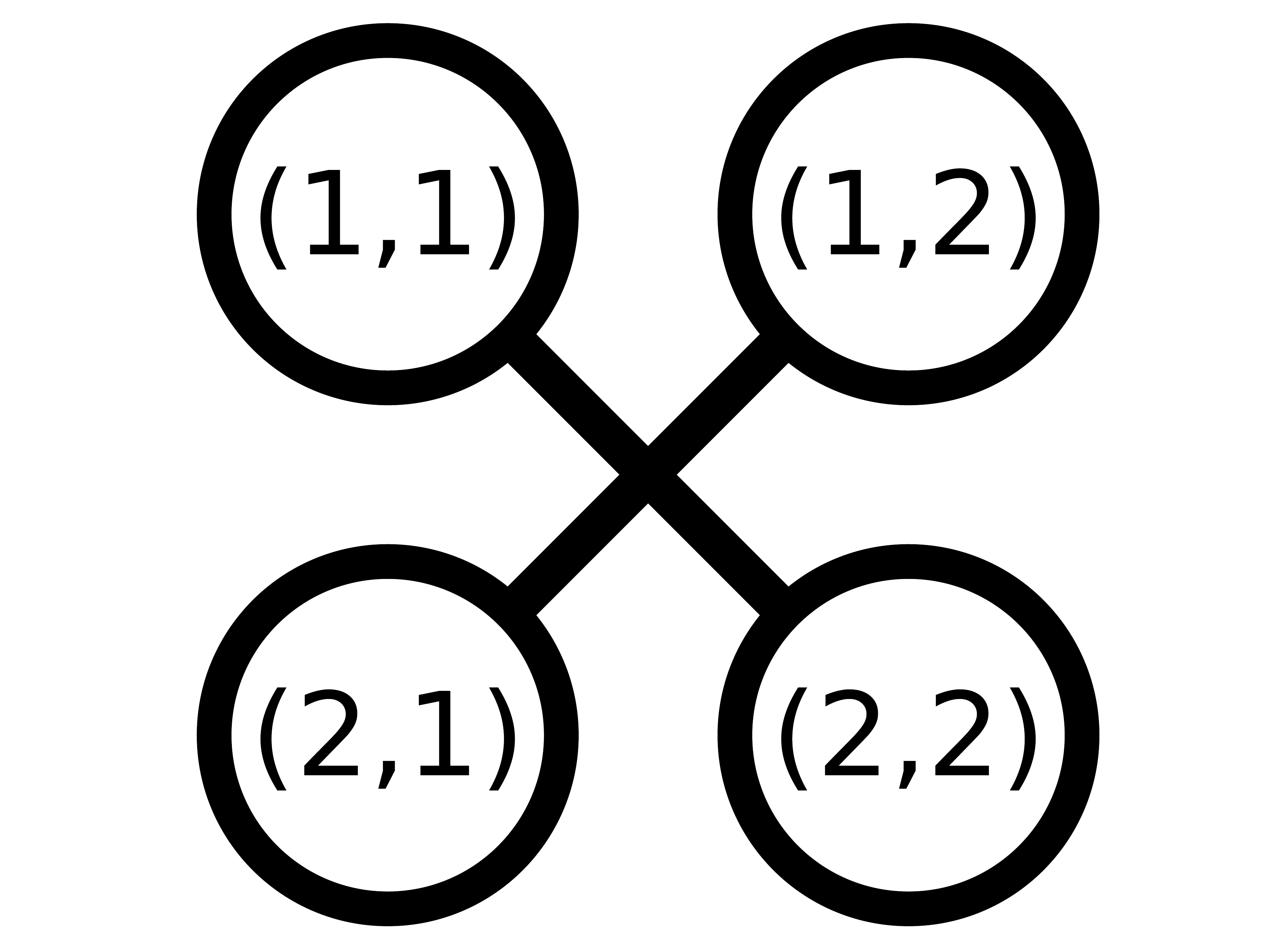}
\caption{A $2\times 2$ graph with a cross.}
\label{fig:sepeasy}
\end{figure}
	 Clearly, the set $%
	X=\{H_{1},H_{2},...,H_{k+1}\}$ forms a decomposition of $G_{l}^{a,b}$. To
	prove the proposition it suffices to show that for each $H_{l}^{a,b}\in
	X,H_{l}^{a,b}\in \mathcal{S}$. It is obvious that for $1 \le i\leq k$, $%
	(H_{i})_{l}^{a,b}\cong _{a,b}S_{l}^{a,b}$, where $S_{l}^{a,b}$ is shown in
	Figure \ref{fig:sep}. This graph is an extension of $S_{l}^{2,2}$, shown in
	Figure \ref{fig:sepeasy}. The graph $S_{l}^{2,2}$ is invariant under partial transpose, and so satisfies the degree criterion. Since the degree criterion is necessary and sufficient for separability in this case ($a=b=2$), we know that $S_{l}^{2,2}\in\mathcal{S}$.
	Hence, for $1\le i\leq k,S_{l}^{2,2}\cong _{a,b}(H_{i})_{l}^{a,b}$, and so $%
	(H_{i})_{l}^{a,b}\in \mathcal{S}$.
	
	The last component of the decomposition $X$, $(H_{k+1})_{l}^{a,b}$, has only
	horizontal and vertical edges (since by definition, all diagonal edges in a pair-symmetric graph are involved in a counterpart pair). Hence from Corollary \ref{corollary:horizontalandverticaledgesonly}, $(H_{k+1})_{l}^{a,b}\in \mathcal{S}$, and so $%
	X$ is a separable decomposition of $G_{l}^{a,b}$. The proposition thus
	follows via Theorem \ref{theorem:decomposition}.
\end{proof}

The pair-symmetric grid-labelled graphs belong to a larger family of grid-labelled graphs, which we call \emph{stratified}. Satisfaction of the degree criterion is a necessary and sufficient condition for separability of stratified grid-labelled graphs. Furthermore, such grid-labelled graphs exist in all bipartite dimensions.

\begin{definition}[Stratified grid-labelled graphs]
	A grid-labelled graph is called \emph{row (resp. column) stratified} if for all of its diagonal edges $%
	\{(i,j),(k,l)\}$, $|i-k|=1$ (resp. $|j-l|=1$).
\end{definition}
\begin{figure}[ht!]
	\centering
	\includegraphics[scale=0.22]{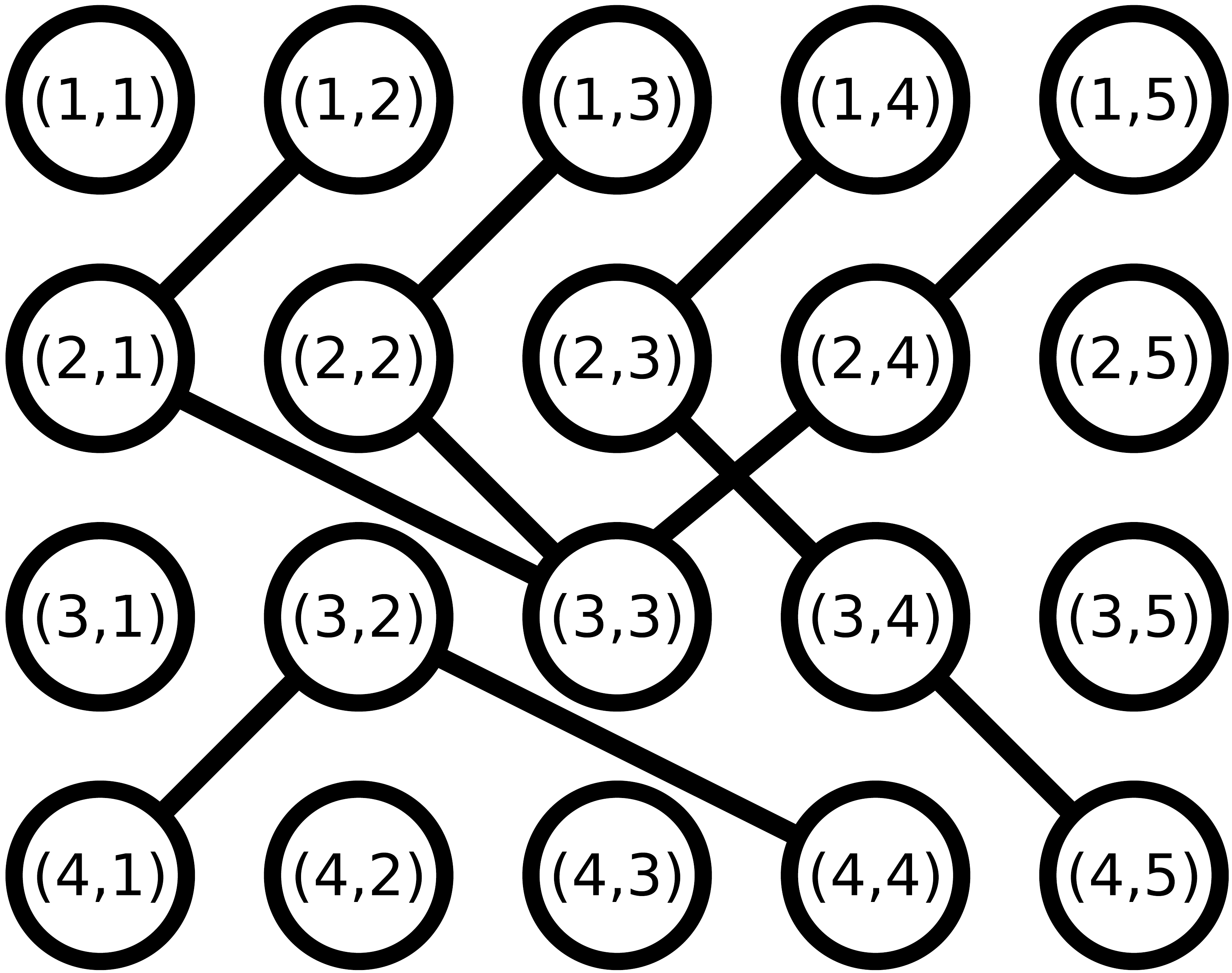}
	\caption{A row stratified graph}
	\label{fig:stratified}
\end{figure}

\begin{example}
	\emph{Figure \ref{fig:stratified} illustrates a row stratified grid-labelled
		graph. The endpoints of all diagonal edges are restricted to nearest
		neighbour rows of vertices, hence the name ``stratified''.}
\end{example}

The following lemma is obvious.

\begin{lemma}[Strata decomposition]
	\label{lemma:stratadecomposition} Let $G_l^{a,b}$ be a row (\emph{resp.} column) stratified
	grid-labelled graph. Let $D_l^{a,b}$ be the diagonal part of the HVD
	decomposition of $G_l^{a,b}$. Then, $D_l^{a,b}$ has a decomposition $%
	S=\{(S_i)_l^{a,b}\}$ for $1\le i<b$, where $(S_i)_l^{a,b}$ is the
	subgraph of $D_l^{a,b}$ containing the edges $\{(p,q),(r,s)\}\in E(D_l^{a,b})$
	with $p=i$ and $r=i+1$ (\emph{resp.} $q=i$ and $s=i+1$). We call $S$ the \emph{strata decomposition} of $G_l^{a,b}$.
\end{lemma}
This can be used to prove the following.
\begin{lemma}
	\label{lemma:stratadegree} Let $G_l^{a,b}$ be a stratified grid-labelled
	graph. If $G_l^{a,b}$ satisfies the degree criterion then each element of
	its strata decomposition satisfies the degree criterion.
\end{lemma}

\begin{proof}
	We prove the lemma for row stratified graphs. The same argument holds for column stratified graphs if all references to ``rows'' are replaced by ``columns''. In what follows, we refer to graphs that satisfy the degree criterion as ``DC'', and those that do not as ``non-DC''.
	
	By definition, the elements of a strata decomposition of a graph $G_l^{a,b}$ are second order locally isomorphic to graphs of type $(2,b)$ -- that is, we are only interested in the vertices in the $2$ ``occupied rows'' of the strata. By Corollary \ref{corollary:rowdegrees} if a grid-labelled graph of type $(2,b)$ is non-DC, then the degrees of the vertices in the upper row do not match those of the lower row. Therefore, if a single element $(S_i)_l^{a,b}$ of the strata decomposition of $G_l^{a,b}$ (in this case, corresponding to the subgraph induced by vertices in row $i$ and $i+1$ of $G_l^{a,b}$) is non-DC, then the degree of at least $1$ vertex in these rows will change after partial transpose. If this happens, then of course $G_l^{a,b}$ is non-DC.
	
	To finish the proof of this lemma, we must show that if more than $1$ element of the strata decomposition of $G_l^{a,b}$ is non-DC then $G_l^{a,b}$ is non-DC. This must be considered because conceivably the degree changes in $2$ or more non-DC strata could cancel out in some way, leaving $D(G_l^{a,b})$ equal to $D(\Gamma(G_l^{a,b}))$.
	
	The only way two strata degree changes could cancel out would be if two non-DC strata shared a row of $G_l^{a,b}$. That is, if the strata $(S_i)_l^{a,b}$ and $(S_{i+1})_l^{a,b}$ were non-DC. However, if this were the case then rows $i$ and $i+2$ would have a partial transpose degree change. By the same argument, no additional non-DC strata can be selected either side of these to cancel out the degree changes. The lemma follows.
\end{proof}

\begin{theorem}
\label{theorem:stratified}
	A stratified grid-labelled graph $G_l^{a,b}$ is separable if and only if it
	satisfies the degree criterion.
\end{theorem}
\begin{proof}
	The degree criterion has been proved to be necessary for separability, hence
	we must only prove sufficiency.
	
	From Lemma \ref{lemma:stratadecomposition}, for any stratified grid-labelled
	graph $G_l^{a,b}$ there exists a decomposition $S=\{(S_i)_l^{a,b}\}$ for $%
	1\le i<b$. Each of the $(S_i)_l^{a,b}$ are second order
	locally isomorphic to a graph $(S_i)_l^{2,b}$ -- informally, we can discard each isolated vertex. If the degree criterion holds
	for $G_l^{a,b}$, then by Lemma \ref{lemma:stratadegree} it holds for each $%
	(S_i)_l^{2,b}$. Hence, by sufficiency of the degree criterion for separability
	for grid-labelled graphs with $a=2$, each $(S_i)_l^{2,b}$ is separable, and
	the strata decomposition $S$ is a separable decomposition of $G_l^{a,b}$.
	Separability of $G_l^{a,b}$ follows from Theorem \ref{theorem:decomposition}.
\end{proof}
Note that the work of Wu \cite{Wu2006} explores separability of matrices with line-sum symmetric blocks, and proves results of the same flavour to what we do in this section. Perhaps there is a link between grid-labelled graphs that are locally isomorphic to stratified grid-labelled graphs and combinatorial Laplacian matrices with line-sum symmetric blocks.

In this section we have explored separability criteria for grid-labelled graphs in detail. In the next section, we use these techniques to classify all grid-labelled graphs of type $(3,3)$ that satisfy the degree criterion. In some cases we are able to go further and identify certain families of separable grid-labelled graphs.
\section{Separability in $3\times 3$}
\label{section:separabilityin3x3}
Let us investigate separability for grid-labelled graphs of type $(3,3)$. We will attempt to go as far as possible using only the degree criterion. As previously seen, only the diagonal edges of a graph are relevant for satisfying the degree criterion (see Lemma \ref{lemma:degreecriterioniff}). To keep things simple, we will not consider any graphs with horizontal or vertical edges. We will need to further classify diagonal edges into \emph{uphill} and \emph{downhill} edges. Uphill edges are those which travel from bottom-left vertices to top-right vertices, and downhill edges are those which travel from top-left to bottom-right. The grid-labelled graph in Figure \ref{fig:uphilldownhill} has $3$ downhill edges and $2$ uphill edges. Here is a formal definition.
\begin{definition}[Uphill and downhill edges]
	Let $G_l^{a,b}$ be a grid-labelled graph. Let $\{(i,j),(k,l)\}\in E(G_l^{a,b})$ be a diagonal edge belonging to $G_l^{a,b}$. We say that the edge is \emph{uphill} if $\text{sgn}(i-k)\neq\text{sgn}(j-l)$, and \emph{downhill} if $\text{sgn}(i-k)=\text{sgn}(j-l)$.
\end{definition}
\begin{figure}[ht!]
\centering
\includegraphics[scale=0.15]{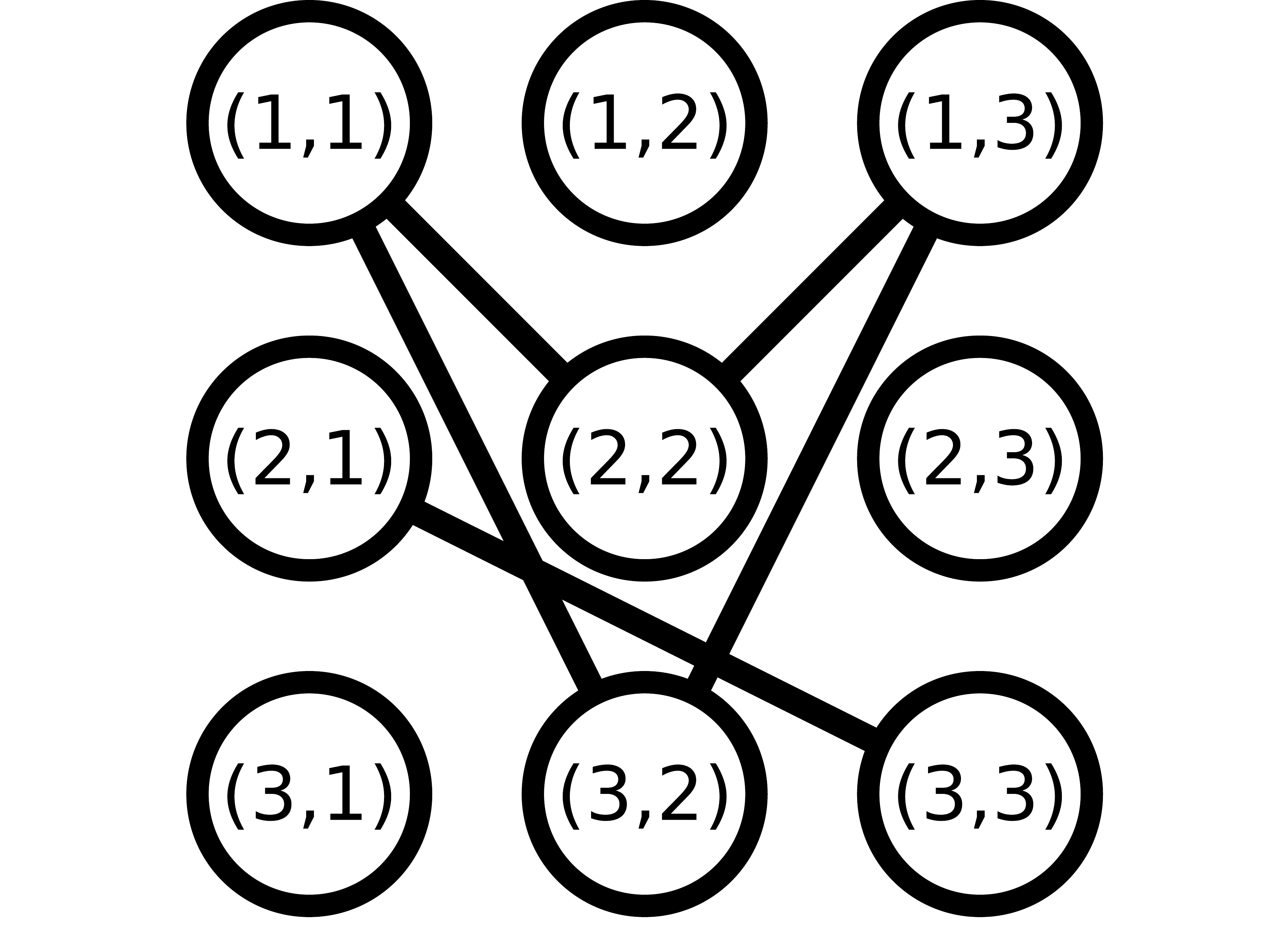}
\caption{A grid-labelled graph with $2$ uphill and $3$ downhill diagonal edges.}
\label{fig:uphilldownhill}
\end{figure}
In order for a grid-labelled graph to satisfy the degree criterion, the degrees of its vertices must not change after the partial transpose operation. Each edge of the graph contributes $1$ to the degrees of two vertices, $v_1$ and $v_2$. If the edge is diagonal, then after the partial transpose operation, this edge will contribute $1$ to the degrees of two different vertices, $w_1\neq v_1$ and $w_2\neq v_2$. 

Let $G_l^{3,3}$ be a grid-labelled graph with a single diagonal edge, $\{(1,1),(3,3)\}$. This graph does not satisfy the degree criterion. Indeed, there is no grid-labelled graph with a single diagonal edge that satisfies the degree criterion. Suppose we want to make $G_l^{3,3}$ satisfy the degree criterion by adding a single diagonal edge. In order for this to be the case, the new edge must be placed such that after partial transpose, the vertices $(1,1)$ and $(3,3)$ have degree $1$. Also, the vertices to which the edge moves to after partial transpose, $(1,3)$ and $(3,1)$, must have degree $1$ \emph{before} partial transpose. The only edge that can be added such that the graph will have these two properties is easily seen to be the edge $\{(1,3),(3,1)\}$.

This reasoning about pre and post-transpose degree properties can be generalised to grid-labelled graphs with more edges via the concept of \emph{edge contributions}.
\subsection{Edge contributions}
\begin{definition}[Edge contribution matrix]
	Let $G_l^{a,b}$ be a grid-labelled graph with edge set $E$. For each diagonal edge $\{(i,j),(k,l)\}\in E,$ define its $a\times b$ \emph{edge contribution matrix} $A_{(i,j),(k,l)}$ such that
	\begin{align*}
	[A_{(i,j),(k,l)}]_{pq}:=\begin{cases}+1~&\text{ if }p=i,q=j\text{ or } p=k,q=l \text{;}\\-1~&\text{ if } p=i,q=l \text{ or } p=k,q=j;\\0 ~&\text{ otherwise.}\end{cases}
	\end{align*}
\end{definition} 

\begin{definition}[Graph contribution matrix]
	Let $G_l^{a,b}$ be a grid-labelled graph with diagonal edge set $D\subseteq E(G_l^{a,b})$. The \emph{contribution matrix} of $G_l^{a,b}$ is defined
	\begin{align}
	\label{eq:contributionSum}
	C(G_l^{a,b})=\sum_{\{(i,j),(k,l)\}\in D}A_{(i,j),(k,l)}.
	\end{align}
\end{definition}
The contribution matrix of a grid-labelled graph encodes the pre and post-transpose degree contributions of all its edges. If the partial transpose changes the degree of a vertex, then this means that there is some edge whose pre and post-transpose degree contribution do not match and cancel out. Hence there is a non-zero component somewhere in the contribution matrix of the graph. The following lemma is easily seen from this line of argument.

\begin{lemma}
	\label{lemma:zerocontribution}
	Let $G_l^{a,b}$ be a grid-labelled graph with contribution matrix $C(G_l^{a,b})$. Then $G_l^{a,b}$ satisfies the degree criterion if and only if $C(G_l^{a,b})$ is equal to the $a\times b$ matrix with all entries equal to $0$.
\end{lemma}
\begin{proof}
	Let $G_l^{a,b}$ be a grid-labelled graph. Without loss of generality we may assume that it contains only diagonal edges. From the definition of the contribution matrix of a graph it is clear that for each vertex $(i,j)\in V(G_l^{a,b})$ with degree $d$, there will be $d$ edge contribution matrices in the sum in Equation (\ref{eq:contributionSum}) that have $+1$ for their $ij^{\text{th}}$ element.
	
	If $G_l^{a,b}$ satisfies the degree criterion then for each vertex $(i,j)\in V(G_l^{a,b})$ with degree $d$ there will be $d$ edges $\{(i,l),(k,j)\}\in E(G_l^{a,b})$ such that after partial transpose, an endpoint of each edge will be the vertex $(i,j)$. By definition, the edge contribution matrices of these edges will have $-1$ in their $ij^{\text{th}}$ entry. Hence, the $ij^{\text{th}}$ entry of $C(G_l^{a,b})$ is zero. This is true for all entries of $C(G_l^{a,b})$.
	
	For the other direction, if $C(G_l^{a,b})$ is equal to the zero matrix then for each element of the edge contribution matrices in the sum in Equation (\ref{eq:contributionSum}) there will be an equal number of positive and negative entries. This means that the degrees of the vertices of the graph are the same before and after the partial transpose, by definition of the graph contribution matrix.
\end{proof}
\begin{figure}[h!]
	\centering
	\includegraphics[scale=0.15]{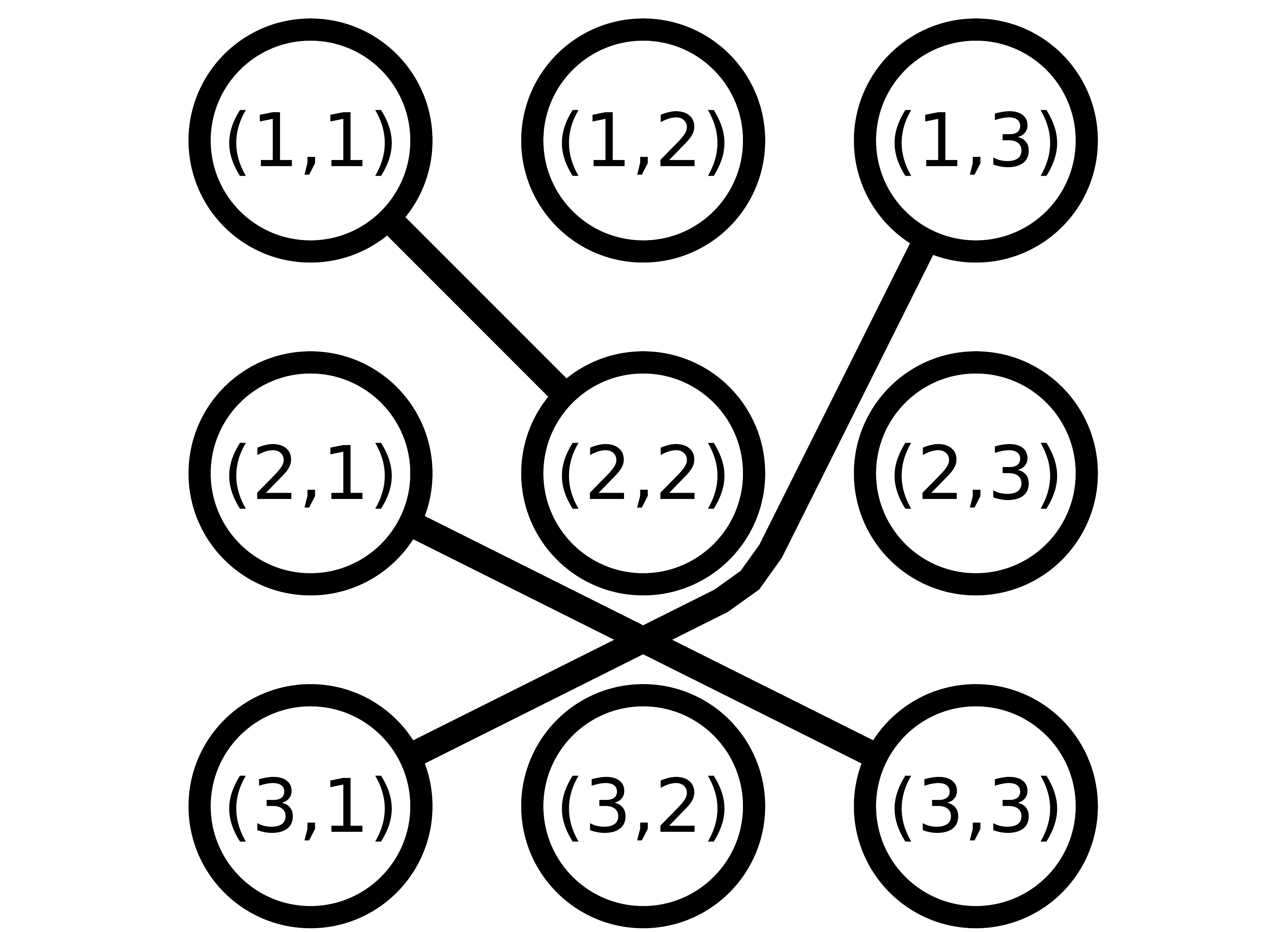}
	\caption{A grid-labelled graph that does not satisfy the degree criterion. Note that its contribution matrix is not equal to the zero matrix.}
	\label{fig:exampleBoardGraph}
\end{figure}
\begin{example}
\label{example:nicetableexample}
	\emph{The contribution matrix of a graph $G_l^{3,3}$ with three edges, $\{(1,1),(2,2)\}, \{(3,1),(1,3)\}$ and $\{(2,1),(3,3)\}$, as illustrated in Figure \ref{fig:exampleBoardGraph}, is} 
	\textbf{\begin{align*}
		C(G_l^{a,b})&=A_{(1,1),(2,2)}+A_{(3,1),(1,3)}+A_{(2,1),(3,3)}\\
		&=
		\begin{pmatrix}
		+1&-1&0\\
		-1&+1&0\\
		0&0&0\\
		\end{pmatrix}
		+
		\begin{pmatrix}
		-1&0&+1\\
		0&0&0\\
		+1&0&-1
		\end{pmatrix}
		+
		\begin{pmatrix}
		0&0&0\\
		+1&0&-1\\
		-1&0&+1
		\end{pmatrix}\\
		&=
		\begin{pmatrix}
		0&-1&+1\\
		0&+1&-1\\
		0&0&0
		\end{pmatrix}.
		\end{align*}}
	\emph{It is clear from the fact that the contribution matrix $C(G_l^{3,3})$ is non-zero that the graph $G_l^{3,3}$ does not satisfy the degree criterion. However, it is easily checked that adding the edge $\{(1,3),(2,2)\}$ causes the contribution matrix to become equal to the zero matrix.}
\end{example}

Let us now discuss a method of pictorially representing the contribution matrix of a graph: \emph{contribution tables}. 
\begin{definition}[Contribution table]
	Let $C$ be an $a\times b$ edge contribution matrix of dimension $a \times b$. The \emph{contribution table} of $C$ is an $a\times b$ grid, whose cells are populated with diagonal dashes running from top left to bottom right (\emph{down} dashes), or bottom left to top right (\emph{up} dashes).
	
	The dash placement on the table corresponding to $C$ is as follows.
	If $C_{ij}=+1$ then place a down dash in the corresponding grid square.
	If $C_{ij}=-1$ then place an up dash in the corresponding grid square.
\end{definition}
We call dashes in the same grid square but in different directions \emph{complementary}.
Two complementary dashes are called a \emph{cross} (because the drawing looks like an `X'). If a cell contains a dash that can not be uniquely paired with a complementary dash, then that dash is called \emph{unmatched}. The \emph{addition} of two contribution tables is the contribution table with all dashes from both tables. The contribution table of a graph is the addition of the contribution tables of each of its edges.

\newcommand{\boardscale}{0.2}
\begin{example}
	\emph{The contribution table of the grid-labelled graph $G_l^{3,3}$ defined in Example \ref{example:nicetableexample} and illustrated in Figure \ref{fig:exampleBoardGraph} is found to be equal to}
	\begin{align*}
	\vcenter{\hbox{\includegraphics[scale=\boardscale]{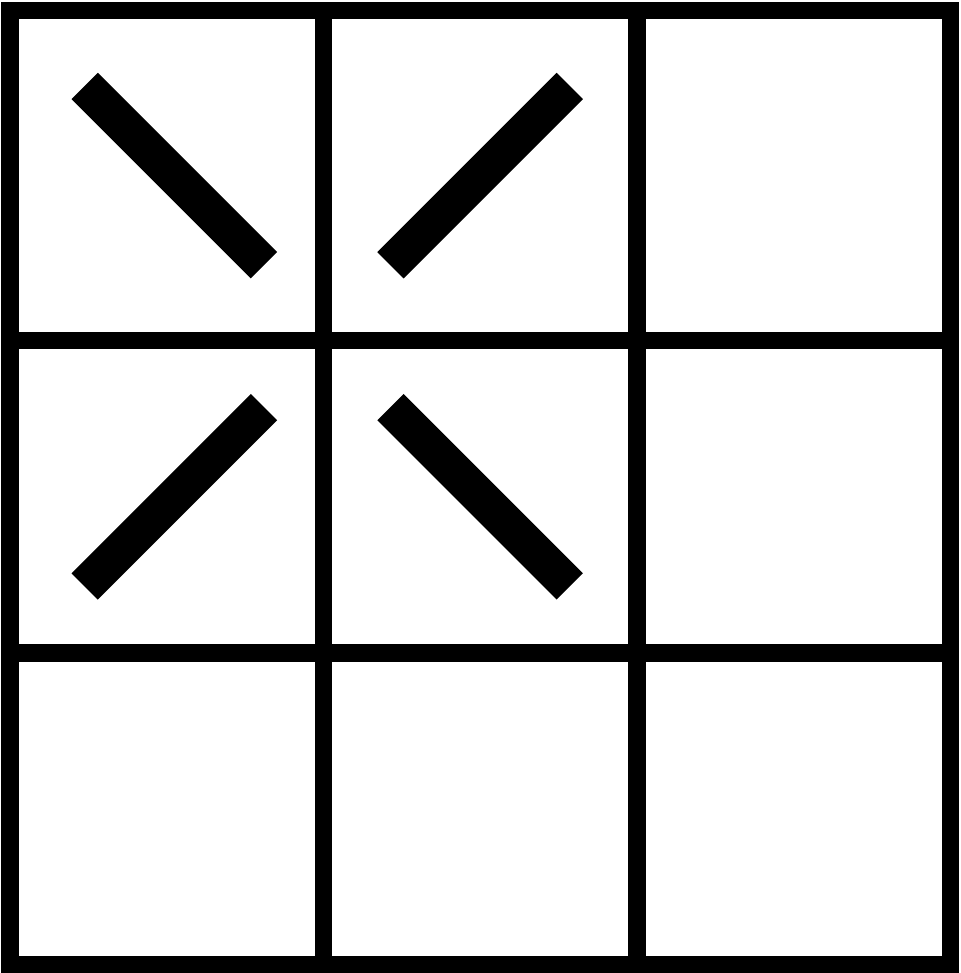}}}+\vcenter{\hbox{\includegraphics[scale=\boardscale]{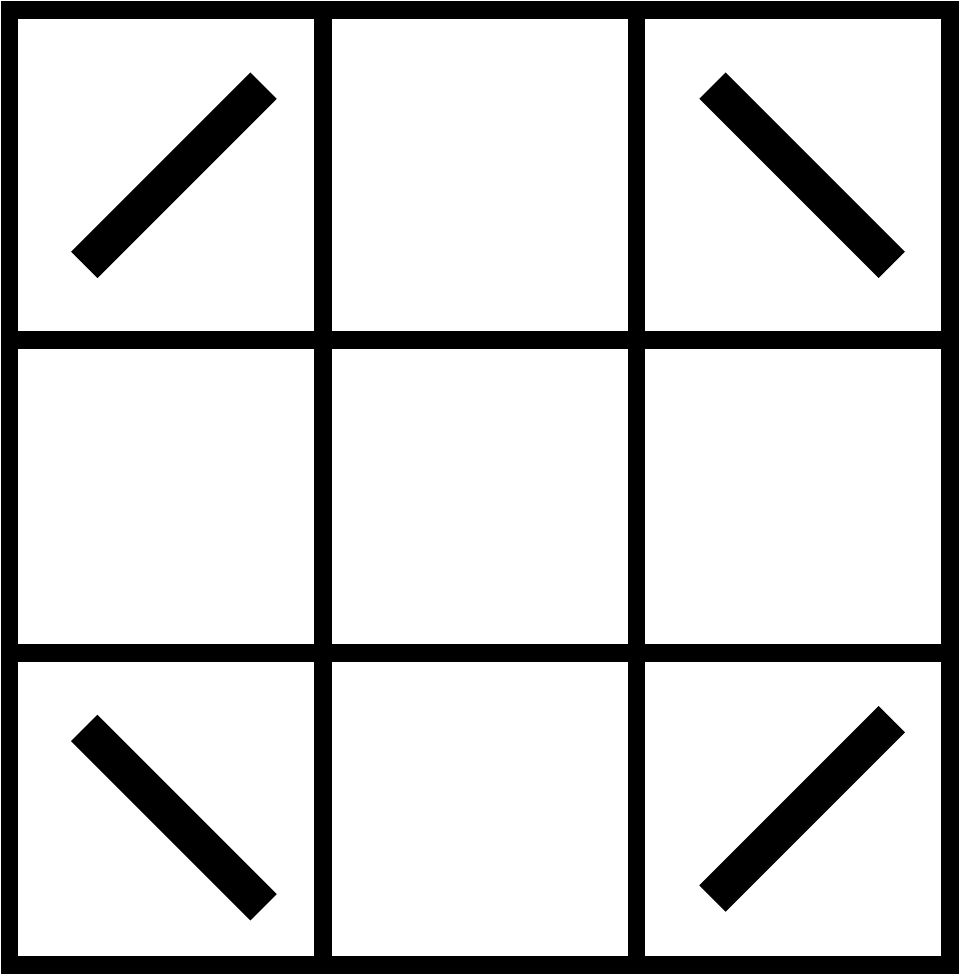}}}+\vcenter{\hbox{\includegraphics[scale=\boardscale]{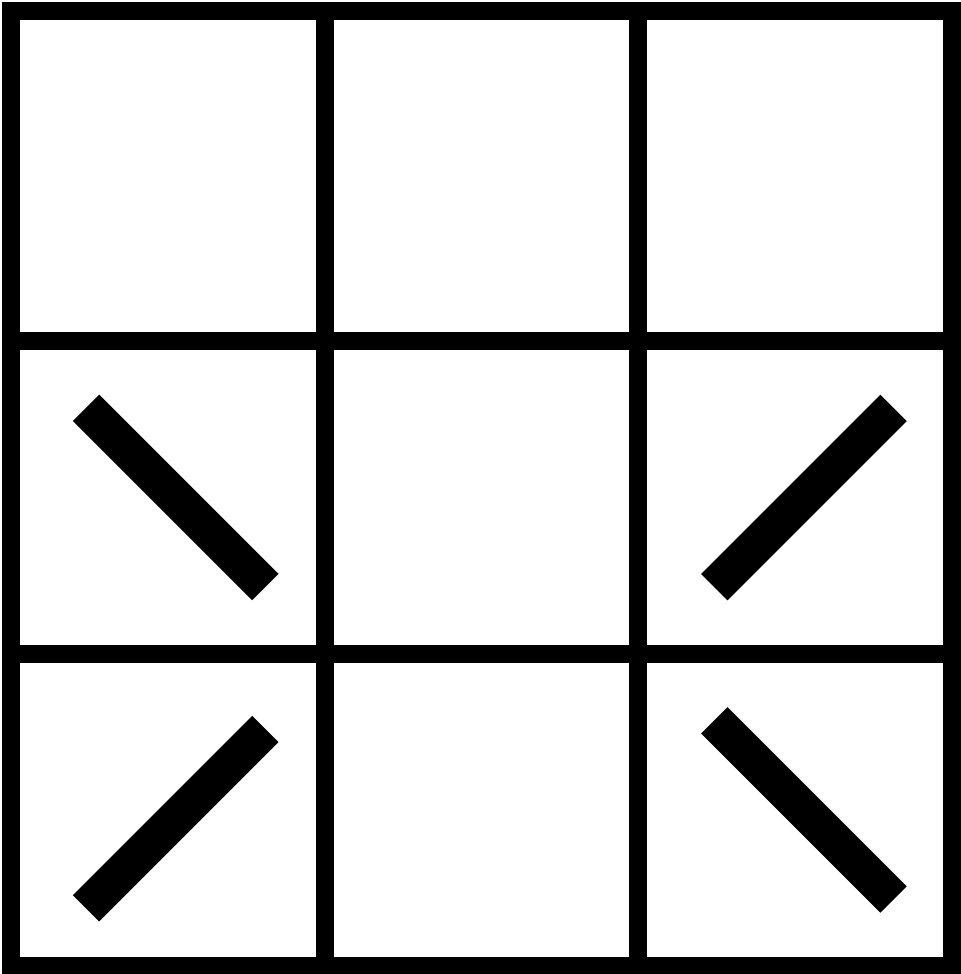}}}=\vcenter{\hbox{\includegraphics[scale=\boardscale]{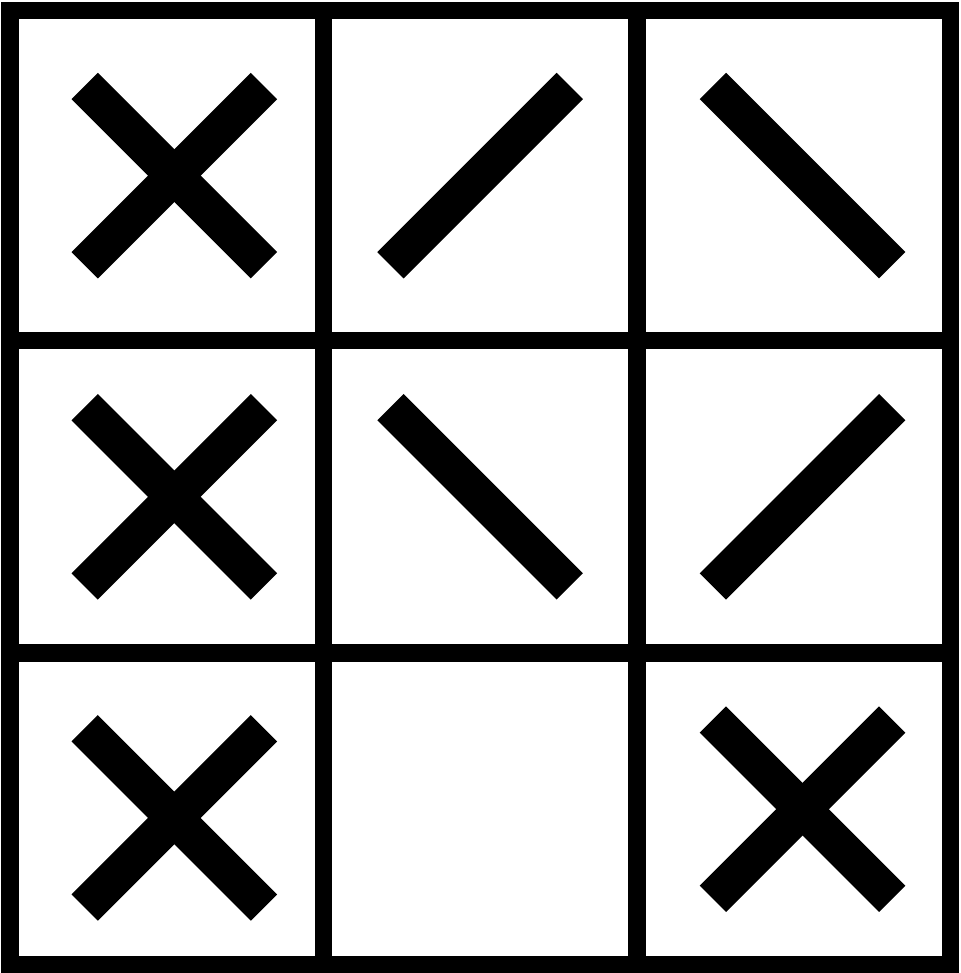}}}.
	\end{align*}
	\emph{The contribution table of $G_l^{3,3}$ has four crosses and four unmatched dashes.}
\end{example}

The next lemma follows directly from Lemma \ref{lemma:zerocontribution}.
\begin{lemma}
	Let $G_l^{a,b}$ be a grid-labelled graph. Then $G_l^{a,b}$ satisfies the degree criterion if and only if its contribution table contains no unmatched dashes.
\end{lemma}
\subsection{Building-block grid-labelled graphs}
We have introduced edge contribution matrices and tables in order to find $3\times 3$ grid-labelled graphs that satisfy the degree criterion. We shall see that such grid-labelled graphs can be obtained from a small set of grid-labelled graphs which we call \emph{building-blocks}.
\begin{figure}[ht!]
	\centering
	\includegraphics[scale=0.3]{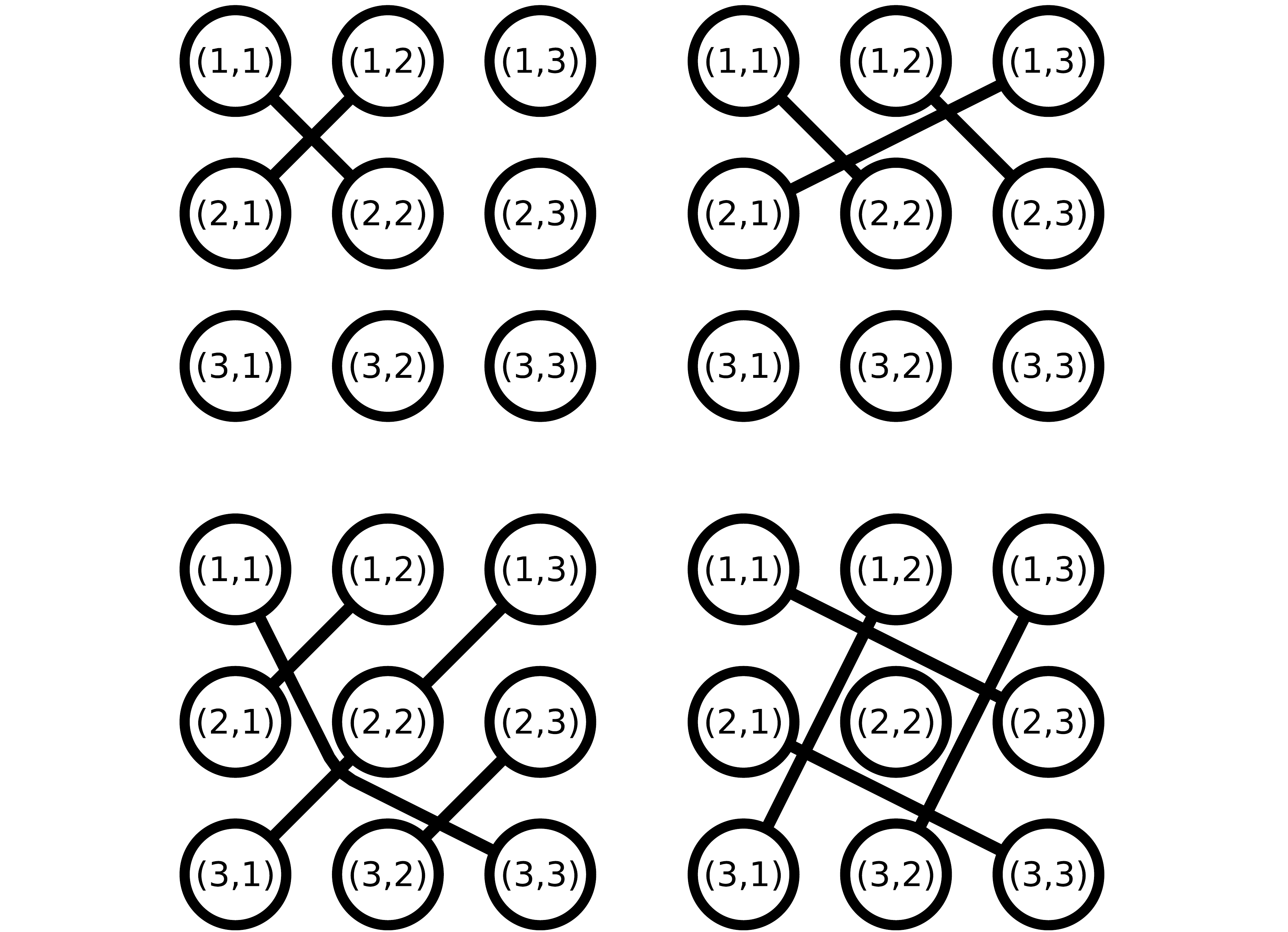}
	\caption{Clockwise from top left, the graphs $(B_2)_l^{3,3}, (B_3)_l^{3,3}, (B_4)_l^{3,3},$ and $ (B_5)_l^{3,3}$.}
	\label{fig:b2graph}
\end{figure}
\begin{definition}[Building-blocks]
	The following grid-labelled graphs are illustrated in Figure \ref{fig:b2graph}.
	\begin{itemize}
		\item The \emph{criss-cross} $(B_2)_l^{3,3}$ has two edges: $%
		\{(1,1),(2,2)\}$ and $\{(1,2),(2,1)\}$.
		
		\item The \emph{tally} $(B_3)_l^{3,3}$ has three edges  $\{(1,1),(2,2)\}$, $\{(1,2),(2,3)\}$ and $\{(2,1),(1,3)\}$.
		\item The \emph{cross-hatch} $(B_4)_l^{3,3}$ has four edges: $%
		\{(1,1),(2,3)\}$, $\{(2,1),(3,3)\}$, $\{(1,2),(3,1)\}$ and $\{(1,3),(3,2)\}$.
		\item The \emph{skew-mesh} $(B_5)_l^{3,3}$ has five edges: $%
		\{(1,1),(3,3)\}$, $\{(1,2),(2,1)\}$ and $\{(1,3),(2,2)\}$, $\{(2,2),(3,1)\}$%
		, $\{(2,3),(3,2)\}$.
	\end{itemize}
\end{definition} 

From our exploration so far of the contribution table framework, we can observe that the only type $(3,3)$ grid-labelled graphs with $2$ diagonal edges satisfying the degree criterion are locally isomorphic to the graph $(B_2)_l^{3,3}$, \emph{i.e.}, the $3\times 3$ criss-cross. In the next section we will show that with the use of an additional lemma, this can be proved formally.

\subsection{$3\times 3$ degree criterion with $2$ diagonal edges}
\begin{lemma}
	\label{lemma:matchedBoards}
	Let $G_l^{a,b}$ be a grid-labelled graph with $m$ edges. If $G_l^{a,b}$ satisfies the degree criterion, then its contribution table will contain exactly $2m$ crosses. 
\end{lemma}
\begin{proof}
	Each edge in a grid-labelled graph contributes $4$ dashes to the contribution table, so there are $4m$ dashes on the table. If a grid-labelled graph satisfies the degree criterion, then each dash must be matched, meaning there are $4m/2=2m$ crosses.
\end{proof}

\begin{lemma}
	\label{lemma:twodiagonaledges}
	Let $G_l^{3,3}$ be a grid-labelled graph with $2$ diagonal edges. Then $G_l^{3,3}
	$ satisfies the degree criterion if and only if it is locally isomorphic to $(B_2)_l^{3,3}$.
\end{lemma}
\begin{proof}
	From Lemma \ref{lemma:matchedBoards}, it follows that any graph with $2$ edges that satisfies the degree criterion must have a table with $4$ crosses, for example,
	\begin{align*}
	\vcenter{\hbox{\includegraphics[scale=\boardscale]{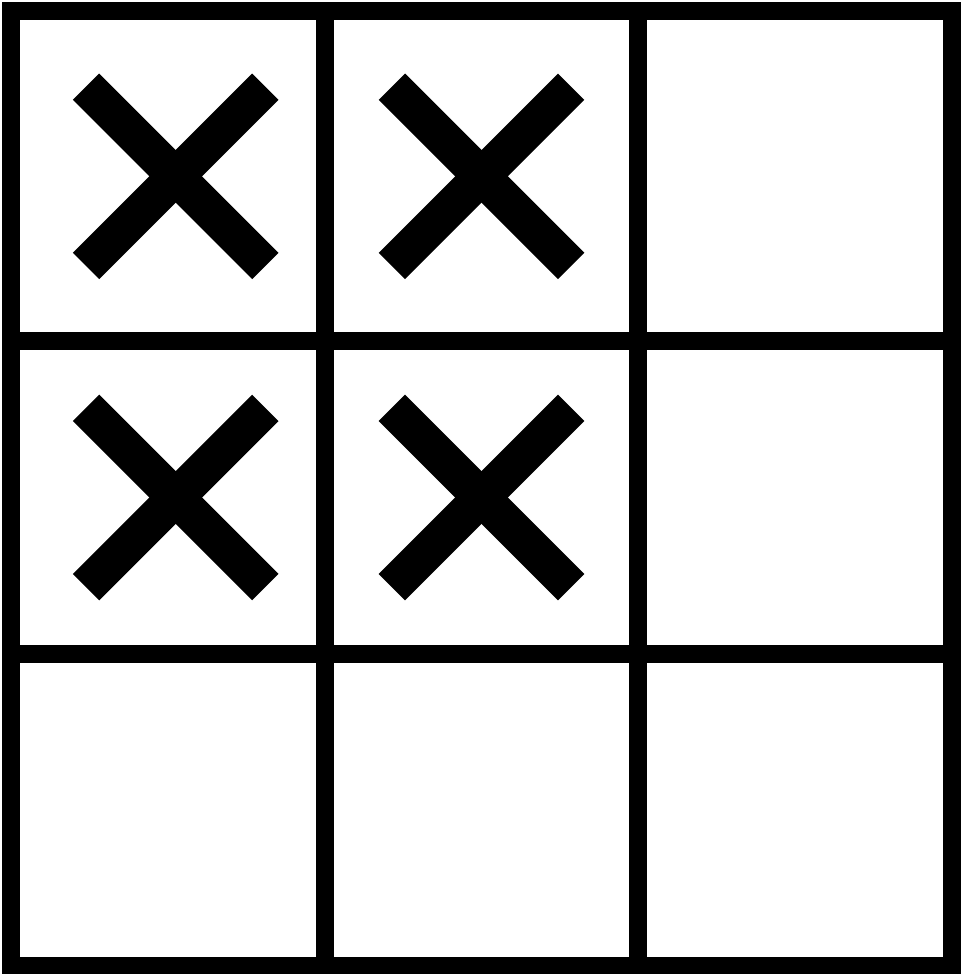}}}~,\text{and}~~
	\vcenter{\hbox{\includegraphics[scale=\boardscale]{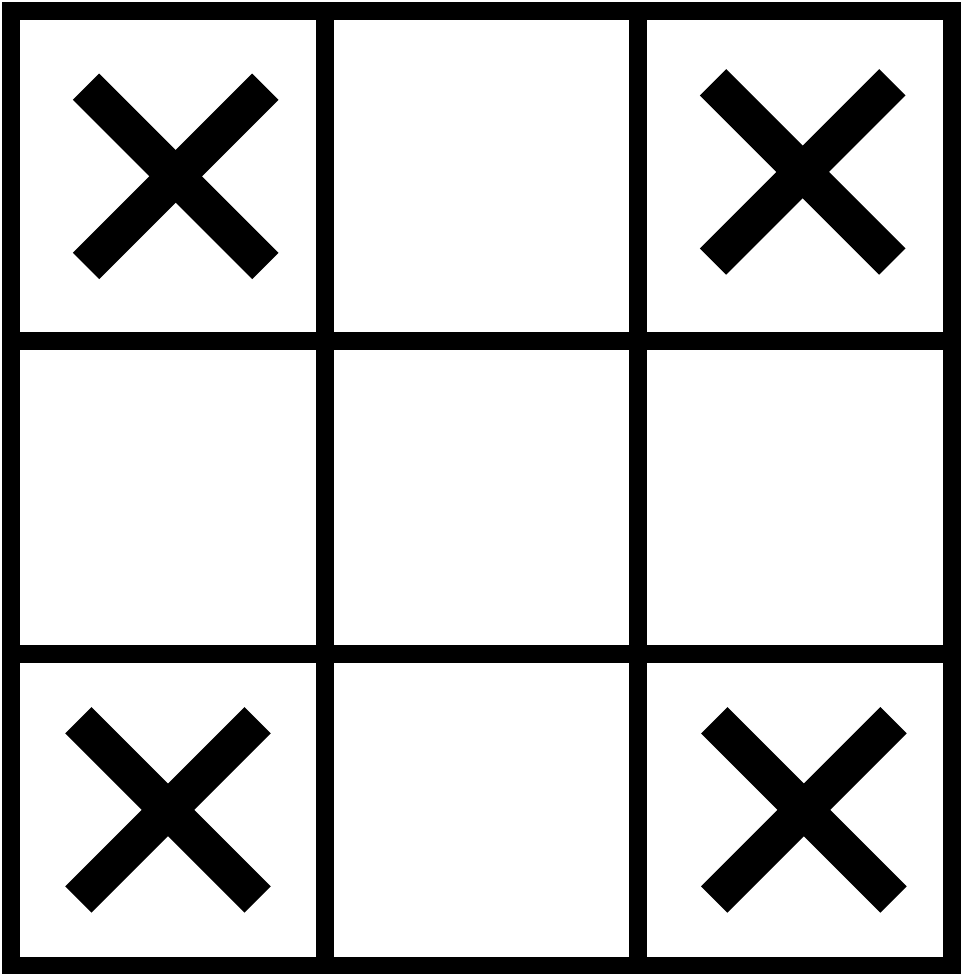}}}, \text{\emph{etc.}}
	\end{align*}
	There are of course several additional tables that have $4$ crosses. However, not all of these are associated with a valid graph. For example, it is clear that there is no $2$ edge graph that will have the table
	\begin{align*}
	\vcenter{\hbox{\includegraphics[scale=\boardscale]{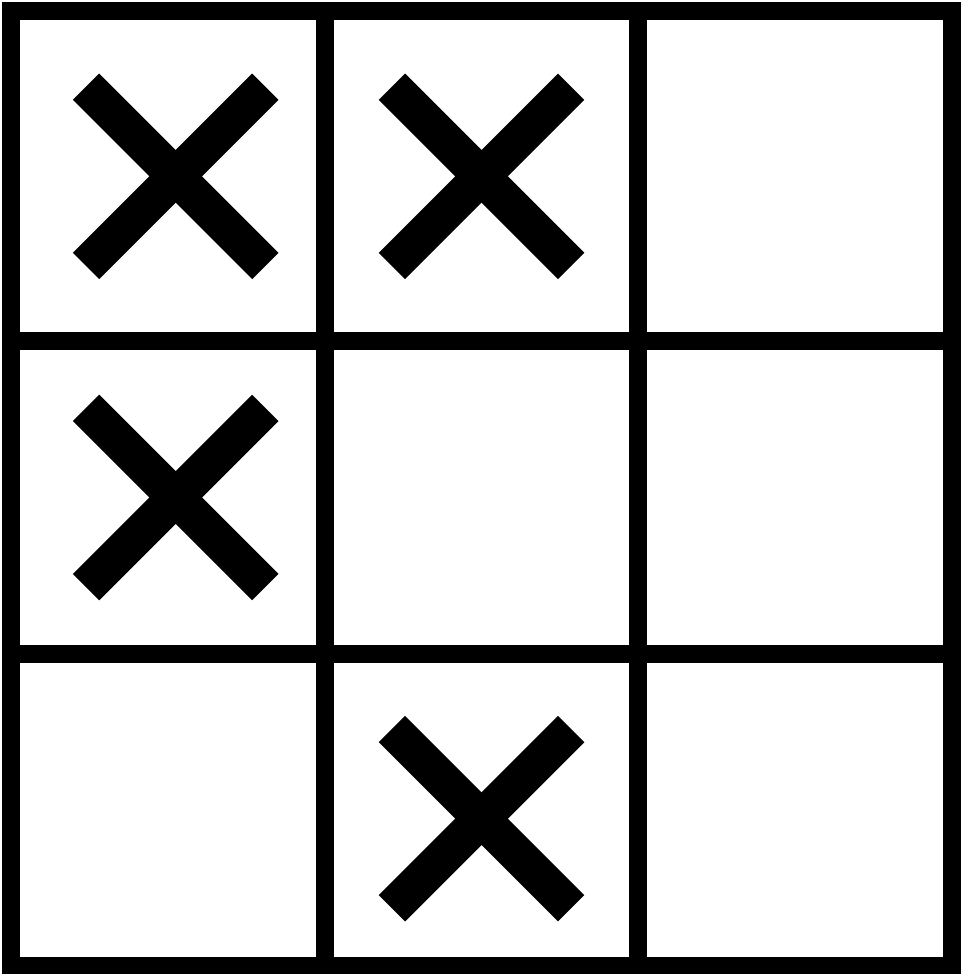}}}.
	\end{align*}
	The only valid tables are those with four crosses in a rectangular placement, as illustrated in the first example.
	It is clear that the only grid-labelled graphs which lead to such tables are locally isomorphic to $(B_2)^{3,3}_l$.
\end{proof}

In proving Lemma \ref{lemma:twodiagonaledges} we have seen that there exist contribution tables that do not correspond to graphs, for example, the table
\begin{align*}
\vcenter{\hbox{\includegraphics[scale=\boardscale]{invalidBoard-cropped}}}.
\end{align*}
If a contribution table does not correspond to a graph then we call it \emph{invalid}. The following lemma is useful.

\begin{lemma}
	\label{lemma:invalid}
	If a contribution table has a column or row with only $1$ non-empty cell then it is invalid.
\end{lemma}
\begin{proof}
	By definition, the sum of any number of edge contribution matrices will never produce a matrix with a column or row with only $1$ non-zero entry. Hence, no contribution table will have a column or row with only $1$ non-empty cell.
\end{proof}

Another useful lemma that comes directly from our experience in proving the $2$ edge case is the following, regarding the contribution tables (equivalently, contribution matrices) of two locally isomorphic grid-labelled graphs.

\begin{lemma}
	Two grid-labelled graphs are locally isomorphic if and only if their contribution tables can be obtained from one another by permuting rows and columns.
\end{lemma}
\begin{proof}
	This follows from the one to one mapping between edges and edge contribution matrices and the definition of local isomorphism.
\end{proof}
Analogously to the case for grid-labelled graphs, we say that two contribution tables that can be obtained from one another by permuting rows and columns are locally isomorphic. We are now in a position to proceed in our characterisation of all $3\times 3$ grid-labelled graphs that satisfy the degree criterion. In what follows we will consider all valid contribution tables for a set number of edges. It does not make sense to consider pairs of contribution tables which can be obtained from one another by rotation since we are only interested in the ``topological'' properties of the graphs that satisfy the degree criterion. In Appendix \ref{appendix:contributiontables} we present some computational problems about contribution tables in an attempt to understand the complexity of finding graphs that satisfy the degree criterion. Let us now continue with our classification of the $3\times 3$ grid-labelled graphs that do so.
\subsection{$3\times 3$ degree criterion with $3$ diagonal edges}
\begin{lemma}
	\label{lemma:threediagonaledges}
	Let $G_l^{3,3}$ be a grid-labelled graph with $3$ diagonal edges. Then $G_l^{3,3}
	$ satisfies the degree criterion if and only if it is locally isomorphic to $(B_3)_l^{3,3}$ or a rotation of $(B_3)_l^{3,3}$.
\end{lemma}
\begin{proof}
	From Lemma \ref{lemma:matchedBoards} we know that any grid-labelled graph $G_l^{3,3}$ with $3$ diagonal edges which satisfies the degree criterion has a contribution table with $6$ crosses. It can be verified by direct inspection that all valid $6$ cross contribution tables are locally isomorphic to rotations of
	\begin{align*} \vcenter{\hbox{\includegraphics[scale=\boardscale]{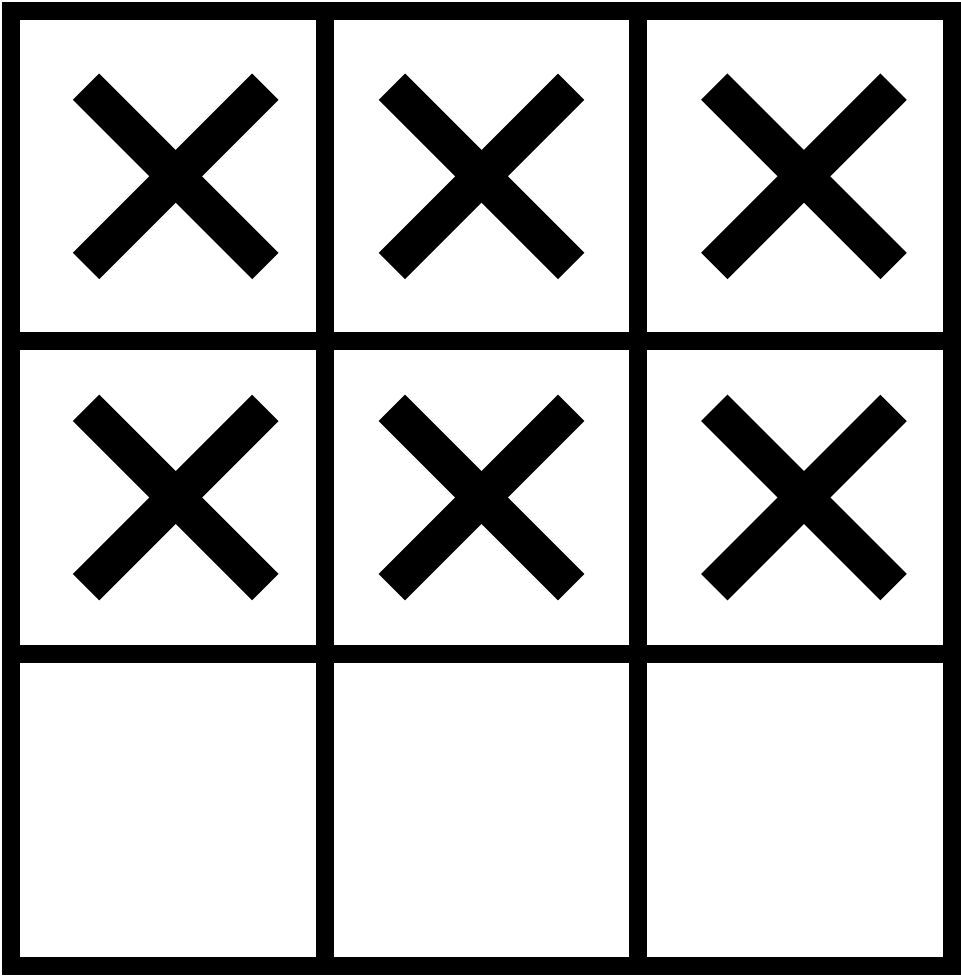}}},
	\end{align*}
	which can only be the contribution table of a grid-labelled graph locally isomorphic to $(B_3)_l^{3,3}$.
\end{proof}

Notice that our reasoning so far about grid-labelled graphs with $2$ or $3$ edges that satisfy the degree criterion is valid for grid-labelled graphs of arbitrary type $(m,n)$, not just $(3,3)$. Indeed, this observation along with 
Lemmas \ref{lemma:twodiagonaledges} and \ref{lemma:threediagonaledges} can be used to prove the following. 
\begin{theorem}
\label{theorem:2or3sep}
Let $G_l^{m,n}$ for any $m,n\ge 2$ be a grid-labelled graph with $2$ (\emph{resp.} $3$) diagonal edges. Then $G_l^{m,n}\in\mathcal{S}$ if and only if $B_2^{3,3}\cong_{m,n}G_l^{m,n}$ (resp. $B_3^{3,3}\cong_{m,n}G_l^{m,n}$).
\end{theorem}
\begin{proof}
We know from Lemmas \ref{lemma:twodiagonaledges} and \ref{lemma:threediagonaledges} that a grid-labelled graph of type $(3,3)$ with $2$ (\emph{resp.} $3$) diagonal edges satisfies the degree criterion if and only if it is locally isomorphic to $(B_2)_l^{3,3}$ (\emph{resp.} a rotation of $(B_3)_l^{3,3}$). The reasoning in the proofs of these lemmas is independent of grid dimension. Due to the structure of $(B_2)_l^{3,3}$ and $(B_3)_l^{3,3}$, all grid-labelled graphs second order locally isomorphic to rotations of these grid-labelled graphs are row or column stratified. From Theorem \ref{theorem:stratified}, such grid-labelled graphs are separable.
\end{proof}

\subsection{$3\times 3$ degree criterion with $4$ and $5$ diagonal edges}

\begin{lemma}
	\label{lemma:fourdiagonaledges}
	Let $G_l^{3,3}$ be a grid-labelled graph with $4$ diagonal edges. Then $G_l^{3,3}
	$ satisfies the degree criterion if and only if it is locally isomorphic to $(B_4)_l^{3,3}$, or is equal to the union of two grid-labelled graphs locally isomorphic to $(B_2)_l^{3,3}$.
\end{lemma}
\begin{proof}
	From Lemma \ref{lemma:matchedBoards}, we know we must only consider tables with $8$ crosses. We leave to the reader to verify that the only ways of placing $8$ crosses on a table such that the table is valid correspond to placements that are locally isomorphic to
	\begin{align}
	\vcenter{\hbox{\includegraphics[scale=\boardscale]{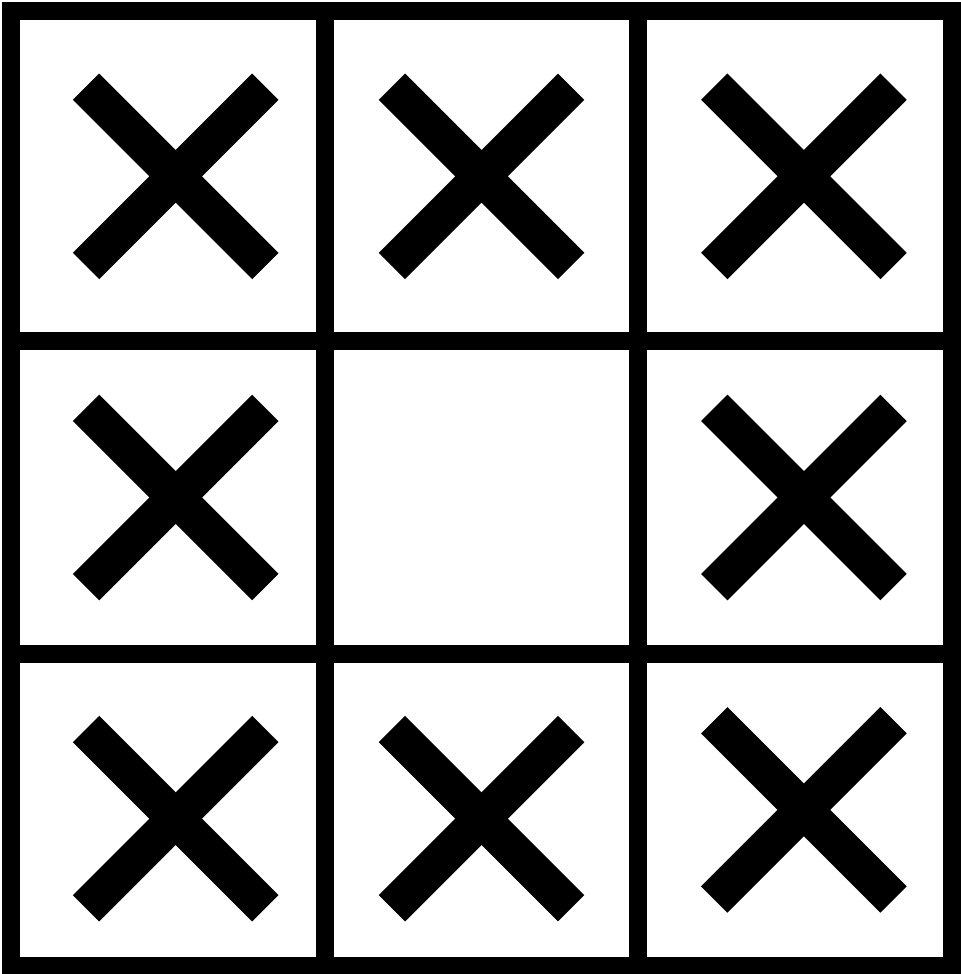}}},\label{eq:table1}
	\end{align}
	or locally isomorphic to rotations of 
	\begin{align}
	\vcenter{\hbox{\includegraphics[scale=\boardscale]{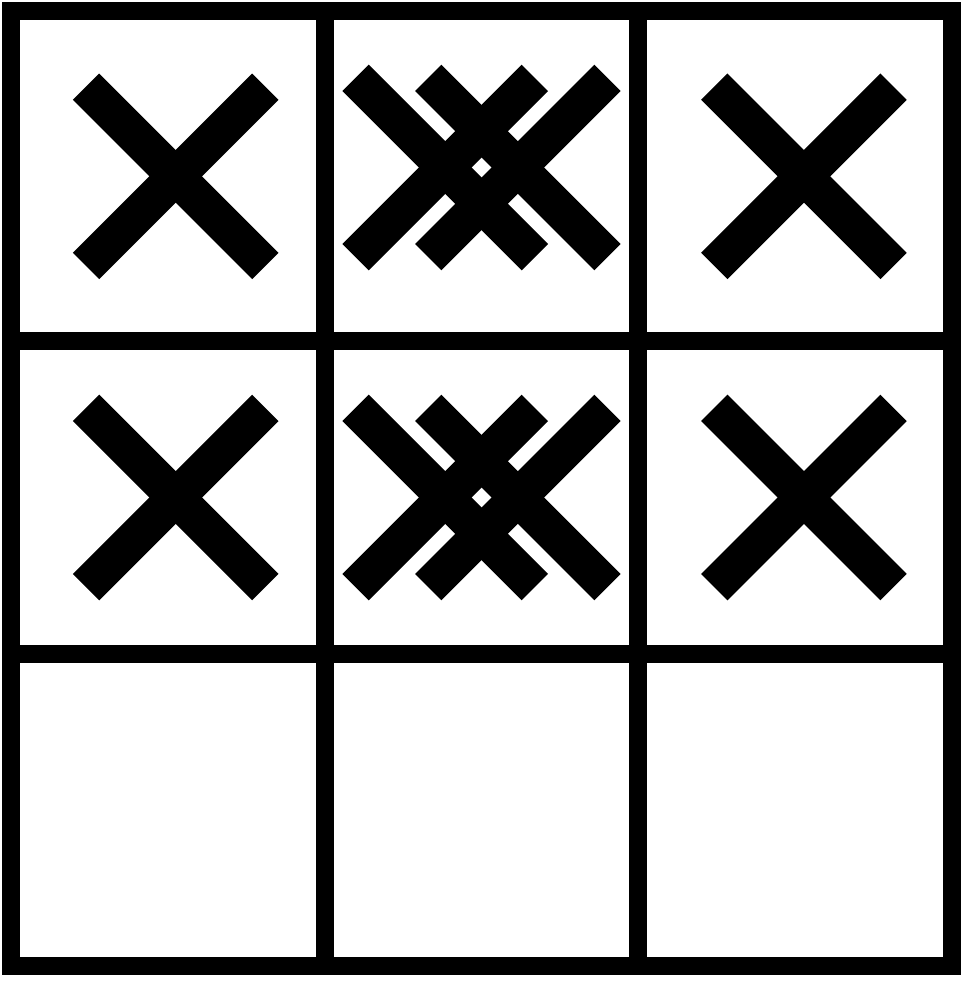}}}.
	\label{eq:table2}
	\end{align}
	Let us now reason about which edges a grid-labelled graph must have to attain the tables (\ref{eq:table1}) and (\ref{eq:table2}), starting with (\ref{eq:table1}). 

	Without loss of generality we may pick any edge supported by the table as a starting point. 
	However, selecting some edges will lead us to a dead-end, unable to proceed further. For an example of an edge leading to a dead-end, let us choose $\{(1,1),(3,3)\}$, removing its associated dashes from the table. This leaves us with 
	\begin{align*}
	\vcenter{\hbox{\includegraphics[scale=\boardscale]{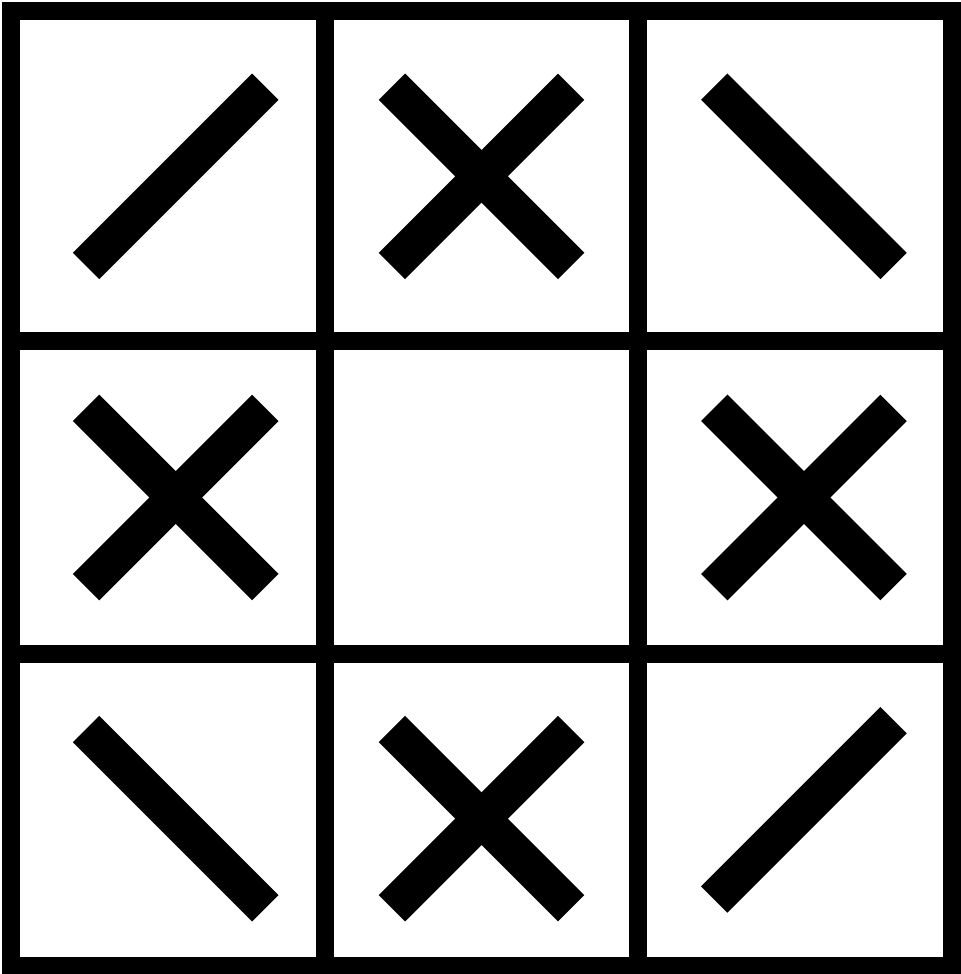}}}.
	\end{align*}
	\begin{figure}
		\centering
		\includegraphics[scale=0.15]{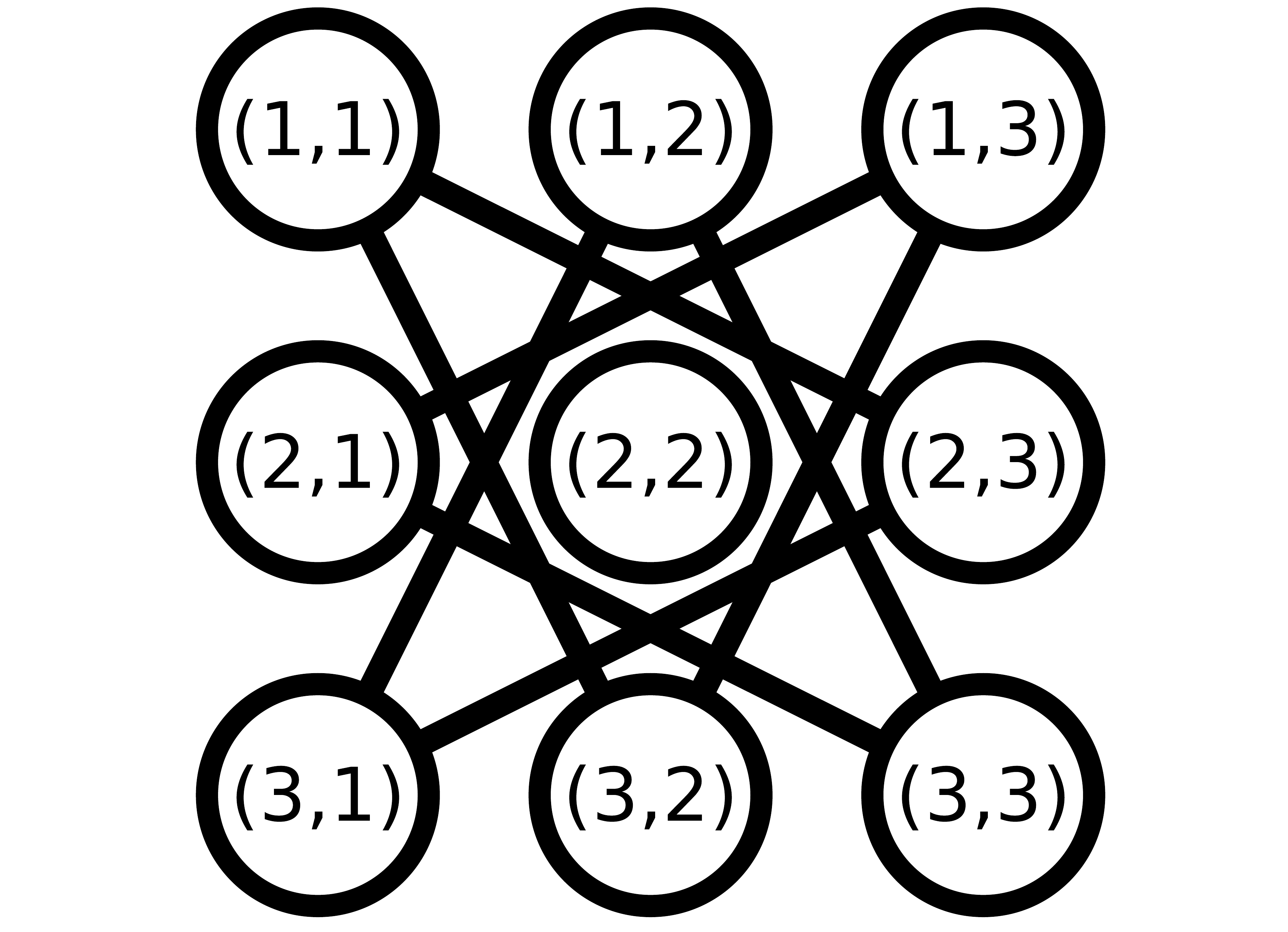}
		\caption{Edges that can be removed from the contribution table in Equation (\ref{eq:table1}) that lead to a valid contribution table.}
		\label{fig:possibleEdges}
	\end{figure}

	Removing any edge from this table will cause it to have at least $1$ row or column with an isolated cross. From Lemma \ref{lemma:invalid}, we know that this will mean the table is invalid.
	
	Indeed, the only single diagonal edges that can be removed from the table in (\ref{eq:table1}) that do not lead to a dead-end are those belonging to the graph illustrated in Figure \ref{fig:possibleEdges}. Removal of any one of these $8$ edges results in a table locally isomorphic to a rotation of 
	\begin{align*}
	\vcenter{\hbox{\includegraphics[scale=\boardscale]{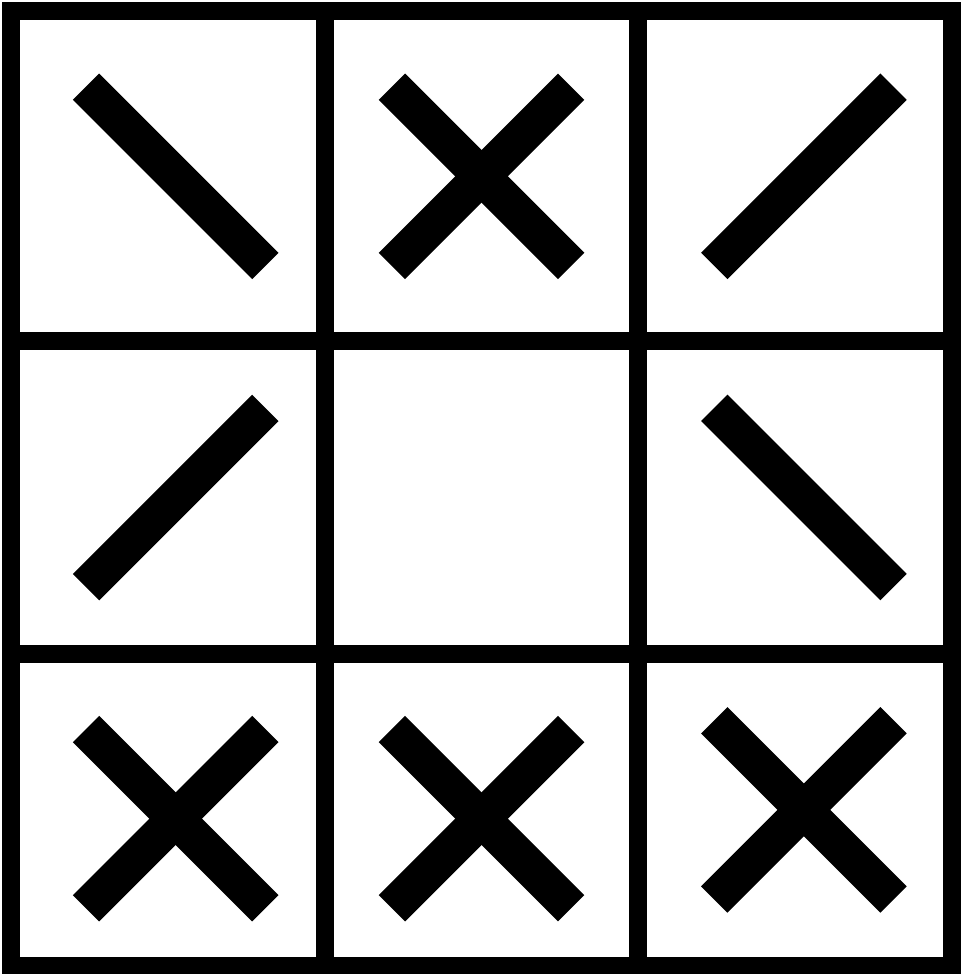}}}.
	\end{align*}
	We leave to the reader to check that this board corresponds to the graph $(B_4)_l^{3,3}$. Each rotation corresponds to a graph that is locally isomorphic to $(B_4)_l^{3,3}$.
	
	All there is left to prove is that the graphs compatible with the table (\ref{eq:table2}) are exclusively made up of the union of two graphs locally isomorphic to $(B_2)_l^{3,3}$. Indeed, this is self-evident. The only compatible grid-labelled graphs are those with their edges arranged into two criss-crosses.
\end{proof}

Before continuing to the $5$ edge case, it will be useful to have the following technical lemmas. The first concerns adding single edges to a graph that satisfies the degree criterion.

\begin{lemma}
	\label{lemma:singleEdgeAddition}
	Let $G_l^{a,b}$ be a grid-labelled graph that satisfies the degree criterion. Let $H_l^{a,b}$ be a grid-labelled graph obtained by adding a single diagonal edge to $G_l^{a,b}$. Then $H_l^{a,b}$ does not satisfy the degree criterion.
\end{lemma}
\begin{proof}
	The graph $G_l^{a,b}$ satisfies the degree criterion. Hence, when considering which edges can be added so that this property is maintained, it is logically equivalent to consider adding edges to the empty graph. As we have seen previously, a graph with a single diagonal edge never satisfies the degree criterion.
\end{proof}

The next lemma will make the $5$ edge case easier to prove.
\begin{lemma}
	\label{lemma:fivediagonaledgestwocrosses}
	Let $G_l^{3,3}$ be a grid-labelled graph with $5$ edges that satisfies the degree criterion. If $G_l^{3,3}$ has $6$ vertices with degree $1$, and $2$ vertices with degree $2$, then $G_l^{3,3}\cong_{3,3} (B_2)_l^{3,3}\cup(B_3)_l^{3,3}$.
\end{lemma}
\begin{proof}
We will consider a number of possible edge contribution tables. Each will have $6$ cells containing a single cross (corresponding to the $6$ vertices with degree $1$), $2$ cells containing a double cross (corresponding to the $2$ vertices with degree $2$), and a single empty cell. As before, we will start with a particular edge contribution table, and attempt to clear it by adding edges to an empty $3\times 3$ grid-labelled graph. If we can not clear the table by adding edges, we know that the table is invalid. In order to clear a single cross from the table, $2$ edges must be added. In order to clear a double cross from the table, it is clear that $4$ edges must be added. A double cross can be in the same row (resp. column) as the empty cell, or on a different row (resp. column). 

Let us reason about both cases now. Without loss of generality, assume $1$ of the double cross cells is on the same row as the empty cell. Then the table is locally isomorphic to the following,
	\begin{align*}
  	\vcenter{\hbox{\includegraphics[scale=\boardscale]{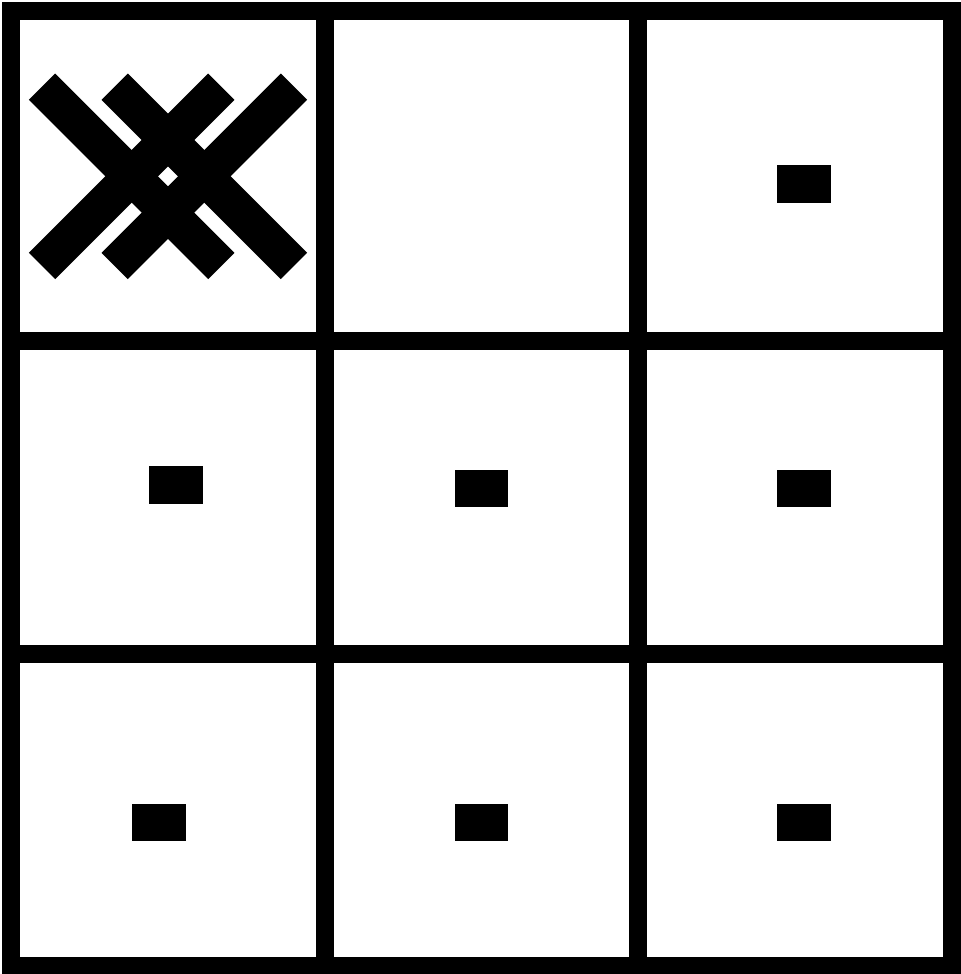}}},
  	\end{align*}
where the black dots denote cells whose contents are irrelevant to the proof. It is obvious that the empty cell restricts the choice of edges that can be selected to remove the double cell. To remove each dash from the double cell under consideration, the four edges in the following graph must be selected,
\begin{align*}
  	  \vcenter{\hbox{\includegraphics[scale=0.15]{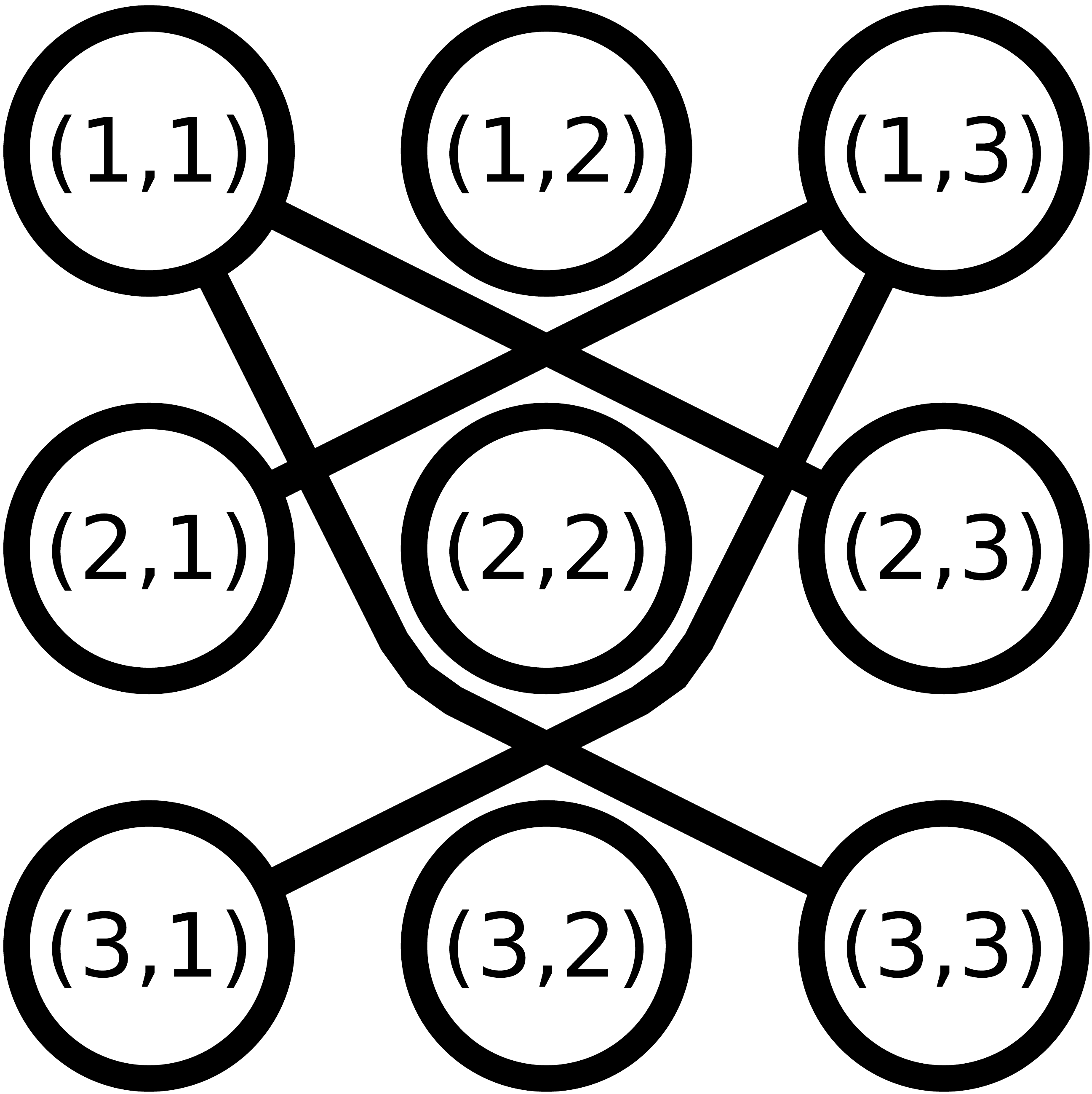}}},
\end{align*}
because these are the only edges that are compatible with such a contribution table.
Selecting these four edges will produce a grid-labelled graph that is locally isomorphic to the union of $2$ grid-labelled graphs locally isomorphic to $(B_2)_l^{3,3}$. We know by Lemma \ref{lemma:singleEdgeAddition} that a single edge cannot be added to such a graph to produce a grid-labelled graph that satisfies the degree criterion. Such tables are therefore invalid.

When a double cross cell is not on the same row as the empty cell then the table is locally isomorphic to the following:
\begin{align*}
  	\vcenter{\hbox{\includegraphics[scale=\boardscale]{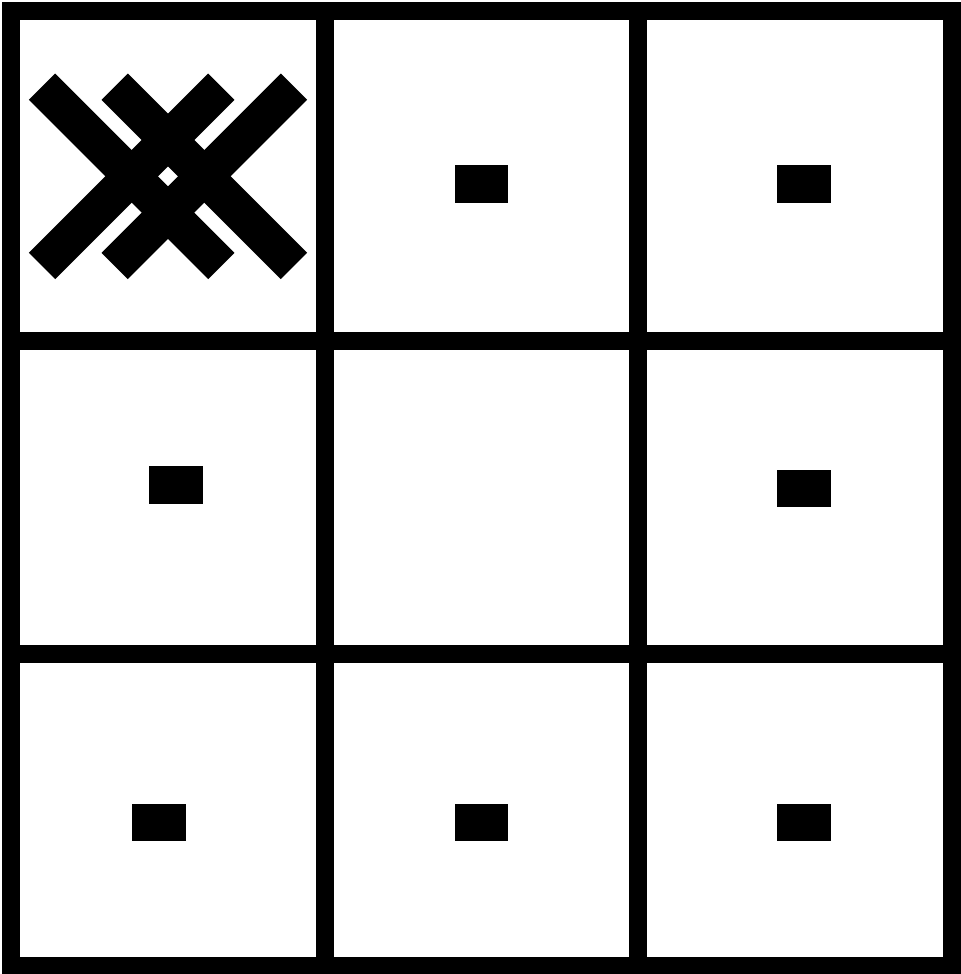}}}.
  	\end{align*}
To remove all the dashes from the double cross cell, $4$ of the $6$ edges from the following grid-labelled graph must be selected:
\begin{align*}
  	  \vcenter{\hbox{\includegraphics[scale=0.22]{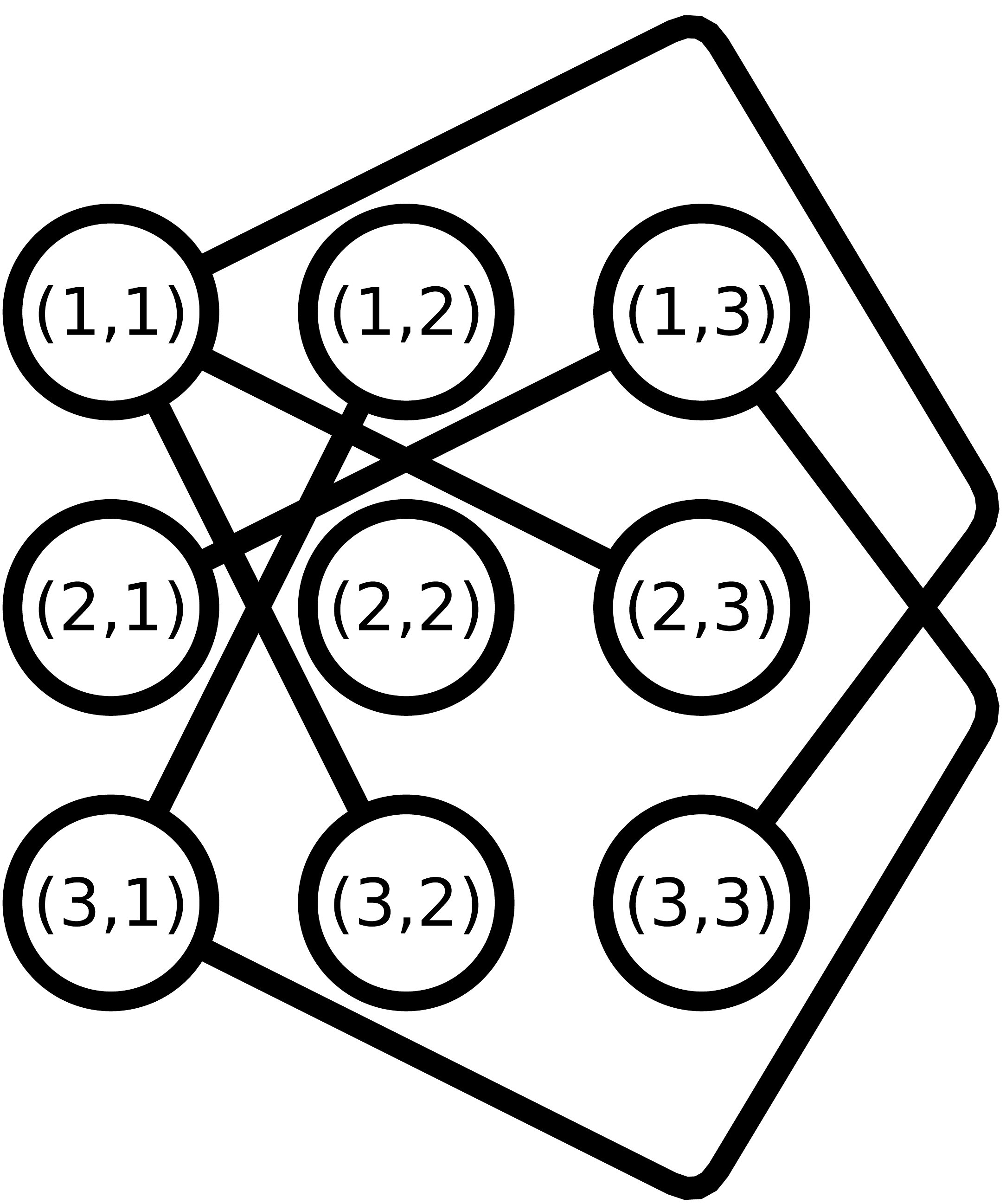}}}.
\end{align*}
It is impossible to choose $4$ edges from this grid-labelled graph without introducing a subgraph locally isomorphic to $(B_2)_l^{3,3}$. Any remaining edges added to this grid-labelled graph such that it still satisfies the degree criterion would need to be from a grid-labelled graph that itself satisfied it. By Lemma \ref{lemma:threediagonaledges}, the only three edge graphs that have this property are locally isomorphic to $(B_3)_l^{3,3}$. The lemma then holds.
\end{proof}
\begin{lemma}
	\label{lemma:fivediagonaledges}
	Let $G_l^{3,3}$ be a graph with $5$ diagonal edges. Then $G_l^{3,3}
	$ satisfies the degree criterion if and only if it is locally isomorphic to $(B_5)_l^{3,3}$, or has a decomposition into two graphs locally isomorphic to rotations of $(B_2)_l^{3,3}$ and $(B_3)_l^{3,3}$ respectively.
\end{lemma}
\begin{proof}
	By Lemma \ref{lemma:matchedBoards} we must consider all tables with $10$ crosses. There are a limited number of ways of placing $10$ crosses on a $3\times 3$ table that lead to the table being valid. In order to fit all the crosses on the table, there must be some cells of the table that have double crosses. 
	In Lemma \ref{lemma:fivediagonaledgestwocrosses} we have seen that we don't need to consider the tables with $2$ double crosses. It is obvious that tables with more than $2$ double crosses are invalid, so, we must only consider the tables with a single double cross.
	We can reason about these tables by considering the example
	\begin{align}
	\vcenter{\hbox{\includegraphics[scale=\boardscale]{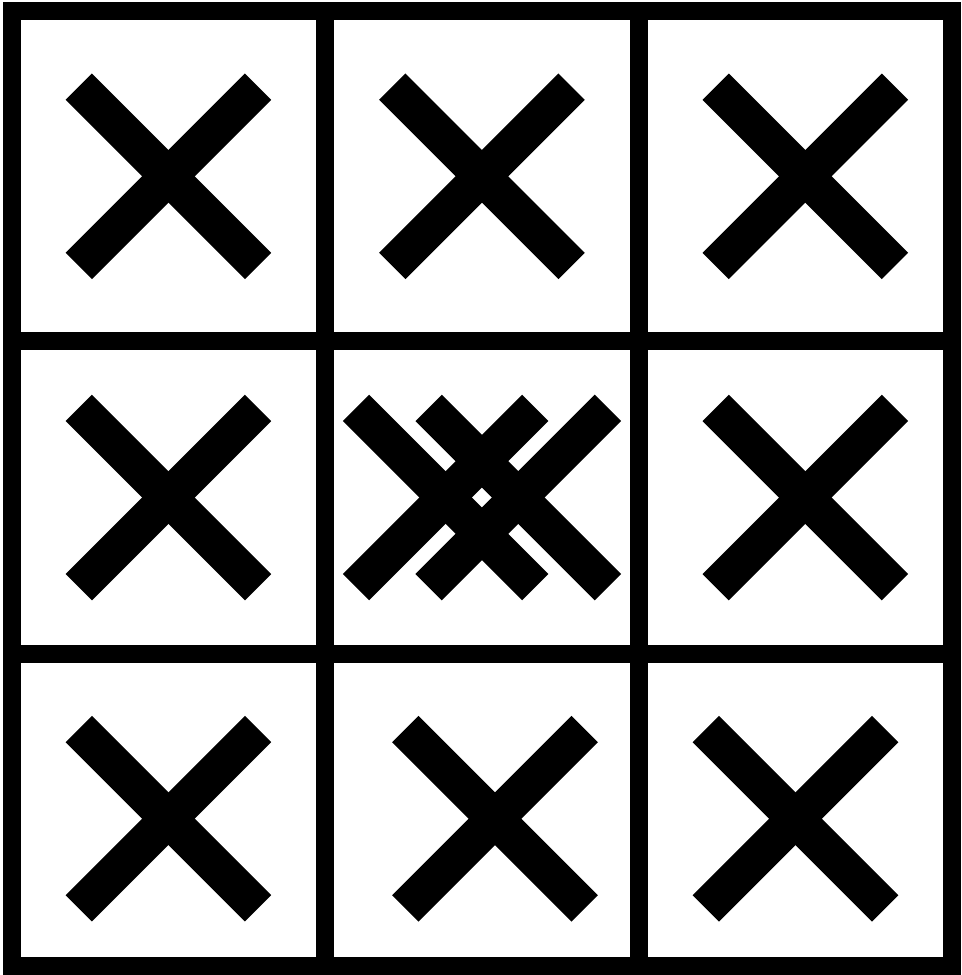}}}\label{eq:onedoublecross},
	\end{align}
	to which all others of this kind are locally isomorphic.
	To eliminate the double cross in the middle of the table, $4$ edges from the graph
	\begin{align}
	\vcenter{\hbox{\includegraphics[scale=0.15]{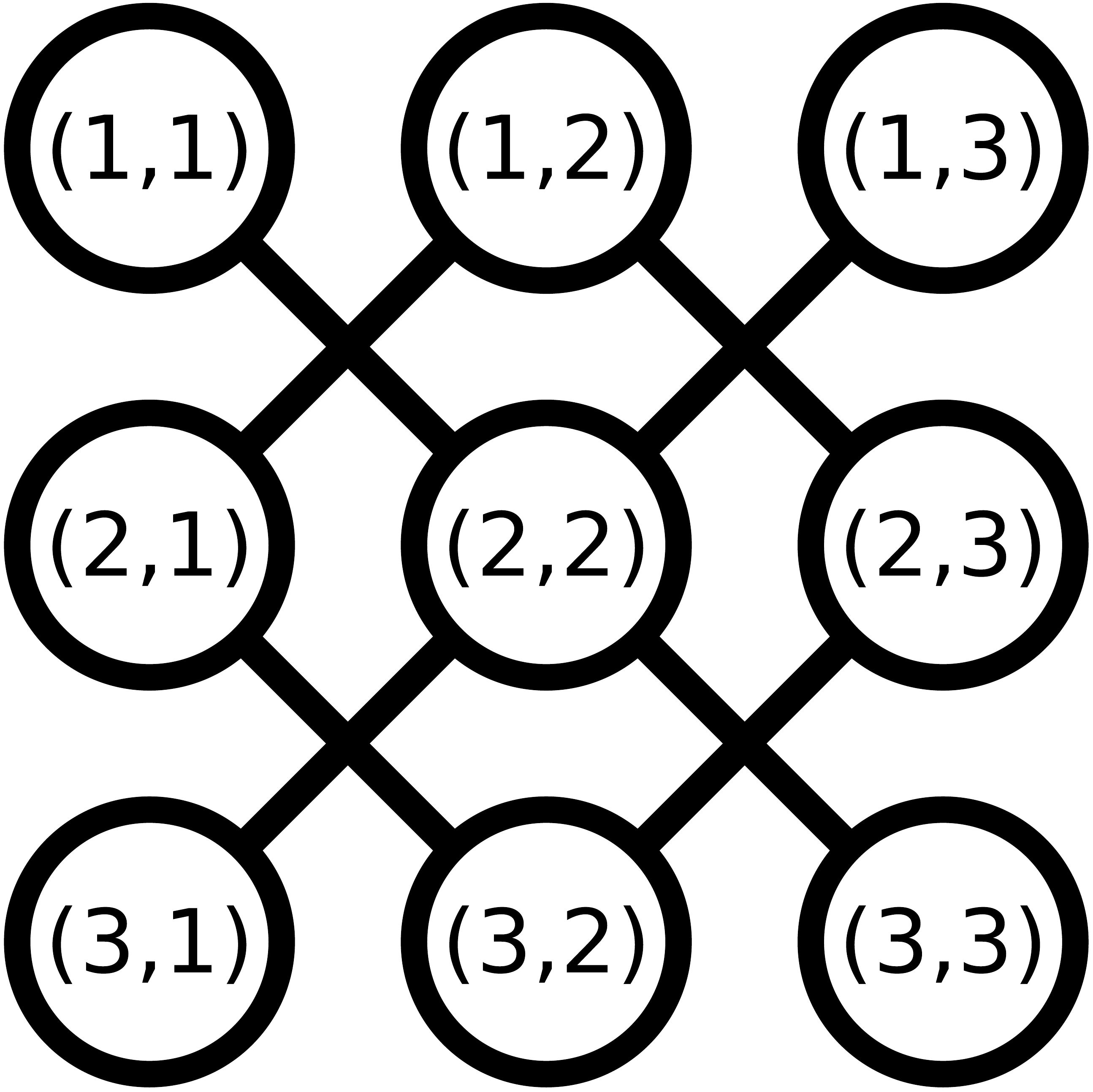}}}.\label{eq:onehop}
	\end{align}
	must be selected. 
	It is clear that the $4$ edges must be chosen such that there is no subgraph locally isomorphic to $(B_2)_l^{3,3}$. If there was a such a subgraph, then the remaining $3$ edges would need to form a graph locally isomorphic to $(B_3)_l^{3,3}$ in order for the degree criterion to be satisfied, and we have already considered such a case.
	Consider the removal of the downhill edge $\{(1,1),(2,2)\}$. This leads to the table
	\begin{align*}
	\vcenter{\hbox{\includegraphics[scale=\boardscale]{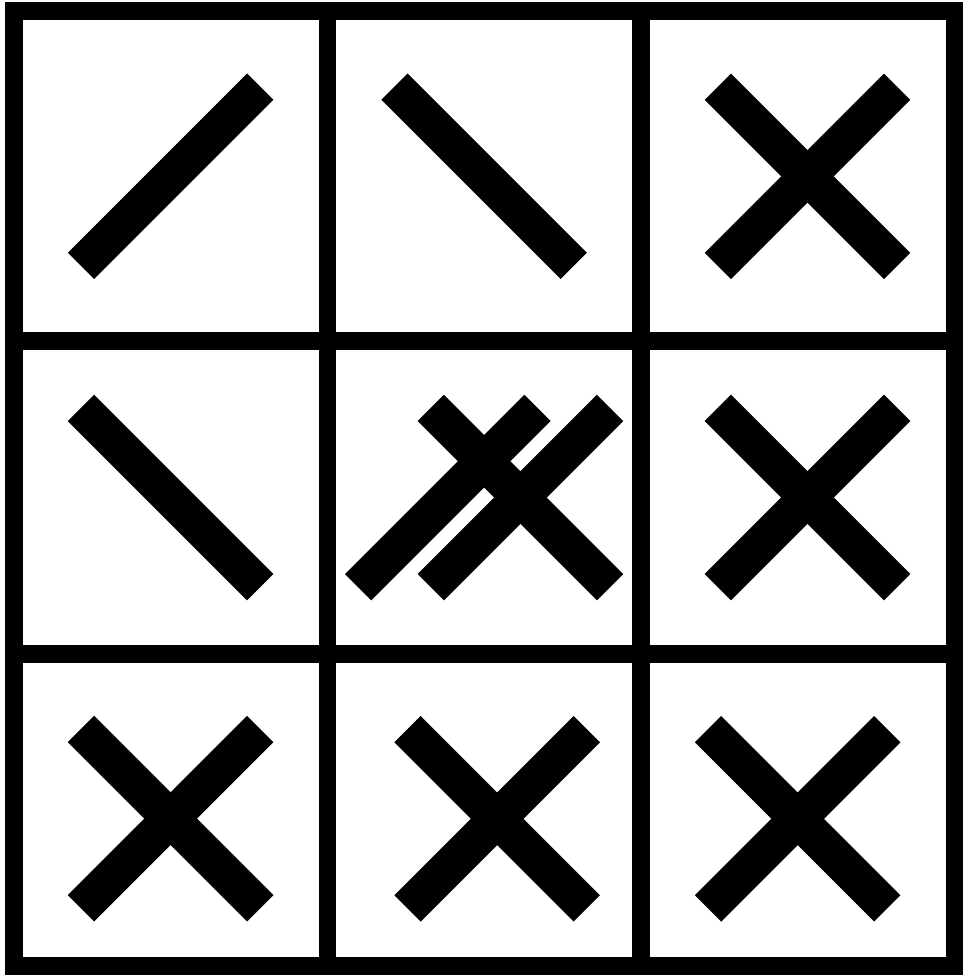}}}.
	\end{align*}
	Clearly the remaining three edges we must choose from (\ref{eq:onehop}) must be chosen to be downhill, as there is no other way of selecting edges without introducing isolated crosses on rows or columns. If instead we chose our first edge to be \emph{uphill}, by similar reasoning the remaining three must also be chosen to be uphill.
	
	We can then observe that removal of $4$ uphill or downhill edges from the board (\ref{eq:onedoublecross}) in the way described leaves a single edge in the direction that completes the grid-labelled graph $(G_5)_l^{3,3}$. Hence, the lemma is proved.
\end{proof}
\begin{lemma}
	\label{lemma:morethansix}
	Let $G_l^{3,3}$ be a grid-labelled graph with $e\ge 6$ diagonal edges. If $G_l^{3,3}$ satisfies the degree criterion, then it has a decomposition $(X_1)_l^{3,3},\dots
	(X_n)_l^{3,3}$ where for all $1\le i\le n$, $(X_i)_l^{3,3}$ is locally isomorphic to a rotation of a building-block.
\end{lemma}
\begin{proof}
	A $3\times 3$ grid-labelled graph can have up to $18$ diagonal edges. If a $3\times 3$ grid-labelled graph has $10$ or more diagonal edges, then it is obvious that it must contain a subgraph locally isomorphic to $(B_2)_l^{3,3}$. 	
 The remainder of the cases, graphs with $6\le e\le9$ diagonal edges, are covered in Lemma \ref{lemma:69edges} which is proved in Appendix \ref{appendix:proof}. The lemma then follows by induction.
	
	\end{proof}

\begin{lemma}
	\label{lemma:69edges}
	Let $G_l^{3,3}$ be a grid-labelled graph with $6\le e\le9$ diagonal edges. If $G_l^{3,3}$ satisfies the degree criterion, then it has a decomposition $\{(X_1)_l^{3,3},\dots
	(X_n)_l^{3,3}\}$ where for all $1\le i\le n$, $(X_i)_l^{3,3}$ is locally isomorphic to a rotation of a building-block.
\end{lemma}

We can finally classify the $3\times 3$ grid-labelled graphs that satisfy the degree criterion.
\begin{proposition}
	Let $G_l^{3,3}$ be a grid-labelled graph. Then $G_l^{3,3}$ satisfies the
	degree criterion if and only if it has a decomposition $\{(X_1)_l^{3,3},\dots
	(X_n)_l^{3,3}\}$ where for all $1\le i\le n$, $(X_i)_l^{3,3}$ is locally isomorphic to a rotation of a building-block.
\end{proposition}
\begin{proof}
	It is obvious that a grid-labelled graph obtained as a union of building-block graphs will satisfy the degree criterion, because each building-block graph satisfies it.
	
	Let us now prove the other direction.
	Let $G_l^{3,3}$ be a grid-labelled graph that satisfies the degree criterion. From Lemmas \ref{lemma:twodiagonaledges}, \ref{lemma:threediagonaledges}, \ref{lemma:fourdiagonaledges} and \ref{lemma:fivediagonaledges} we know that if $G_l^{3,3}$ has $2\le m \le 5$ diagonal edges then it is locally isomorphic to a building-block $(B_m)^{3,3}_l$, or has a decomposition into grid-labelled graphs that are locally isomorphic to building-blocks, or rotations of building blocks.
	
	Finally, from Lemma \ref{lemma:morethansix}, we know that this is also true for graphs with $6$ or more diagonal edges. This concludes the proof.
\end{proof}

We have characterised all of the $3\times 3$ grid-labelled graphs that satisfy the degree criterion by means of a forbidden subgraph type argument. We have found that such grid-labelled graphs satisfy the degree criterion if and only if they are built out of a small set of ``building-block graphs''. It is reasonable to assume that the same is true for the grid-labelled graphs on a larger grid. A natural question is how the size of the building-block set grows as a function of grid dimension.

Note that the results in this section are applicable to the $(k,a,b)$-\textsc{GraphSeparability} problem for $m=n=3$ and all $k$. Along the way we have also found solutions to the $k=1,2,3$ case for any $a$ and $b$ (see Theorem \ref{theorem:2or3sep}). In arbitrary dimensions $a,b$ there are of course a finite number of graphs with fixed number of edges $k$ that satisfy the degree criterion. We do not consider $(4,a,b)$-\textsc{GraphSeparability} and $(5,a,b)$-\textsc{GraphSeparability}, but finding all of the building-blocks for these cases should not prove difficult. An fascinating question would be to count how many graphs satisfy the degree criterion as a function of $k$, when fixed grid dimensions are not taken into account. The reasoning to be used in this case is exactly equivalent to what we have done in this section, but on an infinite grid.

Not all of the $3\times 3$ building-blocks correspond to separable states. Indeed, the well known matrix realignment criterion \cite{MR} shows that both $(B_4)_l^{3,3}$ and $(B_5)_l^{3,3}$ correspond to entangled states. Entangled states that are positive under partial transpose correspond to states that have zero \emph{distillable entanglement} and are referred to as \emph{bound entangled}~\cite{Horodecki2009}. Hence, our characterisation can not be extended trivially to a necessary and sufficient criterion for separability for $3\times 3$ graphs. To do so will require analysis on whether the unions of these bound entangled building-blocks with other building-blocks causes the corresponding state to be separable. We highlight this as an open question.

\subsection{Counting graphs that satisfy the degree criterion}
\label{subsection:enumeration}
In the previous section we showed that type $(3,3)$ grid-labelled graphs satisfy the degree criterion if and only if they are constructed from a small number of building-block graphs. Now we will outline a more general framework for enumeration of grid-labelled graphs of type $(a, b)$, with $k$ edges, that satisfy the degree criterion.

In what follows, the quantity $P_k(a,b)$ is defined to be the number of graphs of type $(a, b)$ with $k$ edges that satisfy the degree criterion. Let $D_k(a,b)$ be the number of graphs of type $(a,b)$ with $k$ edges, all diagonal, that satisfy the degree criterion.

 \begin{definition}[Rook's graph]
The \emph{rook's graph} is the grid-labelled graph $R_l^{a,b}$ with an edge between every pair of vertices in the same row or column. It is well known that the $a\times b$ rook's graph has \begin{align*}r(a,b):= \frac{a\cdot b(a+b)}{2}-a\cdot b\end{align*} edges.
\end{definition}

\begin{lemma}
\label{lemma:genericcounting}
For any $a,b,k\in\mathbb{N}$,
\begin{align*}P_k(a,b)={r(a,b)\choose k}+\sum_{i=2}^{k}D_i(a,b)\cdot {r(a,b) \choose k-i}.\end{align*}
\end{lemma}
\begin{proof}
Let $G(a,b,d,h)$ be equal to the number of grid-labelled graphs of type $(a,b)$ with $d$ diagonal edges and $h$ non-diagonal (horizontal or vertical) edges that satisfy the degree criterion. Clearly,
\begin{align*}
P_k(a,b)=\sum_{i=0}^{k}G(a,b,i,k-i).
\end{align*}
Since the degree criterion is not affected by horizontal or vertical edges, and there are at most $r(a,b)$ horizontal and vertical edges in a grid-labelled graph of type $(a,b)$, it is clear that for any $i,j\ge0$,
\begin{align*}
G(a,b,i,j)=D_i(a,b)\cdot {r(a,b) \choose j}.
\end{align*}
Since for any $a,b$, $D_0(a,b)=D_1(a,b)=0$, the lemma holds.
\end{proof}

We obtain the first few values of $D_k(a,b)$ in the next statement.

\begin{proposition}
\label{proposition:dcounting}
For any $a,b\ge 2$, the following are true:
\begin{itemize}
\item $D_2(a,b)={a\choose 2}\cdot {b\choose 2}$;
\item $D_3(2,b)=2\cdot{b\choose 3}$;
\item $D_4(2,b)=3\cdot{b\choose 3}+9\cdot{b\choose 4}$.
\end{itemize}
\end{proposition}
\begin{proof}
By Lemma \ref{lemma:twodiagonaledges}, a grid-labelled graph $G_l^{3,3}$ with $2$ diagonal edges satisfies the degree criterion if and only if it is locally isomorphic to $(B_2)_l^{3,3}$. Indeed, it is obviously the case that any $G_l^{a,b}$ with two diagonal edges satisfies the degree criterion if and only if $G_l^{a,b}\cong_{3,3}(B_2)_l^{3,3}$. Therefore, we must count the number of graphs of type $G_l^{a,b}$ that are locally isomorphic to $(B_2)_l^{a,b}$. For $a=2$ this is of course ${b\choose 2}$. Increasing $a$ gives an additional ${a\choose 2}$-dimensional degree of freedom, so
\begin{align*}
D_2(a,b)&={a\choose 2}\cdot D_2(2,b)\\
&={a\choose 2}\cdot{b\choose 2}.
\end{align*}

\begin{figure}
\includegraphics[scale=0.2]{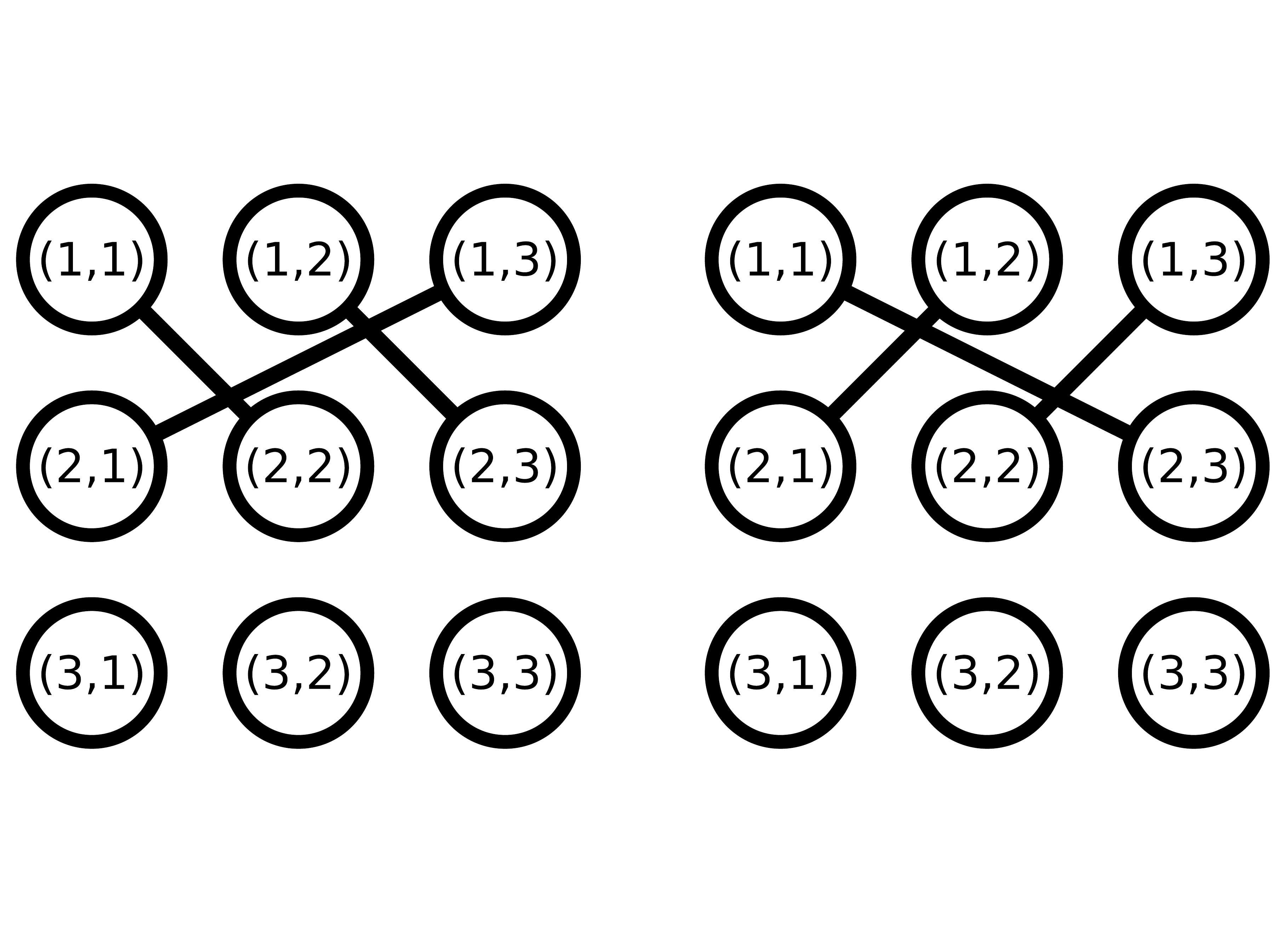}
\caption{The reflections of $(B_3)_l^{3,3}$ that are relevant for separability in grid-labelled graphs of type $(2,b)$.}
\label{fig:b3rotations}
\end{figure}

By Lemma \ref{lemma:threediagonaledges}, we know that any $G_l^{3,3}$ with $3$ diagonal edges satisfies the degree criterion if and only if $G_l^{3,3}$ is locally isomorphic to a rotation of $(B_3)_l^{3,3}$. Similarly to the $2$ edge case, this generalises to arbitrary grid-labelled graphs with $3$ edges. In other words, $G_l^{a,b}$ with three edges satisfies the degree criterion if and only if it is second order locally isomorphic to a rotation of $(B_3)_l^{3,3}$. Here, we only consider the cases $G_l^{2,b}$. Unlike $(B_2)_l^{3,3}$, the graph $(B_3)_l^{3,3}$ is not invariant under reflections. Indeed, there are two graphs that we must consider when counting the $3$ edge degree criterion in type $(2,b)$ grid-labelled graphs, which are illustrated in Figure \ref{fig:b3rotations}. By similar reasoning to the $2$ edge case,
\begin{align*}
D_3(2,b)=2\cdot{b\choose 3}.
\end{align*}
\begin{figure}[ht!]
\includegraphics[scale=0.13]{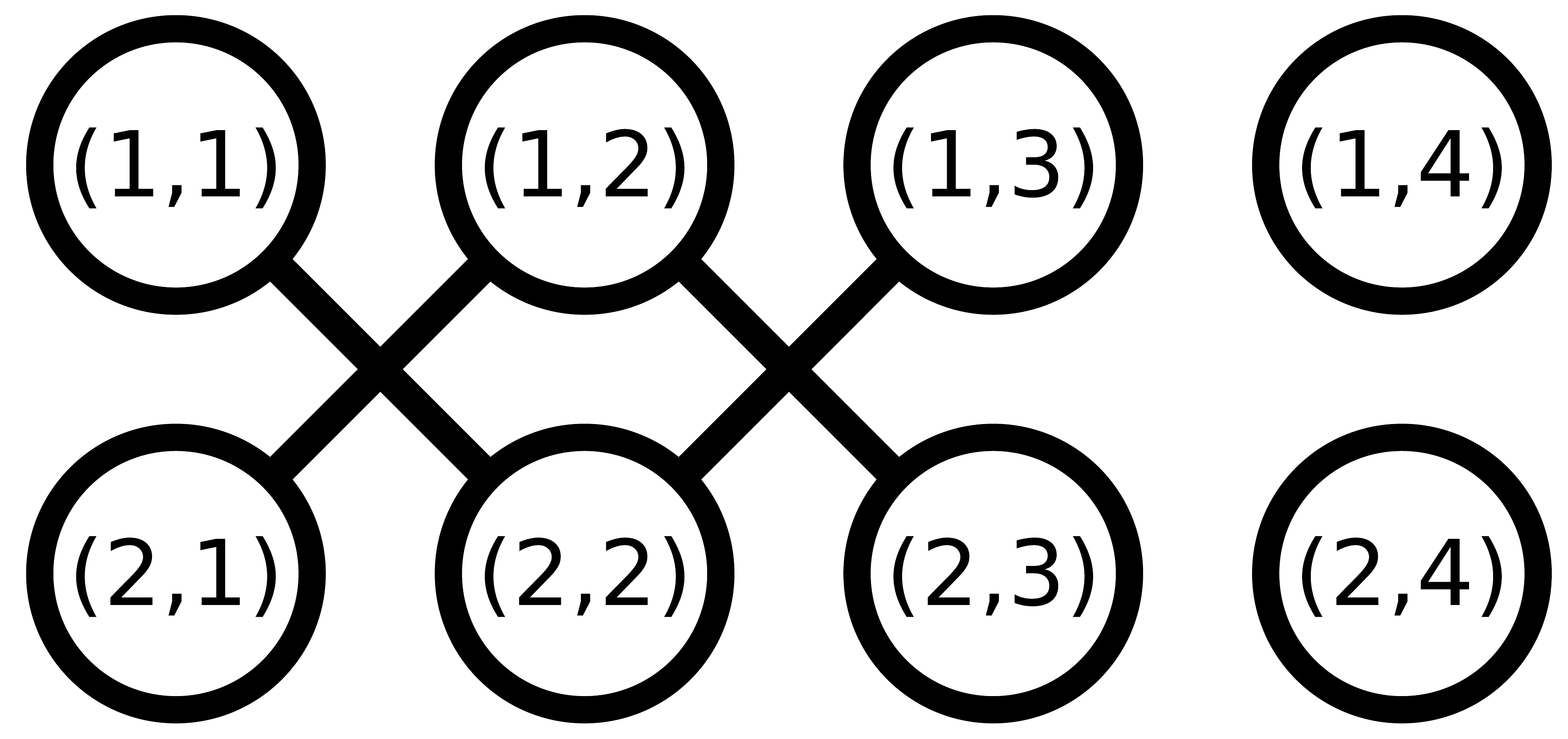}
\caption{An example of a $0$, $1$ and $2$ degree graph.}
\label{fig:threecase}
\end{figure}
\begin{figure}[ht!]
\includegraphics[scale=0.28]{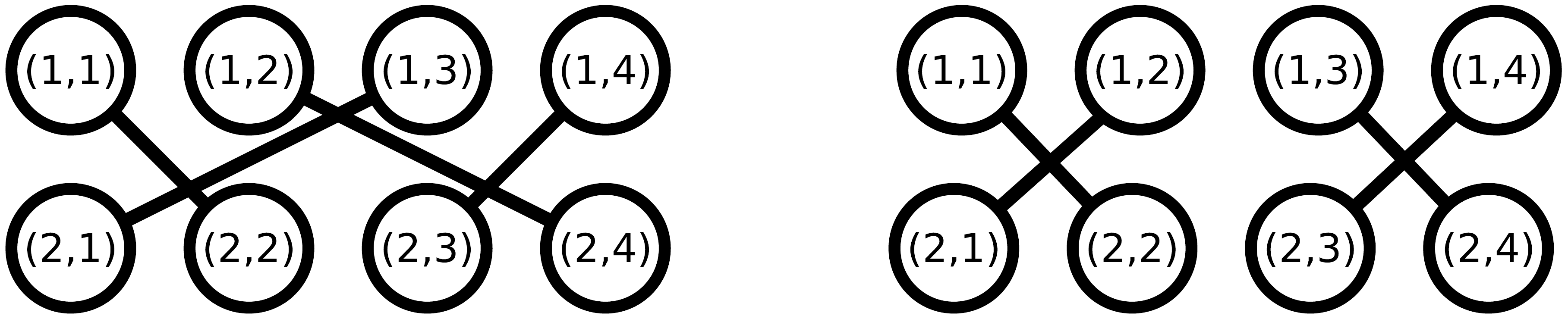}
\caption{Two examples of the $0$ and $1$ degree grid-labelled graphs. There is a bijection between such graphs and the derangements of a set of $4$ elements.}
\label{fig:derangements}
\end{figure}
In order to obtain $D_4(2,b)$, we will use a different approach. We know from Corollary \ref{corollary:rowdegrees} that a grid-labelled graph of type $(2,b)$ satisfies the degree criterion if and only if the degrees of the vertices in the top row are equal to the vertices in the bottom row. We leave it to the reader to verify that there are only two types of four edge grid-labelled graphs of type $(2,b)$ with this condition: those with all vertices of degree $0,1$ or $2$, and those with all vertices of degree $0$ or $1$. The former kind are all locally isomorphic to the grid-labelled graph illustrated in Figure \ref{fig:threecase}. Hence there are $3\cdot {b\choose 3}$ grid-labelled graphs of type $(2,b)$ with this pattern of degrees, because there are three columns of non-zero degree vertices: two with degree $1$ vertices, one with degree $2$. Therefore, there are ${b\choose 3}$
placements of these columns, and $3$ orderings within each placement.

We can enumerate the latter kind of grid-labelled graph by noticing that they are precisely the \emph{derangements} of a set of $4$ elements (consider the illustration in Figure \ref{fig:derangements}. Let us recall that a derangement of a finite set is a permutation without fixed points. The number of derangements on a set of $n$ elements is counted by the \emph{sub-factorial} function, defined
\begin{align*}
!n:=(n-1)(!(n-1)+!(n-2)).
\end{align*}
Hence, in this case, there are $!4\cdot{b\choose 4}=9\cdot{b\choose 4}$ grid-labelled graphs under consideration with this edge pattern. Therefore,
\begin{align*}
D_4(2,b)=3\cdot{b \choose 3}+9\cdot{b\choose 4}.
\end{align*}
\end{proof}

From what we have proved we are able to see the following.

\begin{theorem}[Entanglement of random grid-labelled graphs]
~
\begin{itemize}
\item Let $G_l^{a,b}$ be a random grid-labelled graph with $2$ edges. Then asymptotically almost surely (a.a.s), $G_l^{a,b}\not\in \mathcal{S}$.
\item Let $G_l^{2,b}$ be a random grid-labelled graph with $3$ (\emph{resp.} $4$) edges. Then a.a.s, $G_l^{2,b}\not\in \mathcal{S}$.
\end{itemize}
\end{theorem}
\begin{proof}
We know from Lemma \ref{lemma:degreecriterioniff} that we only need to consider the diagonal edges of a grid-labelled graph to test if it satisfies the degree criterion. Let us compare the growth of $D_k(a,b)$ as a function of grid dimension with that of $G^D_k(a,b)$, the total number of grid-labelled graphs with $k$ edges, all diagonal.

Let us first obtain a general form for $G^D_k(a,b)$. Let $Q(a,b)$ be the number of diagonal edges in the complete grid-labelled graph $(K_{a\cdot b})_l^{a,b}$. Then
\begin{align*}
G^D_k(a,b)={Q(a,b)\choose k}.
\end{align*}

Clearly 
\begin{align*}
Q(a,b)&=|E(K_{a\cdot b})|-r(a,b)\\
&=\frac{a\cdot b(a\cdot b-1)}{2}-\frac{a\cdot b(a+b)}{2}+a\cdot b\\
&=\frac{a^2\cdot b^2 - a\cdot b - a^2\cdot b -a\cdot b^2}{2}+a\cdot b\\
&=\frac{(a^2-a)(b^2-b)}{2}\\
&=2{a\choose 2}\cdot{b \choose 2},
\end{align*}
Hence, 
\begin{align*}
G_k^D(a,b)={2\cdot{a\choose 2}\cdot{b \choose 2}\choose k},
\end{align*}
and so $G_2^D(a,b)\sim (a^4\cdot b^4)$.
We know from Proposition \ref{proposition:dcounting} that $D_2(a,b)={a\choose 2}\cdot{b\choose 2}\sim (a^2\cdot b^2)$.
Therefore, given a random grid-labelled graph with grid dimensions $a,b\ge 2$ and with $k=2$ edges, a.a.s. it will not satisfy the degree criterion and is therefore not separable.

Setting $a=2,$ we find that
\begin{align*}
G_3^D(2,b)&={2\cdot{b \choose 2}\choose 3}\\
&\sim (b^2)^3\\ 
&= b^6,
\end{align*}
and 
\begin{align*}
G_4^D(2,b)&={2{b \choose 2}\choose 4}\\
&\sim (b^2)^4 \\
&= b^8.
\end{align*}
In comparison,
\begin{align*}
D_3(2,b)&=2\cdot {b\choose 3}\sim b^3,
\end{align*}
and 
\begin{align*}
D_4(2,b)&=3\cdot {b\choose 3}+9\cdot{b\choose 4}\sim b^4.
\end{align*}
We can conclude that for a random grid-labelled graph with grid dimensions $a=2,b\ge2$ and $k=3,4$ edges, a.a.s. it will not satisfy the degree criterion and is therefore not separable.
\end{proof}
In order to find a general form for $D_k(a,b)$, more sophisticated techniques from enumerative combinatorics will need to be employed.
On the basis of what we have been able to prove so far however, it is safe to conjecture that for any grid-labelled graph $G_l^{a,b}$, a.a.s. $G_l^{a,b}\not\in\mathcal{S}$. In the next section we will apply the matrix realignment criterion to grid-labelled graphs, in order to find a graph theoretic interpretation of the criterion.

\section{The matrix realignment criterion}
\label{section:thematrixrealignmentcriterion}
\subsection{Definitions}
The matrix realignment criterion is defined in terms of the Ky Fan norm of the \emph{realigned} density matrix of the state under consideration. Let us recall these concepts.

\begin{definition}[Ky Fan Norm]
	The Ky Fan norm of a matrix $M$ is the quantity
	\begin{align*}
	\lVert M \rVert_K = \sum_i s_i(M), 
	\end{align*}
	where $s_i$ is the $i^{\text{th}}$  singular value of $M$.
\end{definition}
While in the definition we have denoted the Ky Fan norm by $\lVert \cdot \rVert_K$, when the context is clear we will denote it by $\lVert \cdot \rVert$. Let $M$ be an $m\times n$ matrix. Then its \emph{vectorization} is the $(m\cdot n)$-dimensional row vector
\begin{align*}
\text{vec}(M):=\begin{pmatrix}[M]_{11}&[M]_{12}&\dots&[M]_{1n}&[M]_{21}&\dots&[M]_{mn}\end{pmatrix}.
\end{align*}
The vectorization of the blocks of a matrix are used to realign it:
\begin{definition}[Realigned matrix]
	Let $M$ be an $m\times m$ block matrix with $n\times n$ blocks $M_{ij}$. The \emph{realignment} of $M$ with respect to $n$ is the $m^2\times n^2$ matrix
	\emph{
		\begin{align*}
		R_n(M):=
		\begin{pmatrix}
		\text{vec}(M_{11})\\
		\vdots\\
		\text{vec}(M_{m1})\\
		\vdots\\
		\text{vec}(M_{1m})\\
		\vdots\\
		\text{vec}(M_{mm})\\
		\end{pmatrix}.
		\end{align*}
	}
\end{definition}

The following theorem describes the entanglement criterion.

\begin{theorem}[Matrix realignment criterion \cite{MR}]
	Let $\rho\in\mathbb{C}^{m}\otimes \mathbb{C}^n$ be the density matrix of a bipartite quantum state. If $\lVert R_n(\rho)\rVert_K > 1$, then $\rho$ is entangled.
\end{theorem}
\subsection{Realignment of combinatorial Laplacian matrices}
The density matrices we consider in this work are real matrices. It is well known that for any real matrix $M$, 
\begin{align*}
s_i(M)=\sqrt{\lambda_i(M^T\cdot M)}=\sqrt{\lambda_i(M\cdot M^T)},
\end{align*}
where $\lambda_i(M)$ denotes the $i^\text{th}$ eigenvalue of a matrix $M$. The fact that $M^T\cdot M$ (resp. $M\cdot M^T$) is real and symmetric implies that $s_i(M)\in\mathbb{R}$. Thus, to apply the matrix realignment criterion to the bipartite quantum state $\rho$ acting on $\mathbb{C}^{m}\otimes \mathbb{C}^n$ we will require the eigenvalues of the matrices $R_n(\rho)^T\cdot R_n(\rho)$ and $R_n(\rho)\cdot R_n(\rho)^T$. The non-zero eigenvalues of both matrices are identical. Therefore, without loss of generality, we need only consider the eigenvalues of $R_n(\rho)\cdot R_n(\rho)^T$. In particular, we will determine the form of $R_n(\rho)\cdot R_n(\rho)^T$ when $\rho$ is the density matrix of a grid-labelled graph.

Xie \emph{et al.} \cite{Xie2013} briefly study the matrix realignment criterion for combinatorial Laplacian matrices and present a structural entanglement criterion. In what follows we find a general form for the realigned Laplacian, and use this to find an algebraic entanglement criteria for grid-labelled graphs. Let $G_l^{a,b}$ be a grid-labelled graph. Then,
\begin{align*}
R_b(L(G_l^{a,b}))&=R_b\left(D(G_l^{a,b})-A(G_l^{a,b})\right)\\
&=R_b\left(D(G_l^{a,b})\right)-R_b\left(A(G_l^{a,b})\right).
\end{align*}
Clearly,
\begin{align*}
R_b(L(G_l^{a,b}))\cdot R_b(L(G_l^{a,b}))^T&=(R_b(D(G_l^{a,b}))-R_b(A(G_l^{a,b})))\cdot(R_b(D(G_l^{a,b}))-R_b(A(G_l^{a,b})))^T\\
&=R_b(D(G_l^{a,b}))\cdot R_b(D(G_l^{a,b}))^T\\&~~-R_b(D(G_l^{a,b}))\cdot R_b(A(G_l^{a,b}))^T\\&~~-R_b(A(G_l^{a,b}))\cdot R_b(D(G_l^{a,b}))^T\\&~~+R_b(A(G_l^{a,b}))\cdot R_b(A(G_l^{a,b}))^T.
\end{align*}
We will now work out each term in the sum. By definition,
\begin{align*}
D(G_l^{a,b})=
\text{diag}(d((1,1)),d((1,2)),\dots,d((a,b))),
\end{align*}
so
\begin{align*}
R_b(D(G_l^{a,b}))&=\begin{pmatrix}
d((1,1))&0&\dots&0&d((1,2))&0&\dots&\dots&0&d((1,b))\\
0&&&&0&&&&&0\\
\vdots&&&&\vdots&&&&&\vdots\\
0&&&&0&&&&&0\\
d((2,1))&0&\dots&0&d((2,2))&0&\dots&\dots&0&d((2,b))\\
0&&&&0&&&&&0\\
\vdots&&&&\vdots&&&&&\vdots\\
\vdots&&&&\vdots&&&&&\vdots\\
0&&&&0&&&&&0\\
d((a,1))&0&\dots&0&d((a,2))&0&\dots&\dots&0&d((a,b))\\
\end{pmatrix}.
\end{align*}
Hence,
	\begin{align*}
	&R_b(D(G_l^{a,b}))\cdot R_b(D(G_l^{a,b}))^T=\\
	\sum_{j=1}^n&\begin{pmatrix} d((1,j))^2&0&\dots&0&d((1,j))d((2,j))&0&\dots&\dots&0&d((1,j))d((b,j))\\
	0&&&&0&&&&&0\\
	\vdots&&&&\vdots&&&&&\vdots\\
	0&&&&0&&&&&0\\
	d((2,j))d((1,j))&0&\dots&0&d((2,j))^2&0&\dots&\dots&0&d((2,j))d((b,j))\\
	0&&&&0&&&&&0\\
	\vdots&&&&\vdots&&&&&\vdots\\
	\vdots&&&&\vdots&&&&&\vdots\\
	0&&&&0&&&&&0\\
	d((a,j))d((1,j))&0&\dots&0&d((a,j))d((2,j))&0&\dots&\dots&0&d((a,j))d((b,j))
	\end{pmatrix},
	\end{align*}
In order to find a general form for the realigned adjacency matrix, we will need some definitions.\begin{definition}[Row subgraph]
	Let $G_l^{a,b}$ be a grid-labelled graph. If a grid-labelled graph $H_c^{2,b}\subseteq G_l^{a,b}$ has vertex set
	\begin{align*}
	V(H_c^{2,n})=\{(p,q):(p,q)\in V(G_l^{a,b}),p=i \lor p=j\},
	\end{align*}
	edge set
	\begin{align*}
	E(H_c^{2,n})=\{\{(p,q),(r,s)\}:\{(p,q),(r,s)\}\in E(G_l^{a,b}),p=i\land r=j\},
	\end{align*}
	and labelling $c:[a]\times[b]\rightarrow [2]\times [b]$ defined such that 
	\begin{align*}
	c(l^{-1}((p,q)))\mapsto \begin{cases}(1,q)&\text{if } p=i;\\(2,q)&\text{if } p=j,\end{cases}
	\end{align*}
	then it is called a \emph{row subgraph} of $G_l^{a,b}$, denoted $R_{i,j}(G_l^{a,b})$, and simply $R_{i,j}$ when $G_l^{a,b}$ is clear from the context.
\end{definition}
Equivalently, $R_{i,j}$ is the graph with vertex set populated with all vertices from rows $i$ and $j$ of $G_l^{a,b}$, and edge set populated by all edges strictly between those rows, and endpoints in different rows. Note that for any grid-labelled graph, all edges of any row subgraph $R_{i,i}$ are horizontal, and also for any $R_{i,j}$,
\begin{align*}
R_{i,j}=\Gamma(R_{i,j}).
\end{align*}
In Figure \ref{fig:rowsubgraphs} we illustrate several row subgraphs.
\begin{figure}
	\centering
	\includegraphics[width=1.0\textwidth]{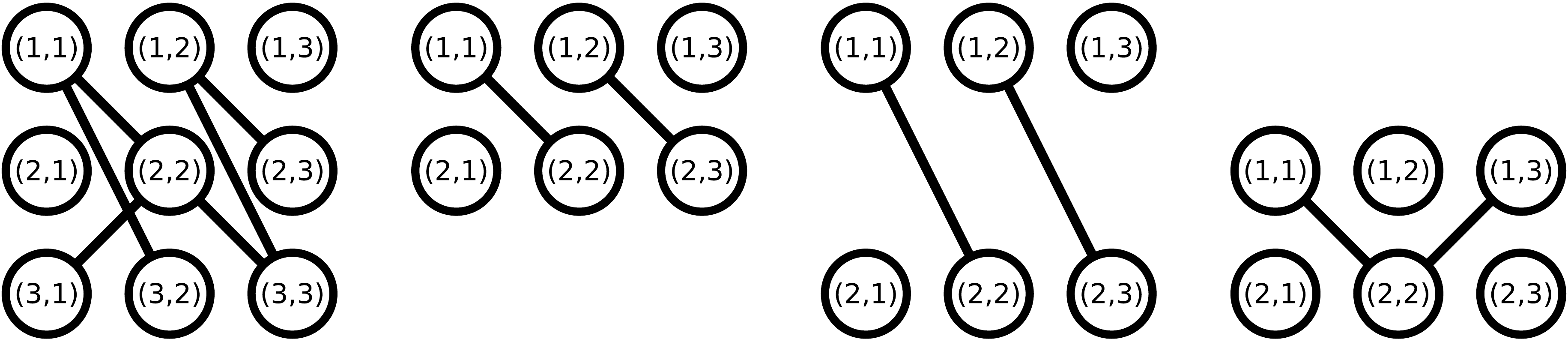}
	\caption{From left to right, $G_l^{3,3}$, and $R_{1,2}, R_{1,3}, R_{3,2}\subseteq G_l^{3,3}$.}
	\label{fig:rowsubgraphs}
\end{figure}
\begin{definition}[Intersection]
	Let $G_l^{a,b}$ and $H_l^{a,b}$ be grid-labelled graphs. Their \emph{intersection} is the grid-labelled graph $G_l^{a,b}\cap H_l^{a,b}$ with edge set $E(G_l^{a,b})\cap E(H_l^{a,b})$.
\end{definition}
The following lemma is useful.
\begin{lemma}
	Let $G_l^{a,b}$ be a grid-labelled graph with adjacency matrix $A$. Then,
	\emph{
		\begin{align*}
		\left[R_b(A)\cdot R_b(A)^T\right]_{ij}=r(i\text{ mod }a,\lceil i/a\rceil;j\text{ mod }a,\lceil j/a\rceil),
		\end{align*}
	}
	where 
	\begin{align*}
	r(i,j;k,l):=|E(R_{i,j}\cap R_{k,l})|.
	\end{align*}
\end{lemma}
\begin{proof}
	By definition,
	\begin{align*}
	R_b(A(G_l^{a,b}))\cdot R_b(A(G_l^{a,b}))^T=\begin{pmatrix}\text{vec}(A_{11})\\\text{vec}(A_{21})\\\vdots\\\text{vec}(A_{aa})\\\end{pmatrix}\cdot\begin{pmatrix}\text{vec}(A_{11})^T~\text{vec}(A_{21})^T~\dots~\text{vec}(A_{aa})^T\end{pmatrix},
	\end{align*}
	where the $n\times n$ blocks $A_{ij}$ of $A$ have the form
	\begin{align*}
	A_{ij}=
	\begin{pmatrix}
	e(i,1;j;1)&e(i,1;j,2)&\dots&e(i,1;j,b)\\
	e(i,2;j;1)&e(i,2;j,2)&\dots&e(i,2;j,b)\\
	\vdots&\vdots&\ddots&\vdots\\
	e(i,a;j;1)&e(i,a;j,2)&\dots&e(i,a;j,b)\\
	\end{pmatrix},
	\end{align*}
	where the function $e(p,q;r,s)$ is equal to $1$ if there is an edge $\{(p,q),(r,s)\}\in E(G_l^{a,b})$, and is equal to $0$ otherwise. Hence, the matrix $A_{ij}$ encodes the edges from row $i$ to row $j$. Clearly
	\begin{align*}
	\text{vec}(A_{ij})\cdot\text{vec}(A_{kl})^T
	&=\sum_{p,q=1}^{b}e(i,p;j,q)e(k,p;l,q)\\
	&=|E(R_{i,j}\cap R_{k,l})|=:r(i,j;k,l).
	\end{align*}
	Therefore,
	\begin{align*}
	\left[R_b(A(G_l^{a,b})\cdot R_b(A(G_l^{a,b})^T\right]_{ij}&=\text{vec}(A_{i\text{mod }a,\lceil i/a\rceil})\cdot\text{vec}(A_{j\text{mod }a,\lceil j/a\rceil})^T\\
	&=r(i\text{ mod }a,\lceil i/a\rceil;j\text{ mod }a,\lceil j/a\rceil).
	\end{align*}
\end{proof}
If we consider graphs with only diagonal edges then it will make the analysis simpler: if the graph has no horizontal or vertical edges we know that $r(i,i;k,l)=r(i,j;k,k)=0$. This means that the matrices $R_b(A(G_l^{a,b}))$ and $R_b(D(G_l^{a,b}))$ have no non-zero entries in common, that is, $[R_b(A(G_l^{a,b}))]_{ij}=0$ if and only if $[R_b(D(G_l^{a,b}))]_{ij}\neq0$. The products $R_b(A(G_l^{a,b}))\cdot R_b(D(G_l^{a,b}))^T$ and $R_b(D(G_l^{a,b}))\cdot R_b(A(G_l^{a,b}))^T$ are therefore equal to $0$.

Hence, \begin{align*}
R_b(L(G_l^{a,b}))\cdot R_b(L(G_l^{a,b}))^T=R_b(D(G_l^{a,b}))\cdot R_b(D(G_l^{a,b}))^T+R_b(A(G_l^{a,b}))\cdot R_b(A(G_l^{a,b}))^T.
\end{align*}
Since both terms in the sum have no non-zero elements in common, it is clear that the set of non-zero eigenvalues of $R_b(L(G))\cdot R_b(L(G))^T$ is identical to the set of non-zero eigenvalues of the matrix
$\mathcal{D}(G_l^{a,b})\oplus\mathcal{A}(G_l^{a,b}),
$
where $\mathcal{D}(G_l^{a,b})$ and $\mathcal{A}(G_l^{a,b})$ are the \emph{degree structure} and \emph{adjacency structure} matrices of $G_l^{a,b}$ respectively, which we now define.
\begin{definition}[Degree structure matrix]
	Let $G_l^{a,b}$ be a grid-labelled graph. Then the \emph{degree structure matrix} of $G_l^{a,b}$ is the $a\times a$ matrix with entries
	\begin{align*}
	\left[\mathcal{D}(G_l^{a,b})\right]_{i,j}:=\sum_{p=1}^{b}d((i,p))d((j,p)),
	\end{align*}
	where $d((i,j))$ is the degree of the vertex $(i,j)\in V(G_l^{a,b})$.
\end{definition}
\begin{definition}[Adjacency structure matrix]
	Let $G_l^{a,b}$ be a grid-labelled graph with all edges diagonal. Then the \emph{adjacency structure matrix} of $G_l^{a,b}$ is the $a(a-1)\times a(a-1)$ matrix
	\begin{align*}
	\mathcal{A}(G_l^{a,b}):=
	\begin{pmatrix}
	r(1,2;1,2)&r(1,2;1,3)&\dots&r(1,2;2,1)&\dots&r(1,2;a,a-1)\\
	r(1,3;1,2)&r(1,3;1,3)&\dots&r(1,3;2,1)&\dots&r(1,3;a,a-1)\\
	\vdots&\vdots&&\vdots&&\vdots\\
	r(2,1;1,2)&r(2,1;1,3)&\dots&r(2,1;2,1)&\dots&r(2,1;a,a-1)\\
	\vdots&\vdots&&\vdots&&\vdots\\
	r(a,a-1;1,2)&r(a,a-1;1,3)&\dots&r(a,a-1;2,1)&\dots&r(a,a-1;a,a-1)
	\end{pmatrix},
	\end{align*}
	where $r(i,j;k,l):=|E(R(i,j))\cap E(R(k,l))|$.
\end{definition}

It is obvious that the degree and adjacency structure matrices of any grid-labelled graph $G_l^{a,b}$ can be obtained from $R_b(D(G_l^{a,b}))\cdot R_b(D(G_l^{a,b}))^T$ and $R_b(A(G_l^{a,b}))\cdot R_b(A(G_l^{a,b}))^T$ respectively, by applying suitable permutations.

We have thus proved the following theorem, where the factor $1/2e$ is for normalisation.
\begin{theorem}
	Let $G_l^{a,b}$ be a grid-labelled graph with $e$ edges, all diagonal. Let \emph{ $\Lambda:=\text{sp}(\mathcal{A}(G_l^{a,b}))$} and \emph{$\Theta:=\text{sp}(\mathcal{D}(G_l^{a,b}))$} be the set of non-zero eigenvalues of the adjacency structure and degree structure matrices of $G_l^{a,b}$. Then
	\emph{
		\begin{align*}
		\lVert R_b(\rho(G_l^{a,b})) \rVert_K = \frac{1}{2e}\left( \sum_{\lambda\in\Lambda}\sqrt{\lambda}+\sum_{\theta\in\Theta}\sqrt{\theta}\right).
		\end{align*}
	}
\end{theorem}
In the next section we will use this theorem to demonstrate an infinite family of entangled quantum states whose entanglement is not detected by the matrix realignment criterion. 
\subsection{Failure of the matrix realignment criterion}

In order to proceed, we require the following two definitions.
\begin{definition}[Row orthogonality]
	Let $G_l^{a,b}$ be a grid-labelled graph. We say that $G_l^{a,b}$ is \emph{row orthogonal} if for all $i,j\in[m]$ and $k\in[n]$,
	\begin{align*}
	d((i,k))d((j,k))\not=0
	\end{align*}
	if and only if $i=j$.
\end{definition}
Equivalently, two rows of vertices are orthogonal if each column (two vertices aligned vertically) has at least one isolated vertex. A grid-labelled graph is orthogonal if all of its vertex rows are pairwise orthogonal.
\begin{definition}[Singleton edge]
	Let $G_l^{a,b}$ be a grid-labelled graph. An edge $\{v_1,v_2\}\in E(G_l^{a,b})$ is described as a \emph{singleton} if 
	\begin{align*}
	d(v_1)=d(v_2)=1.
	\end{align*}
\end{definition}
We can now prove the following theorem.
\begin{theorem}
\label{theorem:roworthogsingletondiag}
	Let $G_l^{2,b}$ be a row orthogonal grid-labelled graph with $e$ edges, such that all edges are singleton and diagonal. If $e\ge 4$ then $\lVert R_b(\rho(G_l^{2,b}))\rVert_K\le 1$.
\end{theorem}
\begin{proof}
	Since $G_l^{2,b}$ has $a=2$ rows,
	\begin{align*}
	\mathcal{A}(G_l^{2,b})=\begin{pmatrix}r(1,2;1,2)&r(1,2;2,1)\\r(2,1;1,2)&r(2,1;2,1)\end{pmatrix},
	\end{align*}
	and
	\begin{align*}
	\mathcal{D}(G_l^{2,b})=
	\sum_{j=1}^{b}\begin{pmatrix}
	d((1,j))^2&d((1,j))d((2,j))\\
	d((2,j))d((1,j))&d((2,j))^2	
	\end{pmatrix}.
	\end{align*}
	Since $G_l^{2,b}$ is row orthogonal, 
	\begin{align*}
	r(1,2;2,1)=r(2,1;1,2)=0,
	\end{align*}
	and
	\begin{align*}
	\sum_{j=1}^b d((1,j))d((2,j))=\sum_{j=1}^b d((2,j))d((1,j))=0,
	\end{align*}
	so both $\mathcal{A}(G_l^{2,b})$ and $\mathcal{D}(G_l^{2,b})$ are diagonal. Hence,
	\begin{align*}
	\lVert R_b(\rho(G_l^{a,b}))\rVert_K&=\frac{\sqrt{r(1,2;1,2)}+\sqrt{r(2,1;2,1)}+\sqrt{\sum_{j=1}^b d((1,j))^2}+\sqrt{\sum_{j=1}^b d((2,j))^2}}{2e}\\
	&=\frac{2\sqrt{e}+\sqrt{\sum_{j=1}^b d((1,j))^2}+\sqrt{\sum_{j=1}^b d((2,j))^2}}{2e}.
	\end{align*}
	Since all edges are singleton, all degrees are either $0$ or $1$. This means that
	\begin{align*}
	\lVert\rho(G_l^{a,b})\rVert_K&=\frac{4\sqrt{e}}{2e}\\
	&=\frac{2}{\sqrt{e}},
	\end{align*}
	from which the theorem follows.
\end{proof}

The fact that such grid-labelled graphs are row orthogonal means that they do not satisfy the degree criterion. Since the degree criterion is necessary and sufficient for separability for grid-labelled graphs with $a=2$ (Theorem \ref{theorem:degrees}), these grid-labelled graphs describe quantum states that are entangled. An example of a family of grid-labelled graphs with the properties described in Theorem \ref{theorem:roworthogsingletondiag} are those of type $(2,2k)$ with edge sets
\begin{align*}
E_{k}:=\bigcup_{i=1}^{k}\Big\{\{(1,2i-1),(2,2i)\}\Big\}
\end{align*}
for $k\ge 4$. For example, the $k=4$ case corresponds to the quantum state with density matrix
\begin{align*}
\rho=&\frac{1}{4}\big(|1,1;2,2\rangle\langle1,1;2,2| + |1,3;2,4\rangle\langle1,3;2,4|\\ &+ |1,5;2,6\rangle\langle1,5;2,6| + |1,7;2,8\rangle\langle1,7;2,8|\big)\
\end{align*}
acting on $\mathbb{C}^2\otimes \mathbb{C}^8$.

In this section we have applied the matrix realignment criterion to the grid-labelled graphs. We showed that entanglement in grid-labelled graphs can be detected using the eigenvalues of their adjacency and degree structure matrices. We used this result to construct a family of entangled quantum states that are not detected as entangled by the matrix realignment criterion. Let us conclude by explicitly calculating the adjacency and degree structure matrices for an example grid-labelled graph, in the hope that it will be illuminating for the reader.
 
\begin{example}
\emph{Let us calculate the degree and adjacency structure matrices for an example graph.}
\begin{align*}
	\mathcal{D}\left(\vcenter{\hbox{\includegraphics[scale=0.4]{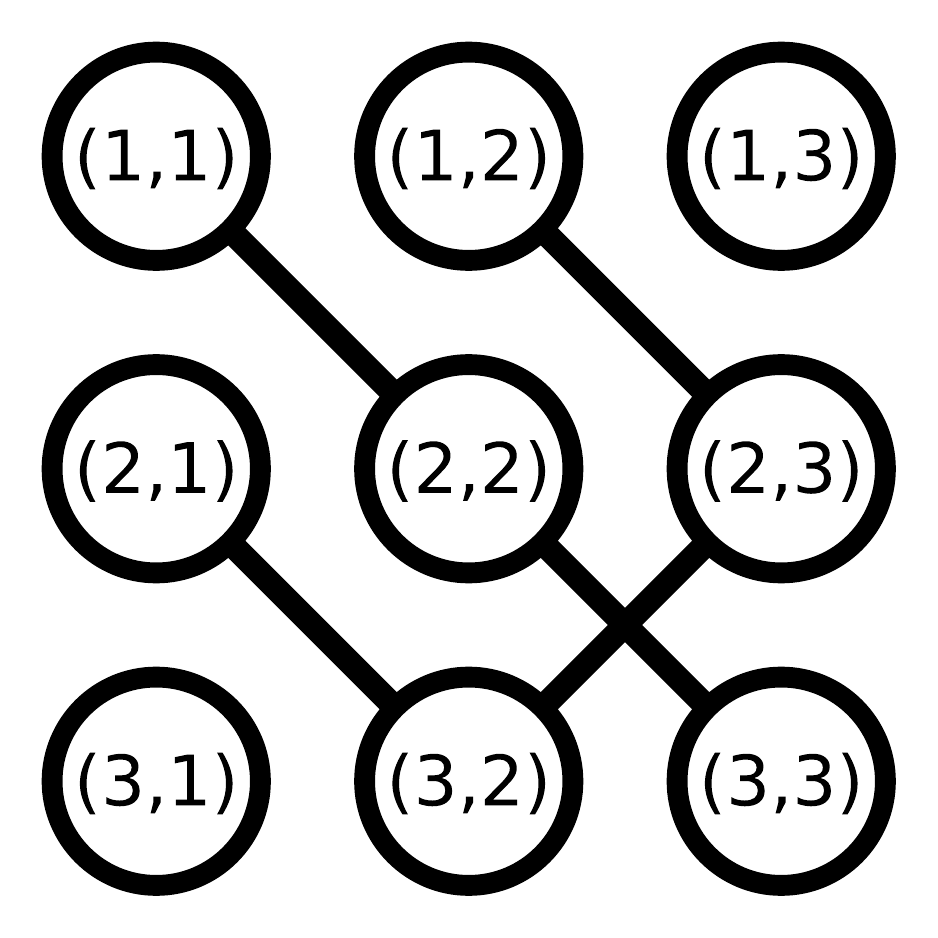}}}\right)&=\begin{pmatrix}\vcenter{\hbox{\includegraphics[scale=0.25]{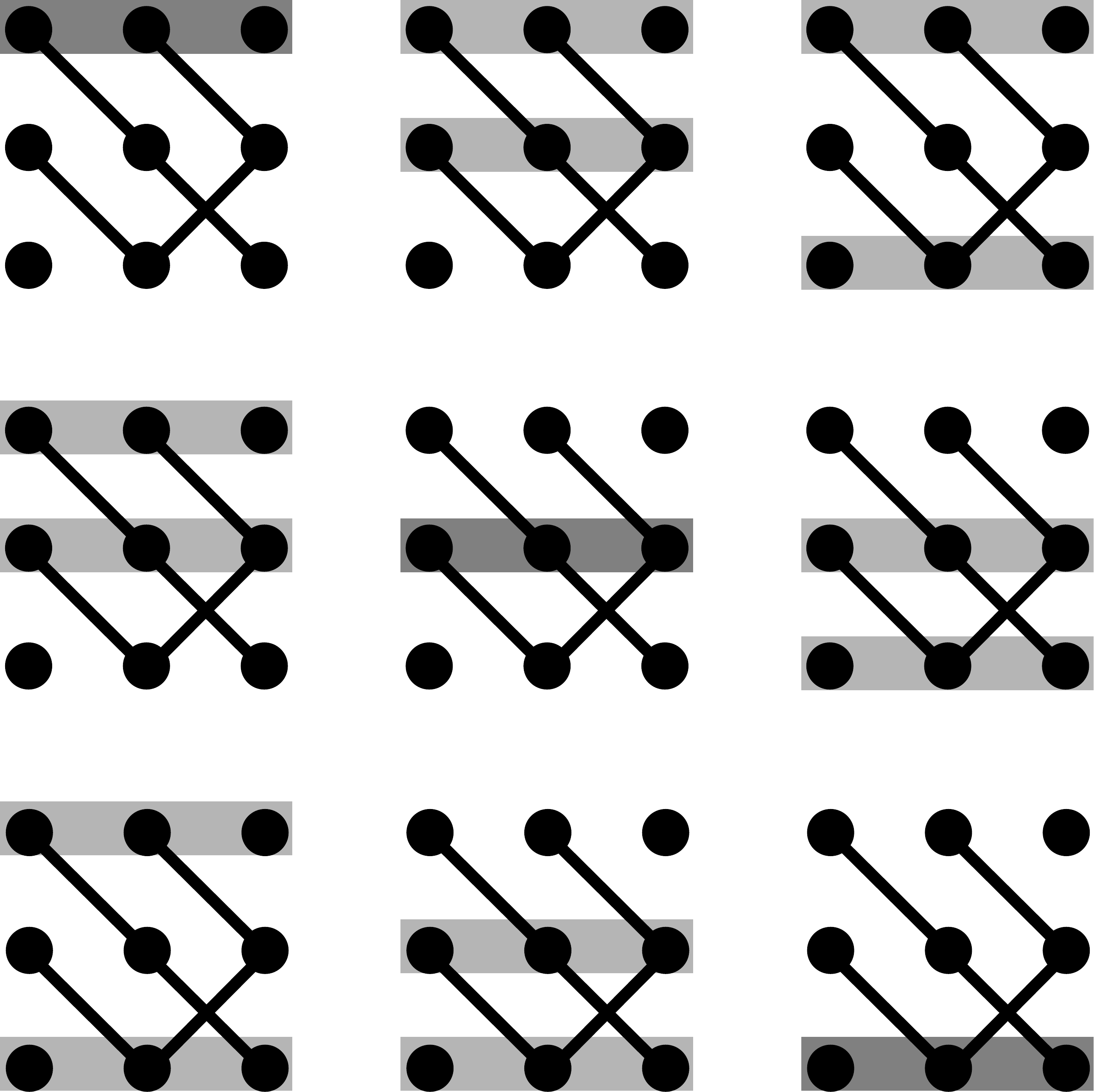}}}\end{pmatrix}\\&=\begin{pmatrix}2&3&2\\3&9&6\\2&6&5\end{pmatrix}
	\end{align*}
	\emph{As illustrated above, finding the entries of the degree structure matrix of a $m\times n$  grid-labelled graph amounts to taking the pairwise inner product of the $m$ $n$-dimensional vectors which have as their entries the degrees of the vertices in that row.}

	\emph{It can be seen that the entries of the adjacency structure matrix are in some sense equal to the ``overlap'' of each row subgraph, as we now illustrate pictorially:}
	\begin{align*}
	\mathcal{A}\left(\vcenter{\hbox{\includegraphics[scale=0.4]{mrexample}}}\right)&=\begin{pmatrix}\vcenter{\hbox{\includegraphics[scale=0.16]{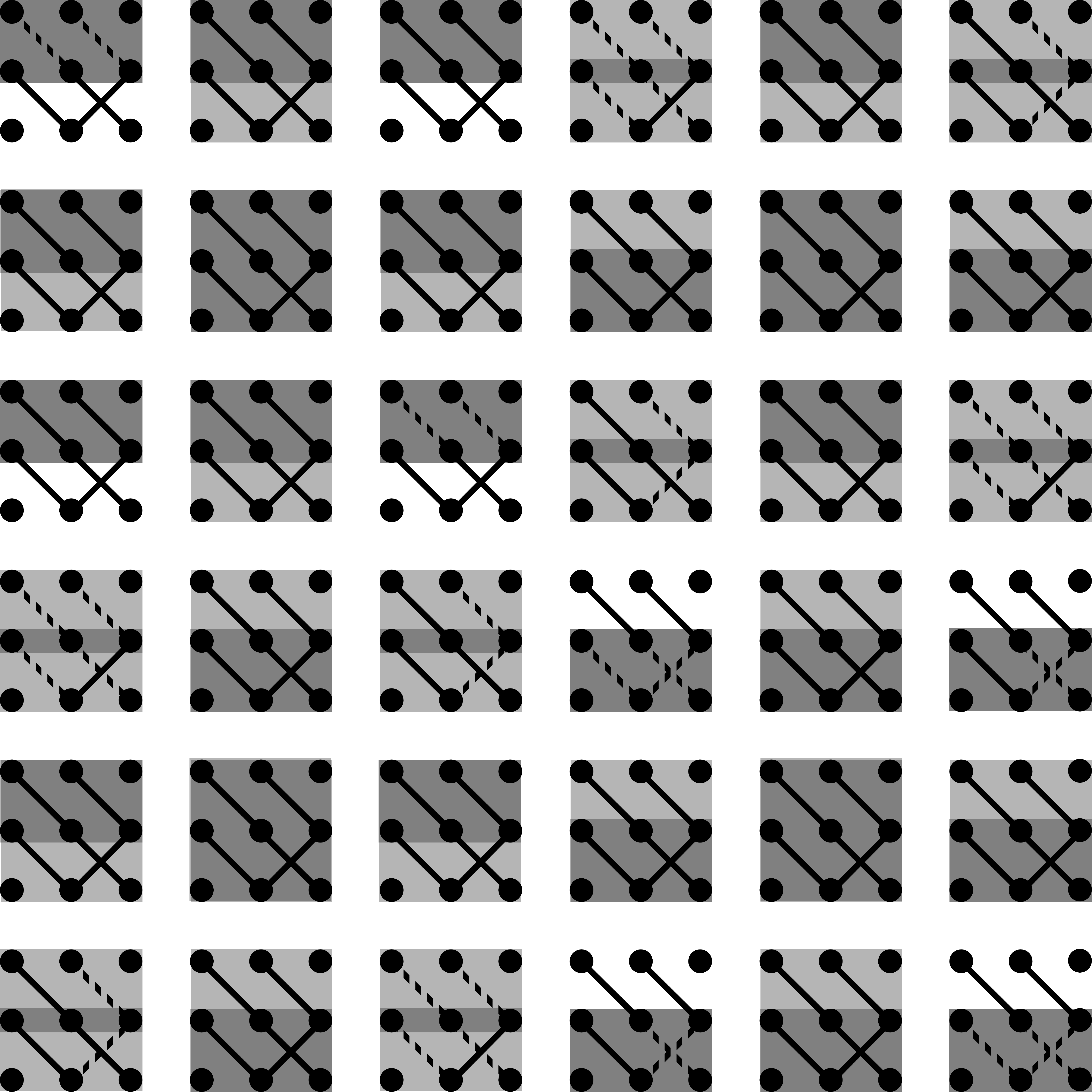}}}\end{pmatrix}\\&=\begin{pmatrix}2&0&0&2&0&1\\0&0&0&0&0&0\\0&0&2&1&0&2\\2&0&1&3&0&2\\0&0&0&0&0&0\\1&0&2&2&0&3\end{pmatrix},
\end{align*}
\emph{where shaded regions denote the row subgraphs being compared. For example, the entry at index $(6,3)$ compares the row subgraphs $R_{1,2}$ and $R_{2,3}$, which have $2$ edges in their intersection. The entry at index $(3,1)$ compares $R_{1,2}$ and $R_{2,1}$, which have $0$ edges in their intersection.}
\end{example}
\section{Summary and open problems}
\label{section:summaryandopenproblems}

In this work we have introduced a new combinatorial object, the \emph{grid-labelled graph}. We have shown that grid-labelled graphs can be used to represent a discrete family of quantum states, and that study of their structural properties helps in determining if these states are entangled. We have used this graphical ``handle'' on separability that grid-labelled graphs afford us to explore the limits of two well known entanglement criteria: Peres-Horodecki and matrix realignment.

In Section \ref{section:thedegreecriterion} we ported the Peres-Horodecki criterion to our framework in the form of the degree criterion of Braunstein \emph{et al.} \cite{Braunstein2006b}, and introduced the concept of a separable decomposition of a graph. We used the fact that the degree criterion is necessary and sufficient in $\mathbb{C}^2\otimes \mathbb{C}^n$ for arbitrary $n$ in combination with the idea of separable decompositions to build a family of quantum states, parameterised by the \emph{stratified} grid-labelled graphs, for which the Peres-Horodecki criterion is necessary and sufficient in all bipartite dimensions.

The matrix realignment criterion is explored in Section \ref{section:thematrixrealignmentcriterion}. We are able to simplify the criterion for grid-labelled graphs, showing that it can be defined in terms of the eigenvalues of the \emph{degree} and \emph{adjacency} structure matrices of the graph. These are matrices that we have introduced in this paper. At the end of Section \ref{section:thematrixrealignmentcriterion} we tried out this new tool on grid-labelled graphs that have diagonal degree and adjacency structure matrices. This was fruitful, and led us to discover a family of entangled quantum states that are not detected by the matrix realignment criterion. This highlights how grid-labelled graphs can be used as a tool for generating ``pathological'' density matrices which can be used to refine existing entanglement criteria.

In Section \ref{section:separabilityin3x3} we classified all $3\times 3$ grid-labelled graphs that satisfy the degree criterion. We showed that such graphs satisfy the degree criterion if and only if they are equal to the union of \emph{building-block} grid-labelled graphs. A natural extension to our work would be an exploration of how this scales to larger grid dimensions: what are the building-blocks for arbitrary grid dimension? As discussed in the conclusion of Section \ref{section:separabilityin3x3}, it would be interesting to explore how the size of the building-block set grows as a function of grid size. Indeed, in our attempt to classify the $3\times 3$ degree criterion graphs, we uncovered two bound entangled states, $(B_4)_l^{3,3}$ and $(B_5)_l^{3,3}$. Indeed, $(B_4)_l^{3,3}$ is the state presented by Hildebrand \emph{et al.} (\cite{Hildebrand+2008}, Section 2.4) as a counter-example to the so called ``degree conjecture'', that the degree criterion is necessary and sufficient for all density matrices that are the combinatorial Laplacian matrix of a graph. We are not aware of any reference in the literature to bound entangled states of the form described by $(B_5)_l^{3,3}$. It is possible that these graphs are respectively members of two families of bound entangled states, obtained by ``scaling up'' the cross-hatch like structure of these $3\times 3$ cases to arbitrary bipartite dimensions. The structure of the edges of these grid-labelled graphs make it easy to see that higher dimensional cases for arbitrary $a\times b$ grid dimensions can be generated, and that they satisfy the degree criterion. The question remains as to whether these higher dimensional states remain entangled. Furthermore, it could be that an attempt to classify the building-blocks for higher dimensional grid-labelled graphs that satisfy the degree criterion would lead to new examples of such states that in turn belong to their own new families of bound entangled states. This would be an interesting result, and it is an enticing reason to explore the higher dimensional building blocks. At the close of this section we were able to apply the results to directly solve $(k,a,b)$-\textsc{GraphSeparability} for $k=1,2,3$ and arbitrary $m,n$. It would be fascinating to completely disregard grid dimension in the study of graphs that satisfy the degree criterion, and attempt to count the number of such graphs with a fixed number of edges, and no restriction on the boundaries of their corresponding grid. In Section \ref{subsection:enumeration} we present a framework for enumerating the graphs with a fixed number of edges in higher dimensional grids that satisfy the degree criterion. We were able to make some progress on the $1\le e\le4$ edge cases. Extending this analysis for graphs with more edges could be an interesting exercise in enumerative combinatorics. 

In the introduction we said that the quantum state corresponding to a grid-labelled graph is extremely simple. Now that we are equipped with the necessary mathemetical tools, we can discuss this in more technical detail. A grid-labelled graph of type $m\times n$ with $e$ edges corresponds to a uniform mixture of $e$ bipartite pure states from $\mathbb{C}^m\otimes \mathbb{C}^n$ with Schmidt rank equal to $2$: the edge states, as defined in Definition \ref{def:edgeStates}. The fact that neither one of the two well known entanglement criteria we covered in this paper reliably detects entanglement in such simple density matrices is interesting. Our preliminary investigations \cite{otfried} into another entanglement criteria known as the \emph{range criterion} \cite{rangecriterion} have shown that this is not necessary and sufficient for separability for grid-labelled graphs either. As we have stated in the introduction, it would be very interesting if we could use the grid-labelled graphs to re-prove NP-hardness of separability testing for two reasons. The first reason is that a proof of NP-hardness would likely utilise a reduction from an NP-complete problem about graphs such as subgraph isomorphism, or that of finding a clique. This would be novel because it would probably be much simpler than the existing proofs of NP-hardness \cite{Gurvits}, and would effectively ``discretize'' the problem. Density matrices are by their very nature continuous objects, which means that reasoning about them in the form of decision problems becomes unwieldy. In contrast, the objects we consider in this paper are discrete. The second reason that a proof of NP-hardness for testing separability of grid-labelled graphs would be interesting is that it would provide a very concrete, easy to describe, family of density matrices that are hard to test separability for. This would provide interesting insight into the complexity ``landscape'' of testing separability.

\section*{Acknowledgements}
JL acknowledges funding from the EPSRC Center for Doctoral Training in Delivering Quantum Technologies at UCL. SS is supported by the Royal Society, the U.K. EPSRC, and SJTU. We thank Laura Man\v{c}inska for useful discussions on mixed state entanglement and bound entanglement, David Roberson for help with various combinatorial and matrix theoretic problems, Otfried G\"{u}hne for computational work on verifying entanglement in $(B_4)^{3,3}$ and $(B_5)^{3,3}$, Erik Demaine and Jayson Lynch for discussion on the computational complexity of Crosses and Lassoes, Bhaskar DasGupta for discussions related to stratified grid-labelled graphs, and Carlo Sparaciari for his keen eye for detail. Most graphs in the paper were drawn with the \emph{yEd} graph editor \cite{yEd}.

\appendix
\section{Contribution tables and a single player board game}
\label{appendix:contributiontables}
In Section \ref{section:separabilityin3x3} we used the contribution table framework to characterise the $3\times 3$ grid-labelled graphs that satisfy the degree criterion. In this appendix we consider a single player board game called ``Crosses and Lassoes'', which comes directly from the contribution table framework. We first describe the game and a simple variation, then prove that these games address questions about grid-labelled graphs that satisfy the degree criterion.

\subsection{Crosses and Lassoes}
The objective of Crosses and Lassoes is to remove all pieces from the board by either ``crossing out'' or ``lassoing'' them. The game is played on an $m\times n$ board. There are two types of pieces: \emph{left} pieces, denoted with a `$\backslash$', and \emph{right} pieces, denoted with a `/'. Each square on the board can contain any number of pieces. A board with pieces on it is called a \emph{setup}. If it is possible to remove all pieces from a setup by performing a sequence of cross and lasso operations, then we say that the setup is \emph{clearable}. If it is not possible to clear a setup in such a way, then we say that the setup is \emph{unclearable}. Lasso and cross operations act on four pieces at a time, which must be in a \emph{rectangular formation}, that is, the pieces involved must be situated on the board such that they make up the corners of a rectangle.

\begin{figure}[h!]
	\centering
	\includegraphics[scale=0.2]{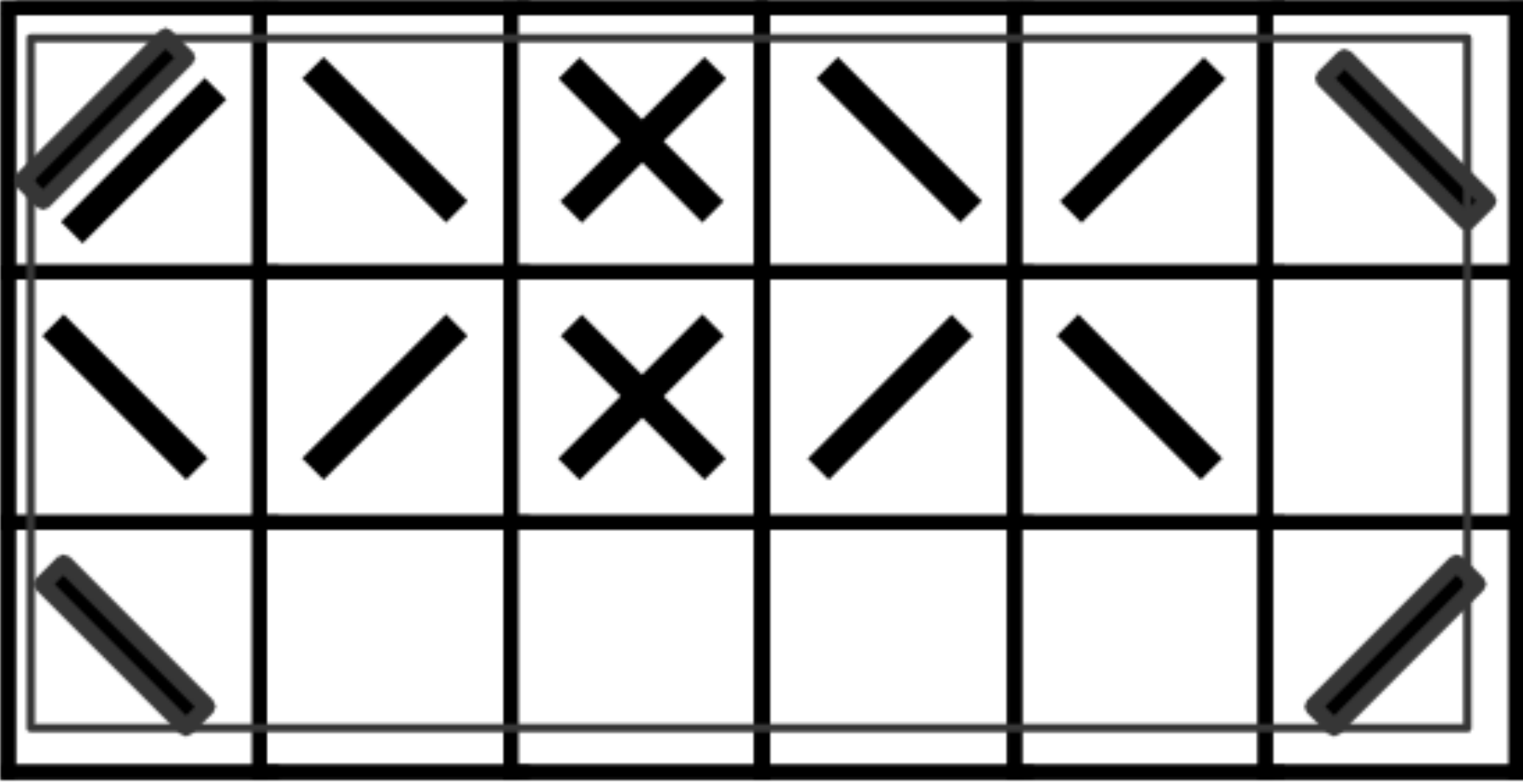}~
	\includegraphics[scale=0.2]{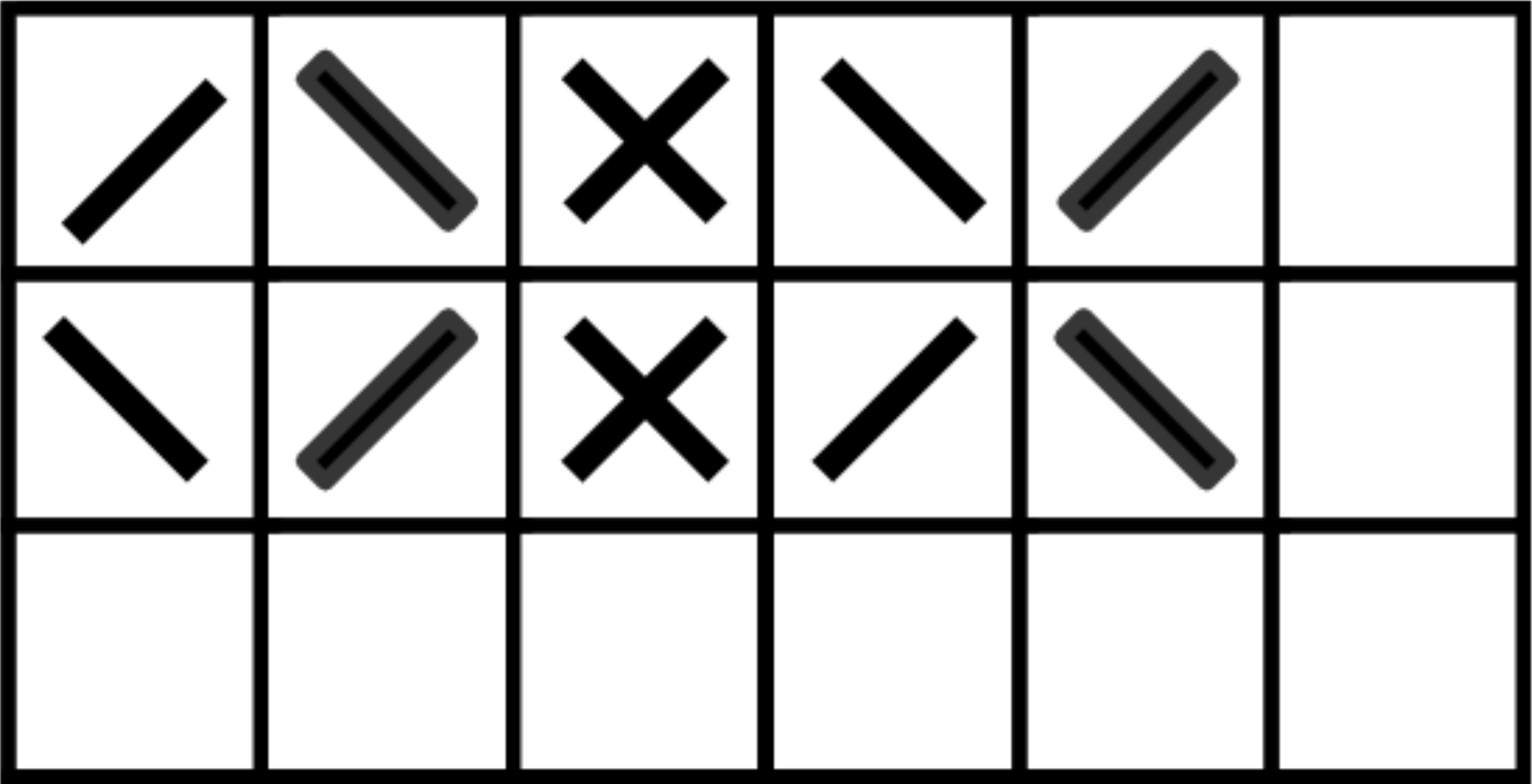}~
	\includegraphics[scale=0.2]{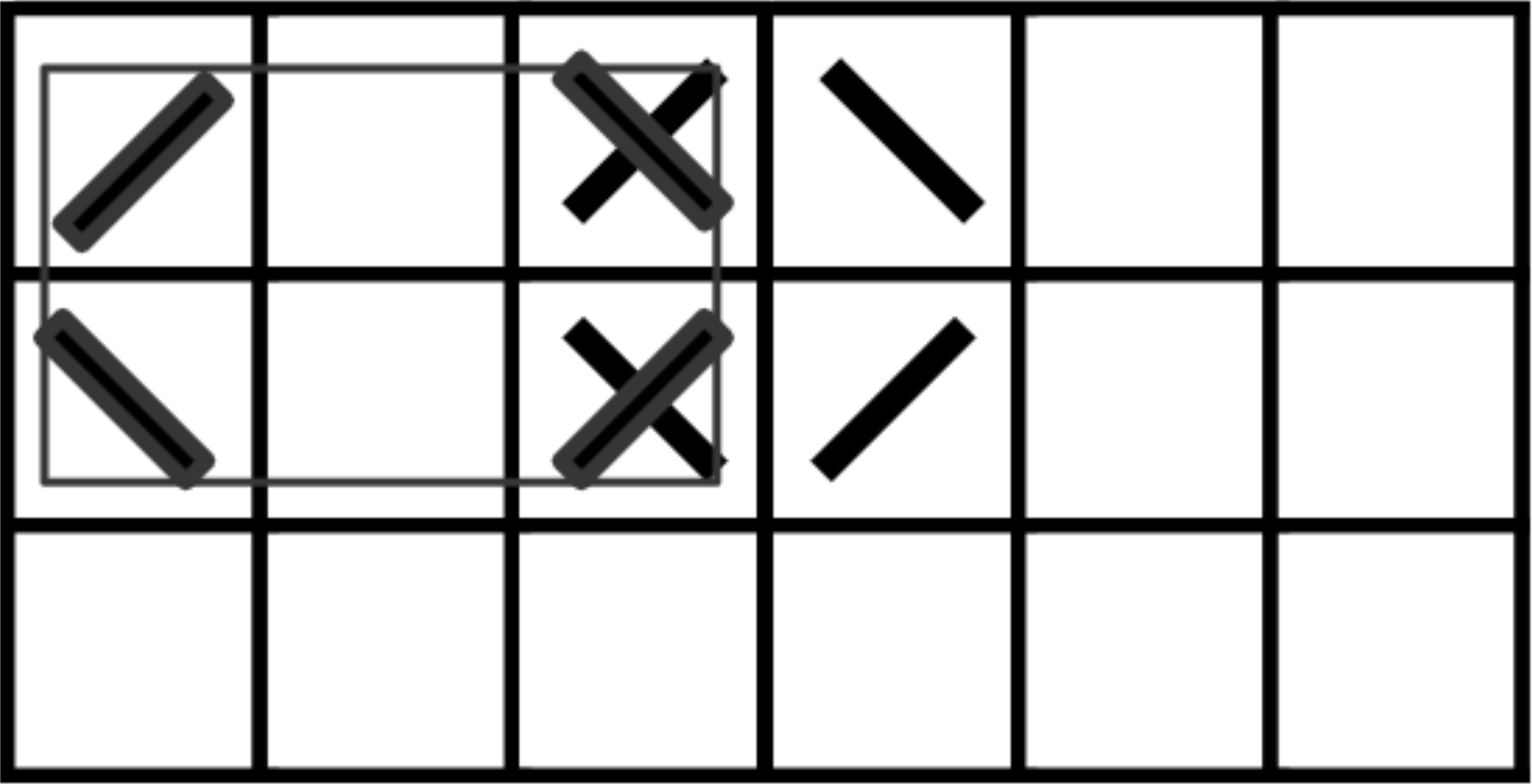}~
	\includegraphics[scale=0.2]{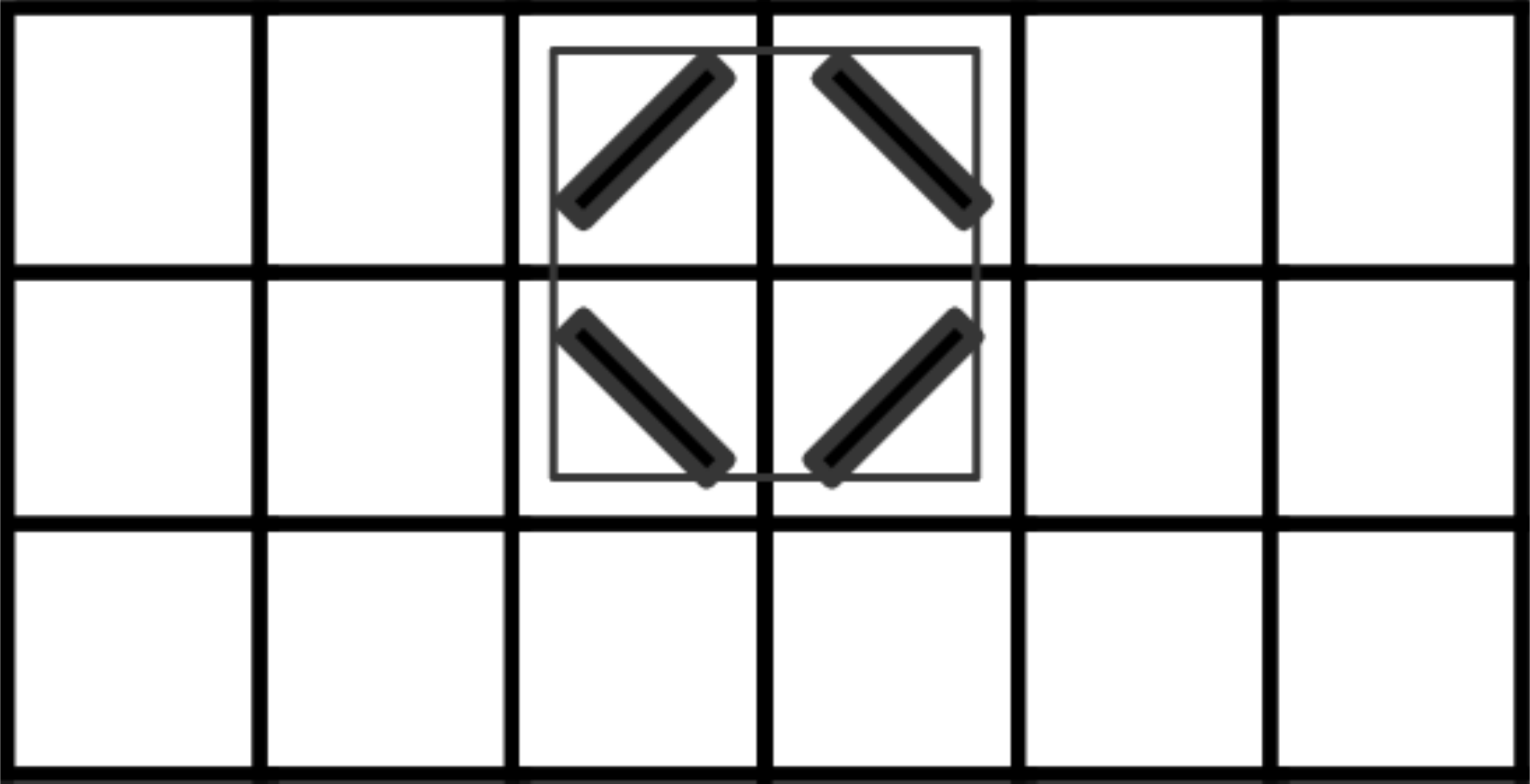}

	 \caption{A setup being cleared. From left to right, the operations performed are lasso, cross, lasso, lasso.}\label{fig:game}
\end{figure}

\subsubsection*{The lasso operation}
Four pieces in a rectangular formation can be removed by a lasso operation if, moving clockwise from the top left of the rectangle, the pieces are right, left, right and left pieces. In boards 1, 3, and 4 in Figure \ref{fig:game} 
 we illustrate four pieces being removed from the board by lasso operations.

\subsubsection*{The cross operation}
Four pieces in a rectangular formation can be removed by a lasso operation if, moving clockwise from the top left of the rectangle, the pieces are left, right, left and right pieces. In the setup second from the left in Figure \ref{fig:game}, we illustrate four pieces being removed from a board by the cross operation.

We can now define the following decision problems, which are clearly equivalent.

\begin{problem}
\textsc{Clearability}\\
\textit{Input: } \emph{An $m\times n$ setup.}\\
\textit{Question: } \emph{Is the setup clearable?}
\end{problem}
\begin{problem}
	\textsc{ContributionTableValidity}\\
	\textit{Input: } \emph{An $m\times n$ contribution table $T$.}\\
	\textit{Question: } \emph{Is there a grid-labelled graph $G_l^{m,n}$ that has contribution table $T$?}
\end{problem}

Instances of the \textsc{ContributionTableValidity} appeared several times in Section \ref{section:separabilityin3x3}, where we had to test if a particular contribution table had a corresponding graph. We leave the computational complexity of these decision problems as an open question.

\subsection{A different angle: combining setups}
Let us now describe an interesting variation on ``vanilla'' Crosses and Lassos. Let $A$ and $B$ be two setups, each on an $m\times n$ board. The setup $A+B$ is the setup with an $m\times n$ board with pieces from $A$ and $B$ combined.
For example,
\begin{align}
\label{eq:addition}
\vcenter{\hbox{\includegraphics[scale=\boardscale]{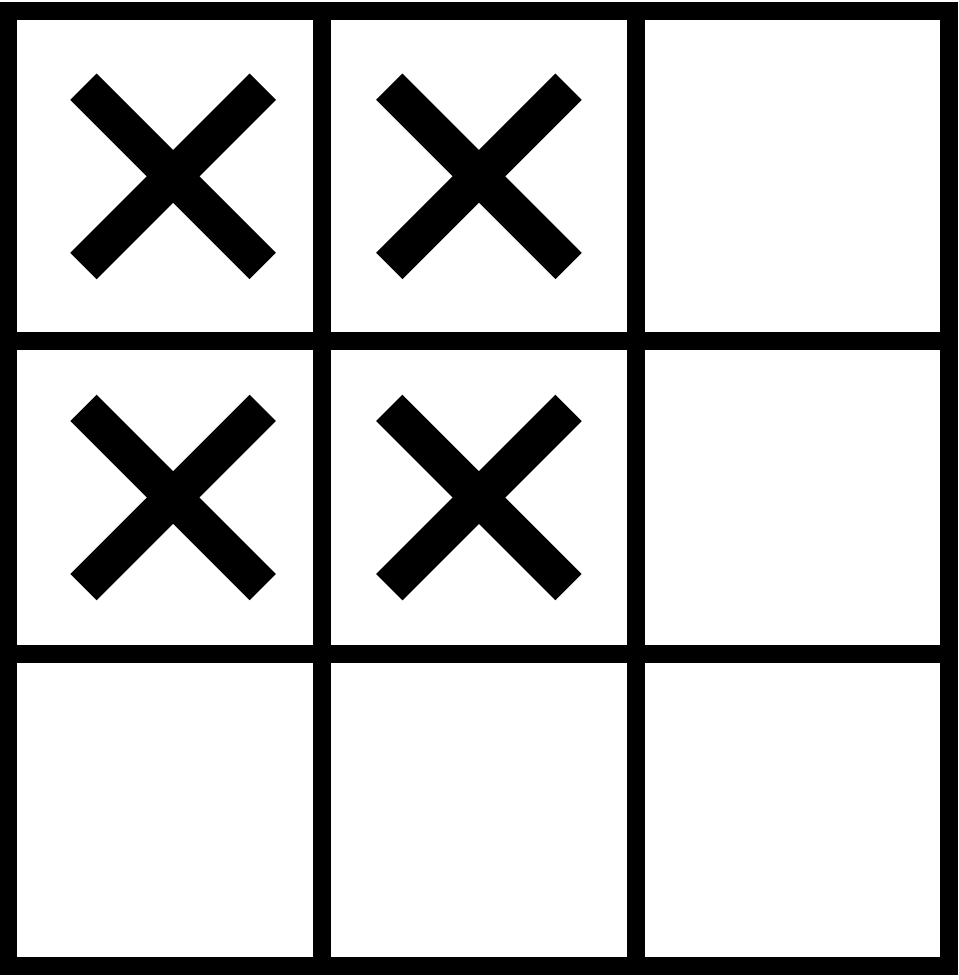}}}+\vcenter{\hbox{\includegraphics[scale=\boardscale]{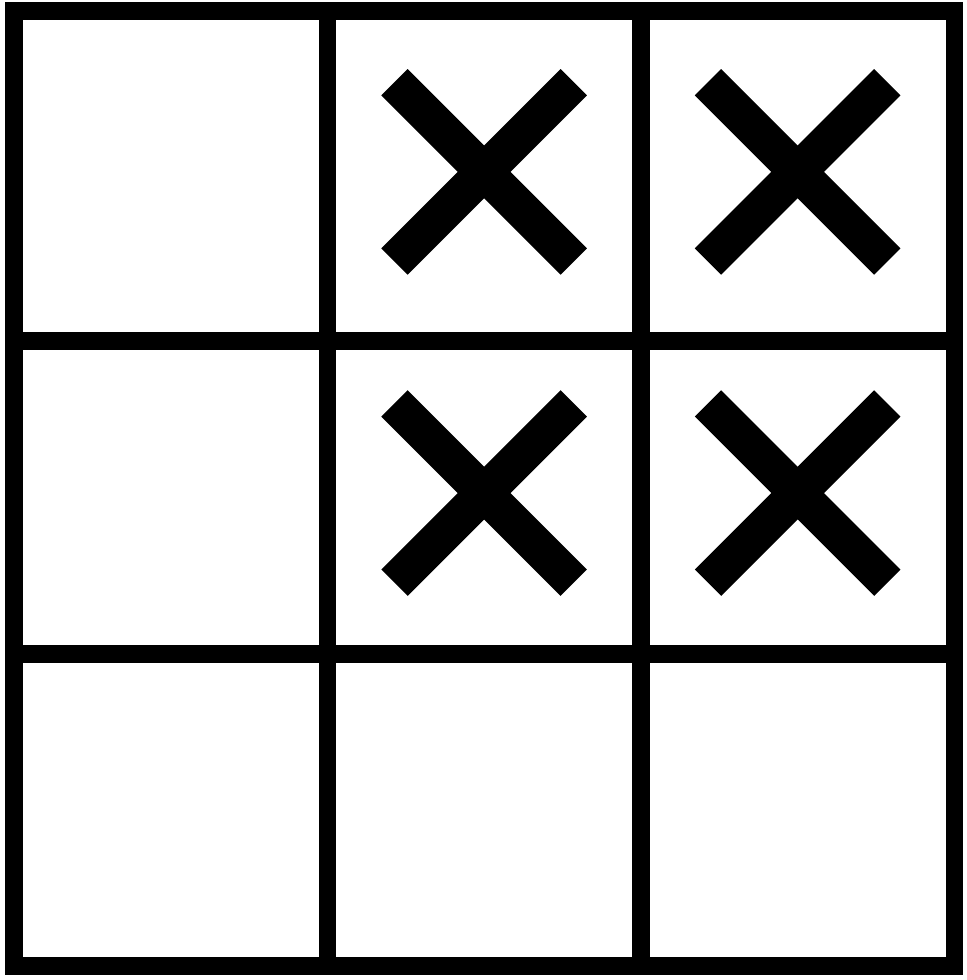}}}=\vcenter{\hbox{\includegraphics[scale=\boardscale]{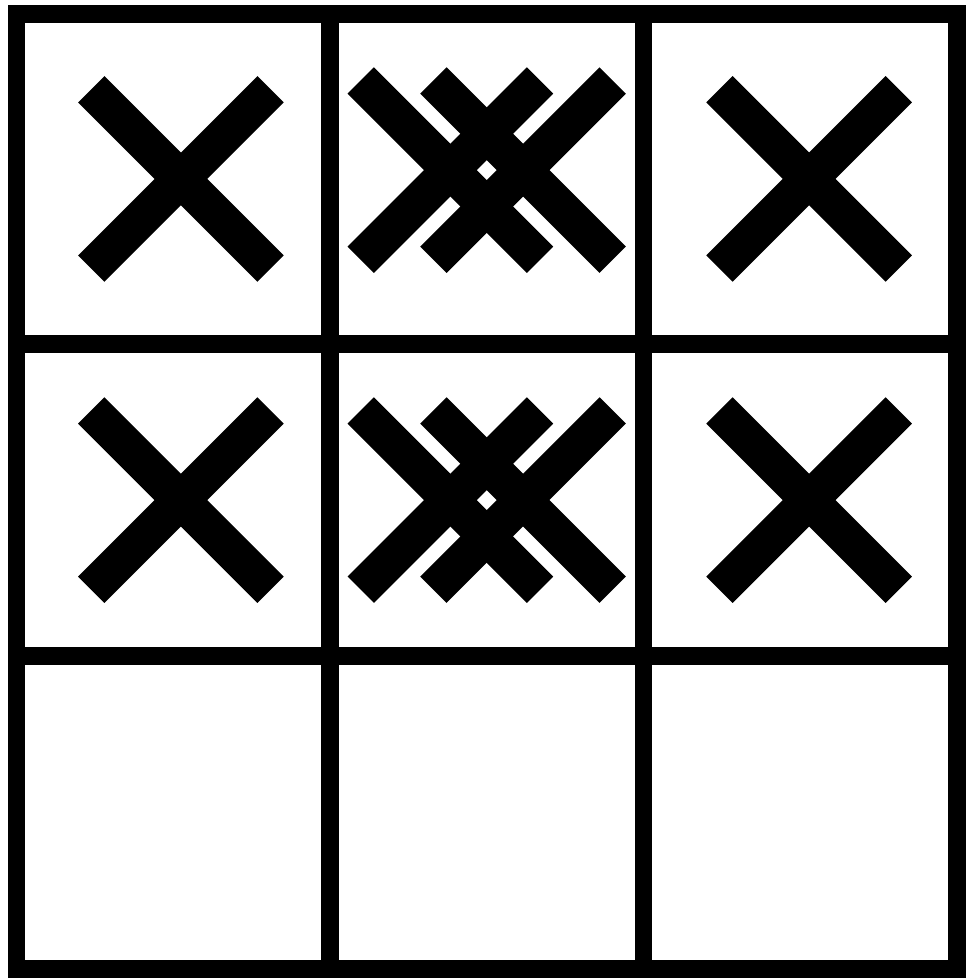}}}.
\end{align}

A cell on the board of a setup is called \emph{complete} if every right piece in that cell can be paired with a unique left piece, and vice versa. A setup is called complete if every cell on its board is complete. Each setup in the summation in Equation (\ref{eq:addition}) is complete, as is the result of the summation.

A setup is called a \emph{singleton} setup if it has a board with exactly four pieces, and those four pieces are arranged in a rectangular formation in the clockwise order right, left, right, left (lasso), or left, right, left, right (cross), starting from the top left piece.

The setups
\begin{align*}
\vcenter{\hbox{\includegraphics[scale=\boardscale]{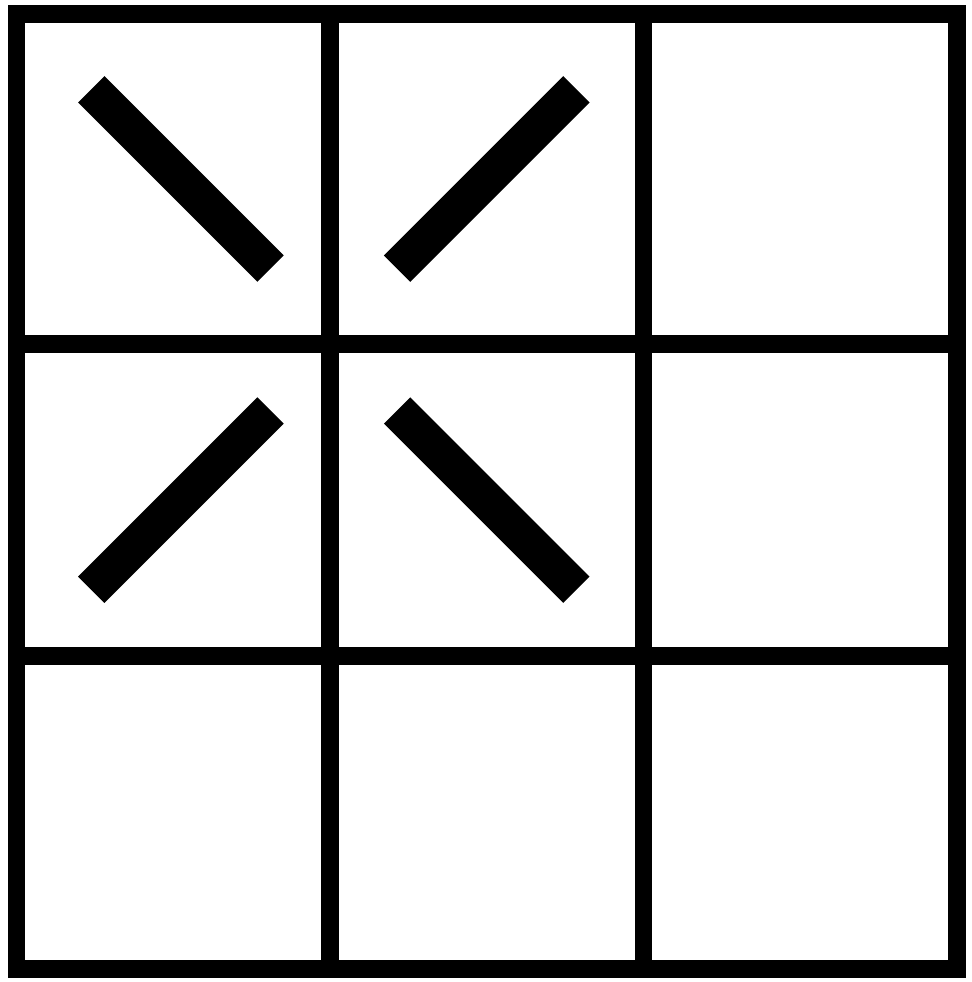}}},~\vcenter{\hbox{\includegraphics[scale=\boardscale]{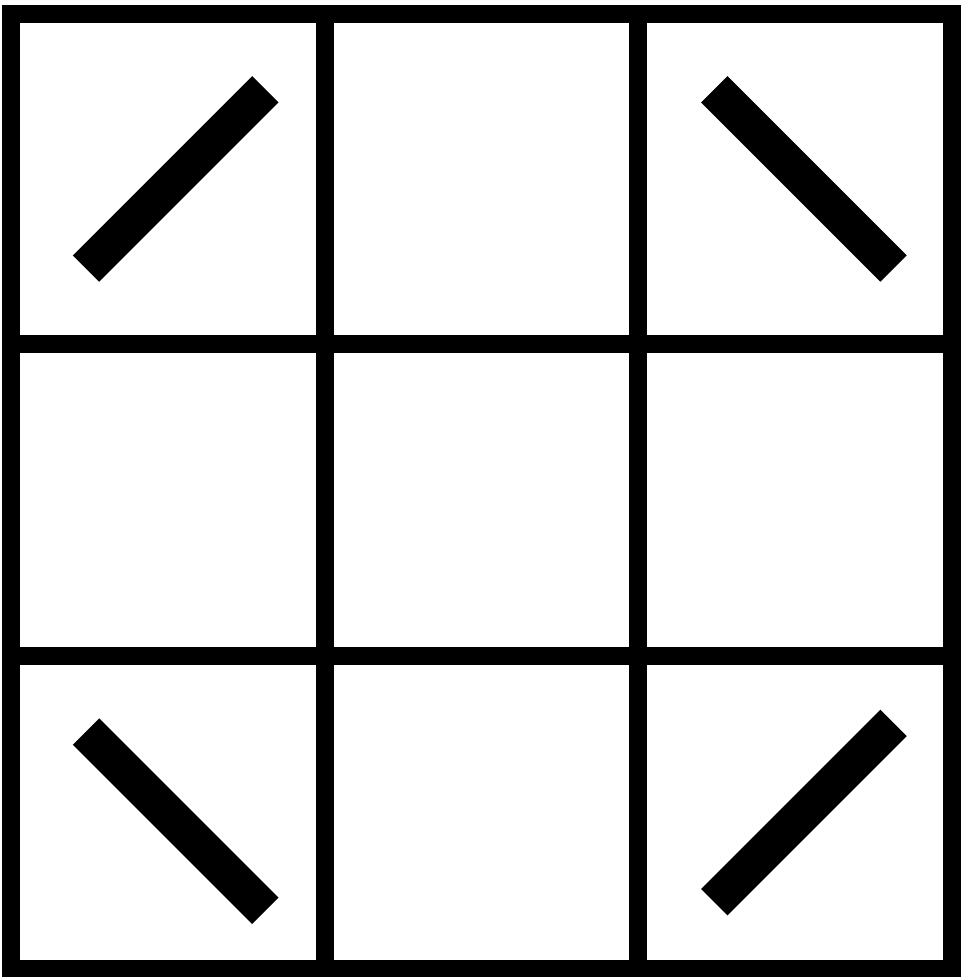}}}\text{, and }\vcenter{\hbox{\includegraphics[scale=\boardscale]{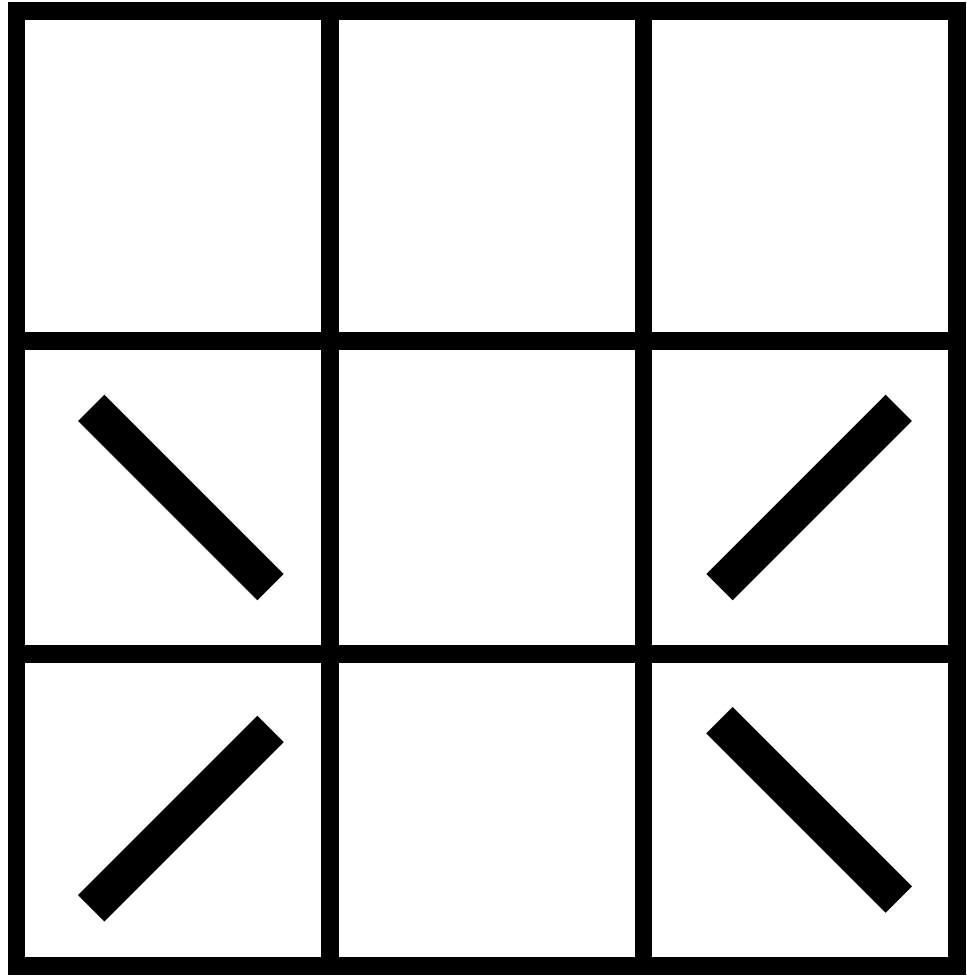}}}
\end{align*}
are singleton setups. We can now formulate another decision problem.

\begin{problem}
\textsc{SingletonSubsetSum}\\
\textit{Input:} \emph{A set of singleton setups, $S$.}\\
\textit{Question:} \emph{Does there exist subset $R\subseteq S$ such that the setup \begin{align*}\sum_{t\in R}t\end{align*} is complete?}
\end{problem}

With every setup $S$ on an $m\times n$ board, we know from Section \ref{section:separabilityin3x3} that we can associate an $m\times n$ matrix $M(S)$ with entries
\begin{align*}
[M(S)]_{ij}:=\begin{cases}
+1&\text{ if } S \text{ has a left piece in its } ij^{\text{th}} \text{ square;}\\
-1&\text{ if } S \text{ has a right piece in its } ij^{\text{th}} \text{ square;}\\
0&\text{ otherwise.}
\end{cases}
\end{align*}
For example, the matrix forms of the singleton setup examples from earlier are as follows,
\begin{align*}
M\left(~\vcenter{\hbox{\includegraphics[scale=\boardscale]{singletonsetup1-cropped}}}~\right)=\begin{pmatrix}
+1&-1&0\\
-1&+1&0\\
0&0&0\\
\end{pmatrix},\\
M\left(~\vcenter{\hbox{\includegraphics[scale=\boardscale]{singletonsetup2-cropped}}}~\right)=\begin{pmatrix}
-1&0&+1\\
0&0&0\\
+1&0&-1\\
\end{pmatrix},\\
M\left(~\vcenter{\hbox{\includegraphics[scale=\boardscale]{singletonsetup3-cropped}}}~\right)=\begin{pmatrix}
0&0&0\\
+1&0&-1\\
-1&0&+1\\
\end{pmatrix}.
\end{align*}
The following is equivalent to Lemma \ref{lemma:zerocontribution}. We denote by $M_0$ the matrix with all entries equal to $0$.
\begin{lemma}
	Let $S$ be a set of setups. Then their sum
	\begin{align*}
	\sum_{t\in S}t
	\end{align*} 
	is complete if and only if
	\begin{align*}
	\sum_{t\in S} M(t)=M_0.
	\end{align*} 
\end{lemma}
Hence, \textsc{SingletonSubsetSum} is equivalent to the following decision problem, where we say a matrix is a \emph{singleton matrix} if it is the matrix form of a singleton setup.
\begin{problem}
\textsc{SingletonMatrixSubsetSum}\\
\textit{Input:} \emph{A set of singleton matrices $S$.}\\
\textit{Question:} \emph{Does there exist a set of singleton matrices $R\subseteq S$ such that their sum}
\begin{align*}
\sum_{t\in R}t=M_0.
\end{align*}
\end{problem}

This problem is not the same as \textsc{SubsetSum} extended to matrices (which is of course trivially NP-complete), because of the restriction that the matrices must have only four non-zero elements (two equal to $+1$, two equal to $-1$) which must be arranged in either a lasso rectangular formation, or a cross rectangular formation. This ``lasso or cross'' restriction means that the problem is not even equivalent to \textsc{SubsetSum} over a finite field, it is something else entirely.

From what we have seen so far about grid-labelled graphs and edge contribution tables, it is clear that the following problem is equivalent to \textsc{SingletonSubsetSum}. 

\begin{problem}
\textsc{SubgraphDC}\\
\textit{Input: } \emph{A grid-labelled graph $G_l^{m,n}$.}\\
\textit{Question: } \emph{Does $G_l^{m,n}$ have a subgraph that satisfies the degree criterion?}\\
\end{problem}
Finally, let us note that the following problem is at least as hard as \textsc{SubgraphDC}.
\begin{problem}
	\textsc{SubsetPPT}\\
	\textit{Input: } \emph{A finite set of pure quantum states} $S=\{|\psi\rangle\in\mathcal{H}_A\otimes \mathcal{B}\cong\mathbb{C}^m\otimes \mathbb{C}^n\}$.\\
	\textit{Question: } \emph{Does there exist a set of states $R\subseteq S$ such that for the state}
	$
	\rho_R:=\frac{1}{|R|}\sum_{|\psi\rangle\in R}|\psi\rangle\langle\psi|,
	$
	\begin{align*}
	\rho_R^{\Gamma_A}\ge 0?
	\end{align*}
\end{problem}
Again, we leave the complexity classification of these decision problems open.
\section{Proof of Lemma \ref{lemma:69edges}}
\label{appendix:proof}
Let us now prove Lemma  \ref{lemma:69edges}.

\begin{proof}
	Let $G_l^{3,3}$ be a grid-labelled graph. We wish to prove that if $G_l^{3,3}$ has $6\le m\le 9$ diagonal edges, then it has a decomposition into building-block graphs. 
	
	That this is true is most easily seen by direct inspection. Through exhaustive computer search, all graphs with a set number of diagonal edges can be enumerated. There are a large number of graphs to be checked, especially in the $9$ edge case. To make the lists smaller and easier to parse, we remove subgraphs locally isomorphic to the cross graph. Then, we remove all duplicate graphs from the list. In Figures \ref{fig:6edgesall}, \ref{fig:7edgesall}, \ref{fig:8edgesall} and \ref{fig:9edgesall} we show all $3\times 3$ graphs that satisfy the degree criterion with $6,7,8$ and $9$ diagonal edges respectively. It can be verified by examining these figures that each of the grid-labelled graphs have decompositions into grid-labelled graphs locally isomorphic to building-blocks or rotations of building-blocks.
\end{proof}
\newpage

\begin{figure}[ht!]
	\centering
	\includegraphics[width=0.5\textwidth]{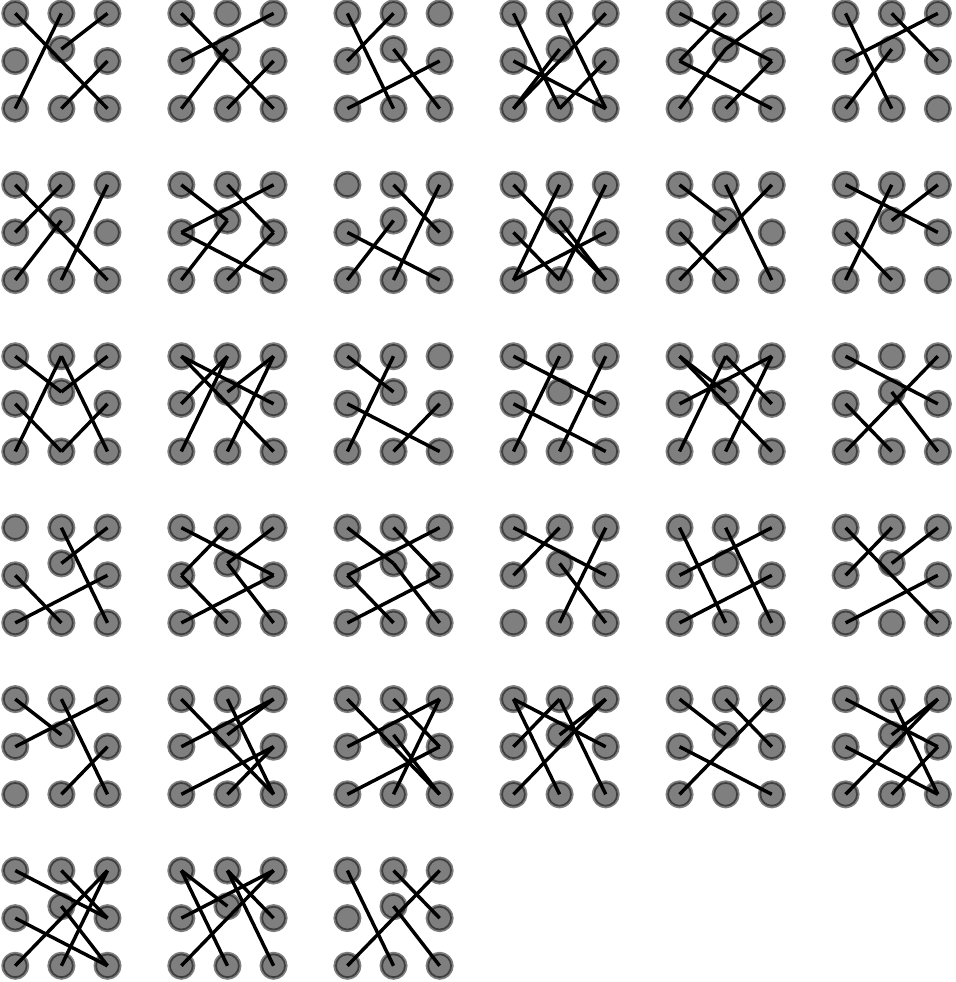}
	\caption{All $6$ edge grid-labelled graphs that satisfy the degree criterion up to local isomorphism, with cross subgraphs removed.}
	\label{fig:6edgesall}
\end{figure}
\begin{figure}[ht!]
	\centering
	\includegraphics[width=0.5\textwidth]{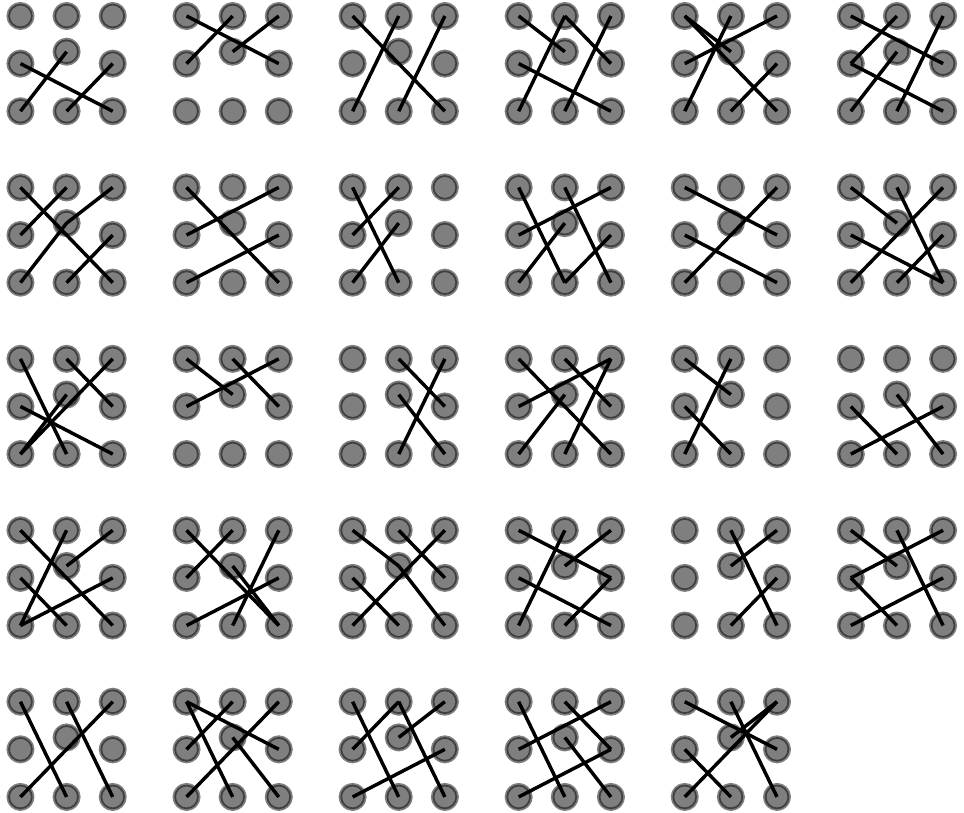}
	\caption{All $7$ edge grid-labelled graphs that satisfy the degree criterion up to local isomorphism, with cross subgraphs removed.}
	\label{fig:7edgesall}
\end{figure}
\begin{figure}[ht!]
	\centering
	\includegraphics[width=0.5\textwidth]{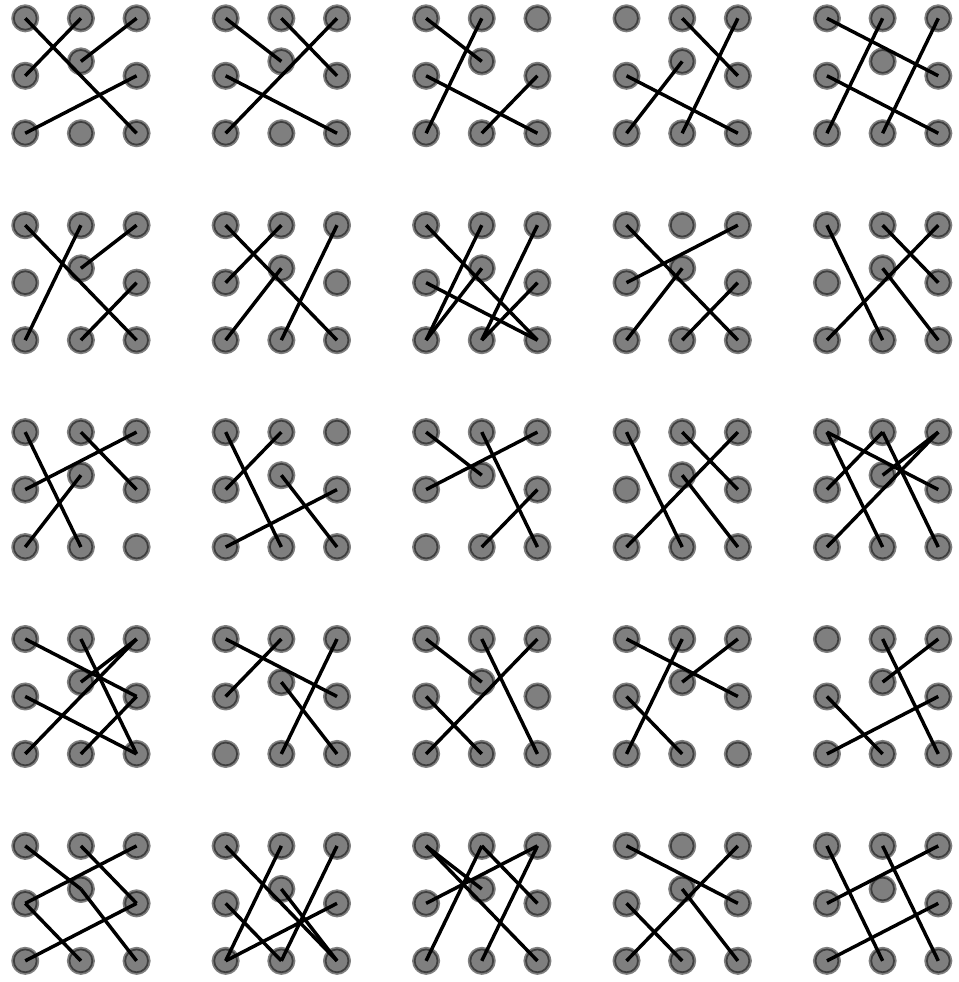}
	\caption{All $8$ edge grid-labelled graphs that satisfy the degree criterion up to local isomorphism, with cross subgraphs removed.}
	\label{fig:8edgesall}
\end{figure}
\begin{figure}[ht!]
	\centering
	\includegraphics[width=0.5\textwidth]{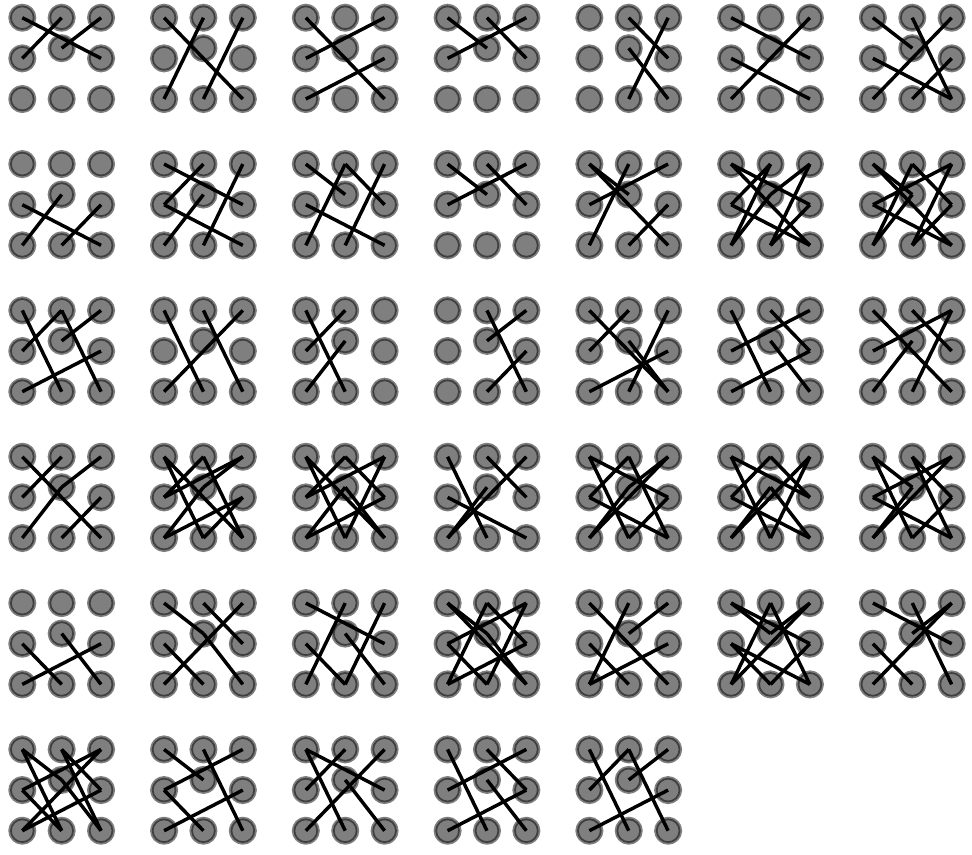}
	\caption{All $9$ edge grid-labelled graphs that satisfy the degree criterion up to local isomorphism, with cross subgraphs removed.}
	\label{fig:9edgesall}
\end{figure}

\end{document}